%% file: main.tex
\documentclass[11pt]{article}

\usepackage[T1]{fontenc}
\usepackage{palatino}
\usepackage{amsmath}
\usepackage{graphicx}
\usepackage[colorlinks=true, allcolors=blue]{hyperref}
\input{commands}

\usepackage{ wasysym }

\title{On the Cryptographic Foundations of\\ Interactive Quantum Advantage}

\author{Kabir Tomer\thanks{UIUC. Email:~\texttt{ktomer2@illinois.edu}}
\and
Mark Zhandry\thanks{Stanford University and NTT Research. Email:~\texttt{mzhandry@stanford.edu}}
}

\date{}
\begin{document}

\maketitle

\thispagestyle{empty}
 \input{abstract}
\newpage

\thispagestyle{empty}
\tableofcontents
\newpage

\pagenumbering{arabic}

\input{intro}

\input{preliminaries}

\input{cloning-hardness}

\input{generic}
\input{separation}
\input{comp-oracle}
\input{classical-black-box}

\section*{Acknowledgments}
KT was supported in part by AFOSR, NSF 2112890, NSF CNS-2247727 and a Google
Research Scholar award. 
This material is based upon work supported by the Air Force Office of
Scientific Research under award number FA9550-23-1-0543.
 Part of this work was done while visiting the Simons Institute for the Theory of
Computing.
\bibliographystyle{alpha}
\addcontentsline{toc}{section}{References}
\bibliography{abbrev0,crypto,custom,bib,main}
\end{document}

%% file: commands.tex
\newcommand{\pred}[1]{\ensuremath{\mathsf{#1}}}

\usepackage{amsfonts,latexsym,amsthm,graphicx,amsmath,enumitem,mathrsfs}
\usepackage{cancel,color,soul}
\usepackage{dsfont}

\usepackage{xcolor}
\definecolor{almostblack}{rgb}{0.6, 0.0, 0.0}

\definecolor{almostblackk}{rgb}{0.3, 0, 0.0}

\usepackage{hyperref,cleveref}
\hypersetup{
    colorlinks=true,
    citecolor=almostblack,
    linkcolor=almostblack,
    filecolor=almostblack,      
    urlcolor=almostblack,
}

\mathchardef\mhyphen="2D
\usepackage{mathtools}

\newtheorem{theorem}{Theorem}[section]
\newtheorem{claim}{Claim}[section]
\newtheorem{subclaim}{SubClaim}[section]

\newtheorem{definition}{Definition}[section]
\newtheorem{corollary}{Corollary}[section]
\newtheorem{lemma}{Lemma}[section]
\newtheorem{remark}{Remark}[section]

\newtheorem{conjecture}{Conjecture}

\newlength{\protowidth}

\definecolor{blue(ncs)}{rgb}{0.0, 0.53, 0.74}

\usepackage{amssymb}
\usepackage{amsmath}
\usepackage{amsthm}
\usepackage{amsfonts}
\usepackage{xspace}
\usepackage{bm}
\usepackage{bbm}
\usepackage{boxedminipage}
\usepackage{xparse}
\usepackage{float}
\usepackage{multirow}
\usepackage{graphicx}
\usepackage{caption}
\usepackage{subcaption}

\usepackage{tikz}
\usetikzlibrary{positioning}
\usetikzlibrary{fit}
\usetikzlibrary{calc}
\usetikzlibrary{backgrounds}
\usetikzlibrary{shapes}
\usetikzlibrary{patterns}
\usetikzlibrary{matrix}

\usepackage{enumitem}

\usepackage{geometry}

\newcommand{\sdotfill}{\textcolor[rgb]{0.2,0.2,0.2}{\dotfill}} 

\newcommand{\namedref}[2]{\hyperref[#2]{#1~\ref*{#2}}\xspace}

\def\poly{\ensuremath{\mathsf{poly}}\xspace}
\def\negl{\ensuremath{\mathsf{negl}}\xspace}

  \newcommand{\cA}{\ensuremath{{\mathcal A}}\xspace}
  \newcommand{\cB}{\ensuremath{{\mathcal B}}\xspace}
  \newcommand{\cC}{\ensuremath{{\mathcal C}}\xspace}
  \newcommand{\cD}{\ensuremath{{\mathcal D}}\xspace}

  \newcommand{\cG}{\ensuremath{{\mathcal G}}\xspace}
  \newcommand{\cH}{\ensuremath{{\mathcal H}}\xspace}

  \newcommand{\cM}{\ensuremath{{\mathcal M}}\xspace}
  
  \newcommand{\cO}{\ensuremath{{\mathcal O}}\xspace}
  \newcommand{\cP}{\ensuremath{{\mathcal P}}\xspace}

  \newcommand{\cV}{\ensuremath{{\mathcal V}}\xspace}

  \newcommand{\bbC}{\ensuremath{{\mathbb C}}\xspace}
  
  \newcommand{\bbE}{\ensuremath{{\mathbb E}}\xspace}

  \newcommand{\bbI}{\ensuremath{{\mathbb I}}\xspace}

  \newcommand{\bbN}{\ensuremath{{\mathbb N}}\xspace}

  \newcommand{\bbR}{\ensuremath{{\mathbb R}}\xspace}
  \newcommand{\bbS}{\ensuremath{{\mathbb S}}\xspace}

        \newcommand{\bbX}{\ensuremath{{\mathbb X}}\xspace}

  \renewcommand{\Pr}[0]{\mathrm{Pr}\xspace}

\newcommand{\td}{\mathsf{TD}}

\newcommand{\kabir}[1]{}
 \newcommand{\markz}[1]{}

\newcommand{\bin}{\{0,1\}}

\definecolor{lg}{gray}{0.89}

\DeclareMathOperator*{\Prr}{Pr}
\DeclareMathOperator*{\Exp}{\bbE}

\newcommand{\srno}{\ensuremath{\mathsf{s}}\xspace}
\newcommand{\sig}{\ensuremath{\mathsf{sig}}\xspace}

\newcommand{\Gen}{\ensuremath{\mathsf{Samp}}\xspace}
\newcommand{\Setup}{\ensuremath{\mathsf{Setup}}\xspace}

\newcommand{\Ver}{\ensuremath{\pred{Ver}}\xspace}
\newcommand{\TD}{\ensuremath{\pred{TD}}\xspace}
\newcommand{\SD}{\ensuremath{\pred{SD}}\xspace}

\newcommand{\Samp}{\mathsf{Samp}}

\newcommand{\ket}[1]{|#1\rangle}
\newcommand{\bra}[1]{\langle #1 |}
\newcommand{\ketbra}[1]{\ket{#1}\bra{#1}}

\newcommand{\reg}[1]{\mathbf{#1}}    
\newcommand{\adv}{\mathcal{A}}
\newcommand{\advB}{\mathcal{B}}
\newcommand{\obf}{\mathsf{Obf}}      
\newcommand{\chk}{\mathsf{Check}}
\newcommand{\ind}{\mathds{1}}
\newcommand{\eval}{\mathsf{Eval}}
\newcommand{\obfC}{\widetilde{C}}
\newcommand{\thres}{\mathsf{thres}}
\newcommand{\puzz}{\mathsf{puzz}}
\newcommand{\key}{\mathsf{key}}
\newcommand{\Sim}{\mathsf{Sim}}
\newcommand{\Sign}{\mathsf{Sign}}
\newcommand{\Comp}{\mathsf{Comp}}
\newcommand{\StdO}{\mathsf{StdO}}
\newcommand{\AdvO}{\mathsf{AdvO}}
\newcommand{\pp}{\mathsf{pp}}
\newcommand{\GComp}{\mathsf{GComp}}

%% file: abstract.tex
\begin{abstract}In this work, we study the hardness required to achieve proofs of quantumness (PoQ), which in turn capture (potentially interactive) quantum advantage. A ``trivial'' PoQ is to simply assume an average-case hard problem for classical computers that is easy for quantum computers. However, there is much interest in ``non-trivial'' PoQ that actually rely on quantum hardness assumptions, as these are often a starting point for more sophisticated protocols such as classical verification of quantum computation (CVQC). We show several lower-bounds for the hardness required to achieve non-trivial PoQ, specifically showing that they likely require cryptographic hardness, with different types of cryptographic hardness being required for different variations of non-trivial PoQ. In particular, our results help explain the challenges in using lattices to build publicly verifiable PoQ and its various extensions such as CVQC.
\end{abstract}

%% file: intro.tex
\section{Introduction}
\kabir{the MSY paper is called "Cryptographic Characterization of Quantum Advantage" which may be a bit confusing}
Arguably the most important question in quantum computing is whether there are tasks that quantum computers can solve super-polynomially faster than classical computers -- whether there is a super-polynomial quantum computational advantage. Interpreted broadly to include potentially interactive tasks, quantum advantage is captured by the notion of a \emph{proof of quantumness} (PoQ).

In a PoQ, a supposedly quantum prover convinces a classical verifier that it is indeed quantum, through a potentially interactive protocol. PoQ are a central object in the study of quantum computation and quantum protocols, and are closely connected to many topics. Quantum algorithmic developments such as Shor's quantum famous algorithms for factoring and discrete logarithms~\cite{FOCS:Shor94} or Hallgren's algorithms for solving Pell's equation and the PID problem~\cite{STOC:Hallgren02} yield (non-interactive) PoQs. In the interactive setting, many quantum protocols using classical communication trivially imply PoQ, examples including certified randomness~\cite{FOCS:BCMVV18}, classical verification of quantum computation~\cite{FOCS:Mahadev18a}, remote state preparation~\cite{FOCS:GheVid19}, and more. 
\kabir{one thought: There is a lot more work on PoQ as an object of study in and of themselves, maybe we should mention some? }

PoQs always require some form of computational hardness, in particular the hardness of simulating the honest quantum prover on a classical device. Achieving such hardness unconditionally is out of the reach of current techniques. As such, a natural and important goal is to understand the types of hardness assumptions needed for PoQs. Given the central nature of PoQs, understanding the hardness needed for PoQs can also help address questions surrounding related concepts.

A very recent work gave an answer to this question: Morimae, Shirakawa, and Yamakawa~\cite{STOC:MorShiYam25} show that PoQs with potentially inefficient verifiers (IV-PoQs) utilizing an arbitrary number of rounds are \emph{equivalent} to a notion called \emph{classically-secure one-way puzzles}. A (classically-secure) one-way puzzle is an average-case classically-hard problem that can be sampled by a quantum computer, with solutions confirmed by a potentially inefficient verifier. This shows in particular that PoQs can be based on very mild assumptions. 

On the other hand, this result proves a rather course-grained view of PoQ, as arbitrary-round IV-PoQs are the most general form of PoQ. In particular, there are a number of fascinating questions surrounding PoQ and related concepts which this result says little about:
\begin{itemize}
    \item What if we want \emph{efficient verification}, so that the quantum advantage can be observed ``in the real world?''
    \item Some flavors of PoQ seem to require \emph{cryptographic} assumptions; is this necessary? 
    \item Why is it so hard to obtaining public verification in some forms of PoQ or related protocols such as classical verification of quantum computation?
    \item Round-efficient PoQ seems to require structured algebraic assumptions such as LWE; is this necessary?
\end{itemize} 

\subsection{This Work}

In this work, we take a more fine-grained view of PoQ, and show several results lower-bounding the hardness required for different flavors of PoQs. This more fine-grained view also lets us provide meaningful answers to the questions above.

We first make a distinction between ``trivial'' PoQ and ``non-trivial'' PoQ. A ``trivial'' PoQ has the verifier generate an instance of a classically-hard but quantumly-easy computational problem, which the prover then must solve. By classical hardness, the ability to solve the problem proves quantumness. The most famous examples of such hard problems yielding trivial PoQs are Factoring and Discrete Logarithms~\cite{FOCS:Shor94}. But many other PoQs utilize several rounds of interaction that have the prover solve some interactive task~\cite{FOCS:BCMVV18}. Unlike trivial PoQs, non-trivial PoQs typically rely on assumed hardness even for \emph{quantum} computers. These types of PoQs are inherent in many related quantum protocols such as certified randomness and classical verification of quantum computation. These non-trivial PoQ tend to utilize \emph{cryptographic} assumptions, in the sense of requiring assumptions arising in cryptography.

We now summarize our results.

\paragraph{Constant-round efficiently-verifiable PoQs and cryptographic assumptions.} Our first results, morally speaking, show that the dichotomy between ``trivial'' PoQ and ``cryptographic'' PoQ is somewhat inherent: any constant-round efficiently-verifiable PoQs either (1) can be simplified to a trivial PoQ; or requires cryptographic assumptions.

Here, a cryptographic assumption posits the existence of a certain cryptographic primitive, such as one-way functions (functions that are efficiently computable but hard to invert). Such cryptographic assumptions imply average-case hard problems, but are typically considered to be much stronger than average-case hardness. The concrete cryptographic primitives we obtain are (with some minor caveats) either one-way functions or a weak form of publicly verifiable quantum money.

\paragraph{Black-box lower-bounds for low-round protocols.} The above results indicate that cryptographic tools are inherent in many PoQs. The vast majority of cryptographic techniques and constructions, and also the most efficient ones, are \emph{black box}, in the sense that the constructions only use the underlying tools as a black box, and do not utilize any particulars of the implementation of that tool. By restricting our attention to black-box constructions, we can obtain even stronger lower-bounds. We obtain 3 different such lower-bounds, all of which apply to even inefficiently-verifiable PoQs (IV-PoQs):
\begin{itemize}
    \item There is no black-box construction of constant round IV-PoQ from (even subexponentially secure, post-quantum) indistinguishability obfuscation (iO) plus one-way permutations. iO plus one-way permutations are often considered ``Crypto Complete'', together implying numerous foundational cryptographic concepts including public key encryption~\cite{STOC:SahWat14}, (adaptive) SNARGs~\cite{STOC:SahWat14, C:WatWu25}, watermarking~\cite{STOC:CHNVW16},
    (domain invariant) trapdoor permutations~\cite{TCC:BitPanWic16,C:ShmZha25}, public key quantum\allowdisplaybreaks money~\cite{EC:Zhandry19b} and much, much more. Our result implies that constant-round PoQ cannot be based in a black box way from \emph{any} of these primitives. Our results do \emph{not} rule out building PoQ from collision-resistance, which may help explain why all known constructions of constant round PoQ rely on assumptions which imply collision resistance. 
    
    \item 3-message IV-PoQ cannot be proved secure via a black box reduction to a quantum-secure (potentially inefficiently) falsifiable assumption. Note that this result even handles the case where the \emph{construction} is non-black box, as long as the reduction proving security is black box. Such results are known as \emph{meta reductions}~\cite{EC:BonVen98}. Note that the reduction can even be quantum and make quantum queries to the adversary. This improves on prior results showing the same for \emph{2-message} protocols~\cite{C:MorYam24} restricted to \emph{classical-query} reductions. Our result is tight as 4-message constructions are known assuming trapdoor claw-free functions~\cite{FOCS:BCMVV18}.
    
    \item Even ignoring round complexity, any IV-PoQ (even poly round) proved secure via (even quantum-query) black box reduction to a \emph{quantumly-secure} falsifiable assumption implies the existence of \emph{quantumly-secure} one-way puzzles. This improves on \cite{STOC:MorShiYam25} who only obtain a \emph{classically-secure} one-way puzzle; on the other hand our result comes at the cost of requiring a black-box reduction for the PoQ. Note that quantumly-secure one-way puzzles seem much stronger than classically-secure ones, as  the former implies a number of important primitives such as quantum commitments, quantum MPC, etc~\cite{ITCS:BraCanQia23}. Thus any PoQ with a classical black-box reduction to a hard quantum assumption implies a number of important quantum protocols. We find it interesting to obtain a security proof for a protocol (e.g. quantum commitments) which utilizes not just the existence of a primitive (e.g. PoQ), but also the fact that the primitive has a black-box reduction. 
\end{itemize}
\begin{remark}Any PoQ with an efficient verifier is also an IV-PoQ, so the above results also apply to efficiently-verifiable PoQ. 
\end{remark}

\paragraph{Public coin PoQ versus computational no-cloning.} We now turn to public coin PoQ, which means that the verifier keeps no private coins. Public coin protocols in general are important for a number of reasons: they allow for \emph{anyone} to verify the result of the protocol, and rounds can be compressed using the Fiat-Shamir heuristic~\cite{C:FiaSha86}. An example of public coin PoQ is a trivial PoQ where the classically-hard problem is in $NP$. On the other hand, all known constant-round PoQ based on hard lattice problems or other ``standard'' quantum assumptions (e.g.~\cite{FOCS:BCMVV18}) are not public coin. The limitation to secret coin protocols extends to the various related protocols based on such proofs of quantumness, such as classical verification of quantum computation~\cite{FOCS:Mahadev18a}, and succinct or non-interactive versions of it~\cite{C:BKLMMV22,TCC:ACGH20}. A significant open question has been whether these protocols could be upgraded to public coin. For example, upgrading said protocols to have public verification could be a path toward solving major open problems like constructing (publicly-verifiable) non-interactive zero knowledge proofs for QMA in the plain model (without oracles).

Our next result helps explain this situation. We show that a \emph{four-message} public coin PoQ admitting a \textit{classical} black box reduction to a quantum-hard assumption implies a weak variant of public key quantum money. Note that~\cite{FOCS:BCMVV18} is four-message. Our notion of quantum money is weak in the sense that the cloning task is to produce a large number of copies of a banknote instead of the usual two copies.\footnote{A harder problem for the adversary makes for a milder security requirement.} Nevertheless, our result shows that some form of publicly-verifiable cloning hardness is inherent to many standard techniques for constructing public coin non-trivial provably-secure PoQs. 

Our result therefore suggests a close relationship between provable non-trivial public coin PoQ and the challenging open question of constructing public key quantum money from standard assumptions, and in particular lattices. 

We also generalize our result to general constant-round PoQs, but with the caveat that we only prove the general-round version under an assumed quantum \emph{de Finetti theorem} for general LOCC (local quantum operation, classical communication) interactions with polylog dependence on dimension. While we do not know how to prove such a statement, our conjectured LOCC de Finetti theorem does not appear to be trivially false. We leave confirming or refuting our conjectured de Finetti theorem as an interesting direction for future work.

\section{Technical Overview and Further Discussion}

\subsection{Efficiently verifiable PoQs and cryptographic assumptions (Section~\ref{sec:generic})}

Our starting point is the result of~\cite{PassV20}, who study what they call ``interactive puzzles.'' An interactive puzzle is an interactive game between a classical efficient verifier and a prover, such that no efficient classical prover can cause the verifier to accept with non-negligible probability, but there exists a potentially inefficient classical honest prover that does cause the verifier to accept. Moreover, the verifier is required to be public coin. One of their main theorems is a round collapse theorem, which shows roughly that if one-way functions do \emph{not} exist, any constant-round interactive puzzle requiring a constant $\ell$ rounds can be converted into an interactive puzzle requiring $\ell-2$ rounds. By repeating this process $\ell/2$ times, they show that either (1) one-way functions exist, or (2) any constant-round interactive puzzle can be collapsed into just 2 rounds.

Interactive puzzles almost look like public coin PoQs, and a round collapse result would allow for turning a constant-round public coin PoQ into a trivial PoQ. However, there is one crucial difference: an interactive puzzle only requires the existence of an inefficient classical honest prover, whereas a PoQ requires an efficient \emph{quantum} prover. This breaks the straightforward adaptation of~\cite{PassV20} to public coin PoQs.

In slightly more detail, the idea of~\cite{PassV20} is to start from the original $\ell$-message prover $P$, and construct a new prover $P'$ that collapses the last three messages into 1. The final message of $P'$ will consist of:
\begin{itemize}
    \item The third-to-last message, $x_{\ell-2}$ from $P$.
    \item $t$ independent samples from the second-to-last message from the verifier, $x_{\ell-1}^{(1)},\cdots,x_{\ell-1}^{(t)}$, for some parameter $t$.
    \item The replies $x_{\ell}^{(1)},\cdots,x_{\ell}^{(t)}$ would make in the last message, for each of the second-to-last messages $x_{\ell-1}^{(1)},\cdots,x_{\ell-1}^{(t)}$.
\end{itemize}
For an efficient quantum prover $P$, the state of the prover after sending $x_{\ell-2}$ will in general be a quantum state $|\psi_{\ell-2}\rangle$, and generating the final message $x_\ell$ based on the verifier's message $x_{\ell-1}$ will in general involve some measurement on $|\psi_{\ell-2}\rangle$ that depends on $x_{\ell-1}$. But such a measurement will typically destroy the state $|\psi_{\ell-2}\rangle$, preventing the prover from answering more verifier messages. In order to answer each of the verifier messages $x_{\ell-1}^{(1)},\cdots,x_{\ell-1}^{(t)}$, it therefore seems that $P'$ will require many copies of the state $|\psi_{\ell-2}\rangle$, which it cannot construct efficiently. This is just one instance of the well-known hard problem of quantum rewinding. While $P'$ could generate these copies inefficiently, we need an efficient $P'$ in order to get a $\ell-2$-round PoQ. 

\paragraph{Our idea.} Our idea is the following: what if we also assumed that cloning quantum states was easy? Of course, the no-cloning theorem says that in general cloning quantum states is unconditionally information-theoretically hard. However, such cloning hardness only applies to states that are information-theoretically unknown and for which the cloner cannot efficiently verify the states. Here, however, we have a public verification mechanism: basically whether the state is able to convince the verifier. The verification mechanism also means that, while the states are computationally hidden, the quantum states are information-theoretically revealed. For these states, the usual no-cloning theorem does not apply, and it is a priori not obvious whether such states should be clonable or not.

We precisely define a notion of computational cloning hardness for publicly-verifiable states, and show that if it is impossible, then we can indeed finish the round collapse protocol. Conversely, if the  round collapse protocol fails, it implies our notion of cloning hardness.

\paragraph{What cloning hardness we need.} The typical notion of publicly verifiable cloning hardness in the literature is public key quantum money. Here, we focus on ``mini-schemes'', which are equivalent to full money schemes if one-way functions exist~\cite{AarChr} and capture the cloning hardness inherent to quantum money. In a mini-scheme, the adversary is given one copy of a publicly verifiable state, and must produce two copies that also pass verification. A mini-scheme is secure if no such efficient adversary can accomplish this.

Unfortunately, we are unable to get a mini-scheme in this sense. The reason is that we need not one but many copies of the state in order to do the round collapse, and in general one-to-two unclonability does not appear to follow from one-to-many unclonablity\footnote{Note that copying many times is a harder problem for the adversary, which makes the security weaker.}. One may try to run the cloner again on the copies it produced. However, the cloner only needs to make copies that pass verification; they may be very different from the original state. While the cloner succeeded on the original state, there is no guarantee that it will succeed again on these clones. We therefore define a weak mini-scheme as one that has 1-to-many hardness\footnote{On the other hand, we actually obtain something stronger than quantum money. Namely, to perform round collapse, we actually only need to be able to generate the prover's final response to many choices of the verifier's message. These responses can be viewed as a ``signature'' on the verifier's message. This gives us a mini-scheme with a signing functionality, which is known as a signature token~\cite{BS16}. Security now requires that it is hard to sign multiple messages; one could sign multiple messages by cloning, but there may be other strategies as well. As with quantum money, signature tokens usually are required to have 1-to-2 security (it is possible to sign once, but impossible to sign twice), whereas we obtain a weak form of signature token where it is only hard to sign a polynomial number of messages. In the 4-round case, we actually get something even stronger, a weak version of \emph{one-shot signatures}~\cite{STOC:AGKZ20}, which roughly are signature tokens that remain unclonable even if the original signature token was created by the adversary itself.}. See Section~\ref{sec:cloninghardness} for a more formal definition of the kinds of cloning hardness we need.

\medskip

\paragraph{The final result.} The final result is that if constant-round public coin PoQ exist, then either:
\begin{itemize}
    \item It can be collapsed into a trivial PoQ; or
    \item One-way functions exist; or
    \item Weak quantum money mini-schemes exist.
\end{itemize}
The last two bullets are both cryptographic in nature. We also note that existing literature~\cite{private-to-public} implies that public coin protocols can be converted to efficiently verifiable private coin protocols, assuming the absence of a certain type of one-way function. Thus, we obtain that constant-round public coin PoQ can either be made trivial, or certain cryptographic primitives must exist. See Section~\ref{sec:generic} for the full proof.

One caveat to the above, which is inherited from~\cite{PassV20}, is that the one-way function and weak signature token only have security against uniform attackers, and only have security for infinitely-many security parameters. This is typical of ``win-win'' style results like this. In our case, this is because the round-reduction uses the one-way function and weak signature token adversaries, but these adversaries only are guaranteed to succeed for infinitely-many security parameters.

Also note that we do not know any relation between one-way functions and weak quantum money mini-schemes.

\subsection{Black-box Separation between constant-round IV-PoQ and iO (Section~\ref{sec:iOimposs})}

We next turn to black-box separations. Our first result is to show that even IV-PoQ in constant rounds cannot be constructed in a black-box way from (sub-exponentially-secure, post-quantum) indistinguishability obfuscation (iO) and one-way permutations. Recall that iO is a transformation on circuits that preserves functionality while trying to hide the underlying details of how the circuit is implemented. Concretely, iO insists that the obfuscations of equivalent programs are indistinguishable.

To give our black-box separation, we start by providing an oracle where one-way permutations and iO exist. This oracle follows the work of~\cite{FOCS:AshSeg15}. We first provide an oracle for a random permutation, $P$. This gives the one-way permutation. Next, we give an oracle $O$ that takes as input oracle-aided circuits $C^P$ and some random coins $r$, and outputs a string $\hat{C}$. We take $\hat{C}=O(C^P)$ to be the obfuscation of $C^P$. Finally we provide an oracle $E(\hat{C},x)=C^P(x)$ in order to evaluate the obfuscated programs. Note that $E$ makes queries to $P$. It is not hard to show that $O,E$ are indeed secure obfuscators against against adversary who make polynomially-many queries to $(P,O,E)$.
\begin{remark}It is important to allow $O$ to work on oracle-aided circuits, as this captures obfuscating circuits that make use of the one-way permutation $P$, which captures the typical techniques for using iO in cryptography.
\end{remark}

The next step is to show that constant-round IV-PoQ does not exist relative to these oracles. But in this model, constant-round IV-PoQ actually exists, unconditionally! In particular, this model implies a claw-free trapdoor function, which in turn implies 4-message PoQ~\cite{FOCS:BCMVV18}.\footnote{To get the claw-free trapdoor function, first use the one-way permutation $P$ to build a pseudorandom permutation $\Pi^P_{k}$ using standard techniques. Sample a random keys $k_0,k_1$, and then obfuscate the circuits $\Pi_{k_0},\Pi_{k_1}$ using the obfuscation oracle to get ``circuits'' $\hat{\Pi_0}$ and $\hat{\Pi_1}$. The resulting obfuscated circuits are then the claw-free function pair (claw-freeness being a straightforward proof), and the trapdoor is $k_0,k_1$.} We therefore need to provide an extra oracle which breaks any constant-round IV-PoQ, but importantly does not break the underlying one-way permutation or iO provided by the oracles $(P,O,E)$.

\paragraph{Our Idea.} Our idea is the following. We will supply a breaking oracle $B$ which takes as input partial transcripts $T$ of the interactive protocol and outputs a next message $x$ for the prover. Concretely, even though $B$ will be a classical oracle, it's goal is to simulate the honest quantum prover. For each partial transcript, we inefficiently determine the current quantum state of the prover conditioned on the execution so far matching the partial transcript $T$, and then sample $x$ as the honest prover would. We then hardcode the pair $(T,x)$ into the description of $B$; on input $T$, $B$ will output $x$ deterministically.

It's not hard to see that $B$ will break any PoQ, since it behaves identically to the honest quantum prover under a single straight-line interaction. Importantly, since $B$ is inefficiently constructing the honest quantum prover's state, this corresponds to it potentially needing full truth-table access to $(P,O,E)$. Such inefficient access to $(P,O,E)$ would vilate the security of the implied one-way permutation and iO.

For now, we focus on adversaries making classical queries to $(P,O,E)$, and we will discuss quantum queries later.

Going forward, it is helpful to think of $B$ as a way to outsource the honest quantum prover to a classical oracle. Instead of hardcoding the responses, we can think of $B$ as running the the honest quantum prover inside its head in order to respond to the adversary's queries. Hardcoding seemingly just turn the quantum prover into a classical oracle, and at first glance it may seem that the classical oracle is indistinguishable from the honest quantum prover. Since the honest quantum prover is efficient (and hence could be run by the adversary on its own), this would allow us to conclude that giving out $B$ does not impact the security of the one-way permutation or iO. Unfortunately, this intuition turns out to be false,  and classicalizing the honest quantum prover into $B$ actually does let the adversary do things it could not on its own. However, we will show that $B$ nevertheless cannot be used to break the one-way permutation or the iO.

One immediate issue is that an attacker for $(P,O,E)$ may query $B$ on partial transcripts that did not result from previous interactions with $B$. Such partial transcripts do \emph{not} look like a simple straight-line interaction with the honest quantum prover (since the prover may have never generated the messages in the transcript $T$), and so the attacker may learn ``too much'' from such a query. We can fix this problem by having $B$ sign all its messages, and only answer partial transcripts $T$ where previous prover messages came with a valid signature.

A deeper issue is that the adversary may query $B$ on partial transcripts $T_0,T_1$ that are identical up until the most recent verifier message. Such transcripts can be honestly generated, and have all the necessary signatures. However, analogous to our first result, the honest quantum prover generates its next message by measuring its current state, which destroys the state. As such, the honest quantum prover would only be able to respond to one of the transcripts $T_0,T_1$, but not both. But $B$ will respond to both partial transcripts, showing that it can do something the adversary is unable to simulate for itself. If we think of $B$ as outsourcing the computation of the honest quantum prover, essentially what $B$ is doing when queried on $T_0,T_1$ is that it is copying the state of the honest prover. With copies of the state, it can answer $T_0,T_1$ (or any other transcripts that differ in the last message) as needed. We need to show that the ability to clone these particular quantum states -- or more accurately, copying the probabilistic next-message functionality of the honest quantum prover -- does not does not break one-way permutations or iO.

\paragraph{A problem with collision resistance.} We note that copying quantum states, or even copying next-message functionalities, does break some cryptosystem, in particular collision-resistance. For example, we can consider a simplified scenario where the prover's state is the superposition
\[|\psi_y\rangle:=\sum_{x:H(x)=y}|x\rangle\]
where $H$ is some collision-resistant hash function, $y$ some image (derived from the partial transcript), and $|\psi_y\rangle$ is the uniform superposition of pre-images of $y$ (here, we are ignoring normalization to keep the notation simple). Suppose the next message in the honest prover is obtained by measuring $|\psi_y\rangle$. Cloning this state (or even the next-message functionality) immediately allows for finding collisions: if we measure two copies of $|\psi_y\rangle$, the results will be independent random pre-images $x_0,x_1$ of $y$, which with reasonable probability will be distinct (assuming $y$ has multiple pre-images). 

In fact, since $y$ and the first pre-image $x_0$ can be taken to be honestly generated, such an oracle can actually be seen as breaking second-pre-image resistance, which asks that it is hard to find a second pre-image of an honestly generated input. This may seem troubling, as a general second-pre-image resistance adversary can break any one-way function; this is a crucial ingredient in the proof that second-pre-image resistant hash functions can be constructed from any one-way function~\cite{STOC:NaoYun89,STOC:Rompel90}. But this would seem to contradict our claim that the one-way permutation $P$ remains secure!

\paragraph{The Solution.} An analogous issue arises in the study of statistically hiding commitments built from one-way functions/permutations. Namely,~\cite{HHRS15} prove a lower-bound for the round-complexity of such commitments, using a mechanism that also seems to violate second-pre-image resistance. However, their crucial insight is that the algorithm which inverts a one-way function by finding second pre-images is inherently \emph{adaptive}, requiring super-constant many rounds of adaptive queries to the second pre-image finder. However, an adversary for a statistically hiding commitment only allows for an adaptivity that grows with the number of rounds. Intuitively, this is because the adversary doesn't get a truly arbitrary second-pre-image finder, but only one that finds second pre-images at the current stage of the protocol. As such, for few rounds, say a constant\footnote{In their work, they actually extend to $o(n/\log n)$ rounds}, there won't be enough adaptivity to break the underlying one-way function/permutation.

A similar phenomenon also occurs with proofs of quantumness, meaning that the oracle $B$ for breaking PoQs does not actually allow for breaking one-way permuations if the PoQ has only a constant number of rounds. However, the underlying details of how we actually prove this are very different.

At a very high level, we first employ techniques from~\cite{pdqp}, which gave a general set of techniques for upgrading classical oracle separations to the quantum setting. In particular, they used their techniques to prove that one-way permutations plus iO do not imply collision resistance even under quantum queries, lifting the classical analogous result~\cite{FOCS:AshSeg15} to the quantum setting. These techniques are not enough for us, however. To prove that iO and one-way permutations to not imply collision-resistance, the oracles $(P,O,E)$ are augmented with Simon's oracle~\cite{EC:Simon98} for breaking collision-resistance. But our oracle is even stronger, since in particular it finds second pre-images. At a lower-level, one of the key requirements~\cite{pdqp} need of the breaking oracle is that the distribution of its outputs is insensitive to local changes of the oracles $(P,O,E)$. However, this turns out to not be true for our PoQ-breaking oracle $B$. 

To overcome these challenges, we develop a new quantum one-way-to-hiding lemma in the setting of oracles for iO and one-way permutations. This lemma, roughly, lets us show that if an adversary can detect that we punctured the oracles $(P,O,E)$ at some point (which may in turn influence the outputs of $B$), then the adversary must actually know the punctured point. This turns out to be highly non-trivial, and is most involved technical contribution of our work. See Section~\ref{sec:iOimposs} for details.

\subsection{Meta-reduction for 3-message IV-PoQ (Section~\ref{sec:bbreductions})} 

We now turn to constructions that are potentially non-black-box, but where the reduction proving security is black-box, and show barriers to achieving such a reductions. In particular, ``standard'' assumptions in cryptography are almost always falsifiable, in the sense that they take the form of a game between adversary and challenger, and the assumption is that no efficient adversary can win the game with probability noticeably higher than some threshold. Such games are falsifable in the sense that if the assumption is false, there is a procedure to demonstrate its falseness: namely, run the game with the adversary that purports to contradict security.

The vast majority of security proofs in cryptography take the form of a ``black box'' reduction to such a falsifiable assumption. In the context of a PoQ, a black box reduction is an efficient algorithm $R$ which makes queries to a supposed adversary $A$, and attempts to win some problem $P$ specified by an interactive game. If $A$ is an algorithm which breaks the PoQ $\Pi$ -- that is, if $A$ is classical and convinced the verifier of $\Pi$ that it is quantum -- then $R^A$ will break $P$ with too-high a probability.  If $A$ is an efficient adversary for the PoQ, then in particular $R^A$ will be an efficient algorithm for $P$, contradicting its presumed hardness and thereby justifying the impossibility of an efficient $A$. Importantly for black-box reductions, however, $R^A$ must break $P$ even if $A$ is \emph{inefficient}, even though this does not reach a contradiction since $R^A$ is no longer efficient. 

Our first such impossibility is for 3-message IV-PoQ, where we show that it cannot be proved secure via a black-box reduction to a quantum-secure hard problem $P$. The reduction is allowed to be quantum and can even make quantum queries to the adversary. To prove this result, we follow the typical abstract framework in the literature, starting from~\cite{EC:BonVen98}: the first step is to exhibit an classical but inefficient algorithm $A$ which breaks the PoQ $\Pi$. By the black-box reduction, this means that $R^A$ breaks $P$ even though it is inefficient. The second step is then to devise an efficient algorithm $M$ which runs $R$, but manages to efficiently simulate the behavior of $A$. The resulting $M$, often called a ``meta-reduction'', is then efficient but still solves $P$; but since $M$ is efficient this now would contradict the hardness of $P$. Thus the original reduction $R$ could not exist. 

In meta-reductions, $M$ needs to be able to do something $A$ cannot, otherwise $M$ by simulating $A$ efficiently is just directly implementing an efficient attack on the cryptosystem, which we would generally believe is not possible. In the case of proving a PoQ based on a \emph{quantum-secure} assumption, what $M$ can do is utilize a quantum computer, which $A$ cannot.

In the case of 2-message PoQ,~\cite{C:MorYam24} show a simple meta reduction. Basically, the inefficient $A$ is constructed as follows: for each first message $y$, run the honest quantum prover to get a response $z$. Now hardcode the pair $(y,z)$ into $A$, and have $A$ answer every first message deterministically with $z$. $A$ is inefficient since it needs exponentially hardcoded pairs $(y,z)$. However, it is easy to simulate $A$: $M$ just runs the quantum honest prover, which is indistinguishable from $A$. Their result only works for classical-query reductions.

We generalize this to 3-message PoQ, and to allow quantum reductions. We start by describing the generalization to 3-message PoQ. Let us call the first message from the adversary $x$, the first message from the verifier $y$, and the final adversary message $z$. The challenge with 3-message PoQ is that now $R$ may query $A$ on partial transcripts $(x,y)$ and $(x,y')$ with the same first adversary message $x$. Naively, the efficient $M$ will not be able to answer both these queries. This is because when $M$ creates $x$ (which it does by just running the honest prover), it may have resulted in a quantum state $|\psi_x\rangle$. When $M$ receivers $(x,y)$ (corresponding to the honest prover receiving $y$ from the verifier) $M$ then operates on and measures (and hence destroys) $|\psi_x\rangle$ to generate $z$. But now that $|\psi_x\rangle$ is destroyed, $M$ will not be able to answer the query $(x,y')$.

We get around this issue by observing that we can give $M$ the necessary states $|\psi_x\rangle$ \emph{as quantum advice}. The result is that $M$ is now a non-uniform efficient algorithm which breaks $P$, contradicting the non-uniform hardness of $P$.

The above discussion only applies to classical-query reductions. In the case of reductions making quantum queries to the adversary, it does not even make sense to talk about well-defined queries $(x,y)$, since in general the adversary may query on a superposition of $(x,y)$ transcripts. We generalize the above to work even with quantum queries by using a variation on the compressed oracle framework~\cite{C:Zhandry19}. Roughly, the original compressed oracle technique allows for replacing a random oracle with a superposition over ``databases'' of input/output pairs, which very roughly correspond to the input/output pairs of the oracle that have been ``observed'' by the adversary. This database plays an intuitively similar role to classical lazy sampling, which is widely used in classical query results. Analyzing the quantum database sometimes results in proofs that look similar to the classical-query counterparts, though this clearly cannot be universal as quantum queries can give advantage over classical~\cite{EC:YamZha21,FOCS:YamZha22}.

While the original compressed oracle was for random functions, it has been generalized to functions with non-uniform outputs~\cite{EPRINT:CMSZ19,EPRINT:HhaYun23}, potentially with different distributions for each output. Necessary ingredients for these works are that (1) the non-uniform distribution is known, and (2) the output distributions for each input are independent. In our setting, property (2) is satisfied, the responses for each $(x,y)$ are sampled independently according to some $y$-dependent measurement applied to $|\psi_x\rangle$. However, (1) is not known, since the measurement distribution for a quantum state cannot in general be efficiently determined. Our key idea is show that we can emulate the correct behavior of the compressed oracle using the copies of $|\psi_x\rangle$ we have, even though we do not know the correct output distribution. To do so we formulate a new variant of the compressed oracle technique which may be of independent interest, and while we incur a small inverse polynomial error in simulation, this does not affect the results.  By analyzing this variant, we are able to extend our result to handle quantum queries. See Section~\ref{sec:bbreductions} for details. 

\subsection{Meta-reduction for general-round IV-PoQ (Section~\ref{sec:bbreductions})}

We now consider protocols with four or more messages. Note that here, we do not expect to get an unqualified impossibility, since there are four-message protocols with black-box reductions, e.g.~\cite{FOCS:BCMVV18}. The reason four-message protocols can have black-box reductions is that now the state in the prover's first message can actually depend on the verifier's first message, and therefore cannot be provided as non-uniform advice. Indeed, protocols such as~\cite{FOCS:BCMVV18} crucially exploit this fact. In the security proof, for the first two messages, the reduction behaves as the verifier, sending a first message $w$ and receiving a message $x$ from the malicious classical prover. Then the reduction queries the prover on the partial transcripts $(w,x,y_0)$ and $(w,x,y_1)$ to get two final message responses $z_0,z_1$ which allow it to solve the problem $P$. However, a quantum algorithm (in particular, the honest prover) would be unable to generate both $z_0,z_1$ simultaneously (even though it could generate \emph{either} $z_0$ or $z_1$ from its internal state after sending $x$). As a result, the reduction does \emph{not} apply to the honest quantum prover, which allows the problem $P$ to remain plausibly quantum-secure.

Instead, we wish to better understand what sort of hardness is necessary to have such a reduction for multi-round protocols. Our key insight is that constructing the prover responses given the current partial transcript can be viewed as solving a certain one-way puzzle. If these one-way puzzles were quantumly easy, then we could get the meta reduction to go through, implying that $P$ is quantumly easy. As a consequence, we obtain that the quantum hardness of $P$ implies that certain one-way puzzles are quantumly secure. 

The above description allows us to obtain quantum-secure one-way puzzles from classical-query reductions; by using compressed oracle techniques similar to the first meta-reduction, we extend this to quantum-query reductions as well. See Section~\ref{sec:bbreductions} for details.

\paragraph{Discussion.} Our result says that in order for any IV-PoQ to be black-box reduce-able to a true quantum falsifiable assumption, it must yield a quantum-secure one-way puzzle. Note that such quantum-secure one-way puzzles are \emph{not} in general implied by quantum-hard falsifiable assumptions, so this is a non-trivial statement. For example, the falsifiable assumption could even utilize an inefficient challenger, which could sample hard instances that require even exponential time to sample; one-way puzzles in contrast must be efficiently sampleable.

Also note that the one-way puzzle only depends on the particular IV-PoQ protocol (not the reduction), so the obtained one-way puzzle is explicit. However, the security proof of the one-way puzzle is non-black box, in that it uses the assumed reduction algorithm. We find it interesting to use an assumed reduction algorithm to prove the security of another cryptosystem, and are not aware of any other security proofs of this sort. 

At first glance, upgrading the prior work of~\cite{C:MorYam24} from classical-secure one-way puzzles to quantum-secure puzzles (at the expense of assuming a black-box reduction to a quantum-secure falsifiable assumption) may seem minor. However, we note that classically-secure one-way puzzles are extremely mild assumptions. In particular, they do not seem useful for cryptography: because one-way puzzles are sampled by a quantum algorithm, when used to build cryptosystems, the resulting protocol will be quantum. However, a classically-secure one-way puzzle will result in a quantum cryptosystem that is only secure against classical adversaries, meaning the adversary is much weaker than the honest parties. Cryptography usually insists on security against adversaries that are even much \emph{stronger} than the honest parties.

On the other hand, quantum-secure one-way puzzles, while still mild, are known to imply numerous quantum cryptographic objects such as quantum commitments~\cite{STOC:KhuTom24}, and thus any of the many objects that are equivalent to commitments, such as quantum MPC.

\subsection{Public coin PoQ versus no-cloning (Section~\ref{sec:publiccoinreductions})} 

We now turn to our final set of results, which shows that provably-secure public coin non-trivial PoQs are likely to require some form of publicly verifiable cloning hardness.

Our proof follows the meta-reduction outline as in our previous two results; like in the last result, instead of obtaining a strict impossibility (which would contradict known results), we instead extract a cryptographic object, whose security follows from the mere existence of a reduction for the PoQ.

We start with the same high-level outline as in our previous two meta-reductions. We first design first a classical but inefficient attacker $A$ for the PoQ: for each partial transcript $T$, we can inefficiently sample a random choice for the honest provers next message $z$ and hardcode that into the description of $A$. $A$ will then always deterministically reply to the partial transcript $T$ with the message $z$. 

Next, we give an efficient (quantum) simulator $M$ for $R^A$, which then breaks the assumption $P$, contradicting its hardness and thereby proving that the reduction $R$ could not exist. We start with the naive quantum simulator which runs $R$, but replaces queries to $A$ with queries to the honest quantum prover. Since $A$ behaved like the honest prover anyway, we hope that $M$ to simply break $R$. The challenge, just as with the general round meta-reduction for IV-PoQ, is that the meta-reduction will be unable to answer queries on two partial transcripts that differ in, say, just a single message, as in order to do so, it would seem to require multiple copies of the internal state of the honest prover.

Here is where cloning hardness comes in. We show that the meta-reduction actually \emph{can} handle reductions making such queries if it could copy the state of the honest quantum prover exactly. Moreover, the states of the honest prover are publicly verifiable, since the PoQ is public coin. At first glance, this seems to give the cloning hardness we need: the states of the honest quantum prover must be hard to clone, else we can simulate $R^A$ using by an efficient quantum algorithm, which would then contradict the assumed hardness of $P$. 

Unfortunately, turning this intuition into a security proof is subtle and comes with some significant caveats. The reason is that publicly verifiable unclonability does not insist that the adversary clones the state exactly, but rather it insists that the adversary cannot come up with two states that pass verification. It is actually not hard to come up with insecure quantum money schemes, even from very mild assumptions such as one-way functions, where copying the state exactly is hard, but it is easy to come up with states that pass verification. If the quantum simulator used such a cloner, it wouldn't correctly simulate the inefficient adversary $A$, meaning the meta-reduction may not work. We could insist on cloning hardness where we require the adversary to clone the original state with high-fidelity. However, the resulting object turns out to be very mild, and is in fact equivalent to one-way puzzles by the techniques of~\cite{KT25}. Intuitively, the reason for this is that there is no efficient verifier to cloning the original state\footnote{Note that the original state is information-theoretically determined by the information seen by the adversary, so it is still the case that cloning the original state only be computationally hard.}. On the other hand, the hardness of creating two states that pass an efficient verifier seems to require very strong assumptions.  The only known constructions of such cloning hardness are simply public key quantum money schemes, which currently seem to require very strong assumptions.

Another related problem is that if the reduction runs $A$ many times, we need the ability to clone many times, but standard publicly verifiable unclonability only asks cloner to turn one copy into two. We could try then applying the cloner again to the copies to get even more. But the cloner may not work on the copies, since the cloner only is guaranteed to clone the original state, which may be different than the outputs of the cloner.

Similar to our first result, we sidestep some of the above difficulty by only obtaining a weak quantum money mini-scheme\footnote{For a 4 round PoQ, we actually get a weak form of quantum lightning~\cite{EC:Zhandry19b}, which is between mini-schemes and one-shot signatures. Essentially, lightning in a mini-scheme that is unclonable even if the adversary comes up with the original state, but there is no signing functionality.}. Even with this weaker unclonability notion, the proof is non-trivial, and we are only able to make the proof work unconditionally in the four-message setting. The reason for four-message protocols is that we are only essentially only concerned about partial transcripts that differ on the last message.

Through a careful analysis, we are able to extend the proof to any constant number of rounds, but under a conjectured quantum LOCC de Finetti theorem\footnote{And our proof here only achieves a quantum money mini-scheme, no lightning.}. Very roughly, the issue is that the states we obtain from the cloner may be highly entangled, which may introduce inherently quantum correlations between the prover messages obtained from these states. Such quantum correlations could potentially allow for distinguishing our simulation from the outputs of the original adversary $A$, which by virtue of being classical only has classical correlations between its different outputs.

We overcome this using a quantum de Finetti theorem. In general, a quantum de Finetti theorem says that if a quantum state is permutation-invariant, then the reduced state of a small portion of the state ``looks like'' a product state. In our case, the permutation-invariant state will be the set of copies of the honest prover's state obtained from the cloner, and by obtaining many more clones than the number of messages we need to produce, we only use a small subset of the copies. The de Finetti theorem allows us to argue that we can use the cloned states in a way that is indistinguishable from the inefficient adversary $A$. 

The particular de Finetti theorem we need is false if ``looks like'' is taken to mean that the reduced state is close in trace distance to a product state. But we only need a weaker notion which states that the reduced state is indistinguishable from a product state under LOCC communication between the various cloned states. The reason for LOCC communication is that ultimately we do not need the states themselves to be indistinguishable, but only the classical messages sent to the verifier.

Our 4-message proof actually already needs a de Finetti theorem. But in the 4-message case, we only need to worry about queries on transcripts differing on the final message. Subsequent queries to the final message will just be handled by querying the cloned states in a sequential manner, meaning the LOCC communication between the clones is actually one-way, in the sense that one clone sends a message to the next, which sends a message to the one after, and so on. For such one-way communication, the needed LOCC de Finetti theorem was actually proved in~\cite{1locc}, allowing us to prove our 4-message result.

For general constant-round protocols, we unfortunately need general LOCC communication, since the reduction may interleave rewinds at different stages of the protocol, which ultimately result in back-and-forth communication between the clones. Even by assuming such a general LOCC de Finetti theorem, our proof is delicate, as we have to handle the reduction querying on arbitrary transcripts that may differ at any point from previously queried transcripts. See Section~\ref{sec:publiccoinreductions} for details.

\paragraph{Discussion.} The only known constructions of public key quantum money require extremely strong or non-standard assumptions  (e.g.~\cite{EC:Zhandry19b,EC:LiuMonZha23,STOC:BosNehZha25}), have security proofs using idealized models (e.g.~\cite{ITCS:Zhandry24a}), or have no proofs at all (e.g.~\cite{ITCS:FGHLS12,KanShaSil22}). There is moreover some evidence supporting the difficulty of constructing public key quantum money from ``mild'' assumptions such as one-way functions or even lattices~\cite{C:LiuZha19,EC:LiuMonZha23,AC:AnaHuYue23,EC:Zhandry25}. Notice that public key quantum money \emph{can} be constructed from indisitnguishability obfuscation (iO) and one-way functions~\cite{EC:Zhandry19b}. In light of our separation between constant-round PoQs and iO, this shows that public key quantum money, while potentially necessary for public coin PoQs, is unlikely to be sufficient.

We leave as an interesting open question extending our result to quantum query reduction. One might try to use the compressed oracle techniques from our previous two meta-reduction results, but we were unable to get this to work with our approach to this result.

%% file: preliminaries.tex
\section{Preliminaries}

In this section, we discuss some notation and preliminary information, including definitions, that will be useful in the rest of the exposition.

\subsection{Notation and Conventions}

We write $\negl(\cdot)$ to denote any \emph{negligible} function, which is a function $f$ such that for every constant $c \in \mathbb{N}$ there exists $N \in \mathbb{N}$ such that for all $n > N$, $f(n) < n^{-c}$.
We will use $\mathsf{SD}(A,B)$ to denote the statistical distance between (classical) distributions $A$ and $B$. For natural numbers $n$ and $m$ we use $[n]$ to refer to the set $\{1,2,\ldots, n\}$ and $[m,n]$ to refer to the set $\{m, m+1,\ldots, n\}$. For a distribution $D$ we use $\Pr_D[x]$ to refer to the probability of sampling $x$ from $D$.\\

\noindent{\bf Quantum conventions.} A register $\reg{X}$ is a named Hilbert space $\bbC^{2^n}$. A pure state on register $\reg{X}$ is a unit vector $\ket{\psi} \in \bbC^{2^n}$, and we say that $\ket{\psi}$ consists of $n$ qubits. A mixed state on register $\reg{X}$ is described by a density matrix $\rho \in \bbC^{2^n \times 2^n}$, which is a positive semi-definite Hermitian operator with trace 1. 

A \emph{quantum operation} $F$ is a completely-positive trace-preserving (CPTP) map from a register $\reg{X}$ to a register $\reg{Y}$, which in general may have different dimensions. That is, on input a density matrix $\rho$, the operation $F$ produces $F(\rho) = \tau$ a mixed state on register $\reg{Y}$.
A \emph{unitary} $U: \reg{X} \to \reg{X}$ is a special case of a quantum operation that satisfies $U^\dagger U = U U^\dagger = \bbI^{\reg{X}}$, where $\bbI^{\reg{X}}$ is the identity matrix on register $\reg{X}$. If an operation acting on a multipartite system is only specified on some parts of the system, it is understood to act as identity on the other parts. A \emph{projector} $\Pi$ is a Hermitian operator such that $\Pi^2 = \Pi$, and a \emph{projective measurement} is a collection of projectors $\{\Pi_i\}_i$ such that $\sum_i \Pi_i = \bbI$.

We say a quantum circuit $C$ outputs strings in $\bin^n$ if $C$ acts on $\ket{0}$ to produce an $n$-qubit output register (potentially along with a junk register). The output of the circuit is the outcome of measuring the output register of $C\ket{0}$ in the computational basis. 

We use $\TD(\rho, \sigma)$ to refer to the trace distance of mixed states $\rho$ and $\sigma$. For pure states $\ket{\phi}$ and $\ket{\psi}$ we sometimes abuse notation and use $\TD(\ket{\phi}, \ket{\psi})$ to refer to $\TD(\ket{\phi}\!\bra{\phi}, \ket{\psi}\!\bra{\psi})$.
\subsection{Useful Theorems}
\begin{lemma}\label{lem:thetaED-to-TD}
For pure states $\ket{\psi}$ and $\ket{\psi'}$, if 
$\min_\theta \| \ket{\psi} - e^{i\theta}\ket{\psi'} \| \geq\delta$, 
then $\td(\ket{\psi},\ket{\psi'}) \geq\delta / \sqrt{2}$.
\end{lemma}

\begin{proof}
\begin{align*}
  \|\ket{\psi} - e^{i\theta}\ket{\psi'}\| 
    &= \sqrt{2 - 2 \mathsf{Re} \bigl( e^{i\theta} \langle \psi | \psi' \rangle \bigr)} \\
  \min_\theta \|\ket{\psi} - e^{i\theta}\ket{\psi'}\| 
    &= \sqrt{2 - 2|\langle \psi | \psi' \rangle|} \geq\delta .
\end{align*}
Thus
\[
  1 - \tfrac{\delta^2}{2} \geq|\langle \psi | \psi' \rangle| .
\]
Note that since the RHS is positive, $\delta^2 \leq2$. We can use the bound on the inner product to bound the trace distance.
\begin{align*}
  \td(\ket{\psi},\ket{\psi'}) 
    &= \sqrt{1 - |\langle \psi | \psi' \rangle|^2} \\
    &\geq\sqrt{1 - \left(1 - \tfrac{\delta^2}{2}\right)^2} \\
    &= \sqrt{\delta^2 - \tfrac{\delta^4}{4}} \\
    &\geq\delta \sqrt{1 - \tfrac{\delta^2}{4}} \\
    &\geq\delta / \sqrt{2}.
\end{align*}
\end{proof}

We also use the following theorem showing that the trace distance of pure states is upper bounded by their Euclidean distance.
\begin{theorem}
\label{thm:trace-dist-and-euclidean-dist}
    Let $\ket{\psi}$ and $\ket{\phi}$ be two pure states such that $|\ket{\psi} - \ket{\phi}| \leq \epsilon$. Then 
    \[
    \TD\left(\ketbra{\psi}, \ketbra{\phi}\right) \leq \epsilon
    \] 
\end{theorem}
\begin{proof}
    We have that
    \begin{align*}
        \epsilon^2 &\geq |\ket{\psi} - \ket{\phi}|^2 \\
        &= (\bra{\psi} - \bra{\phi}) (\ket{\psi} - \ket{\phi}) \\
        &= 2 - (\langle \phi | \psi \rangle + \langle \phi | \psi \rangle) \\
        &= 2 - 2 \text{Re} (\langle \phi | \psi \rangle) \\
        &\geq 2 - 2 |\langle \phi | \psi \rangle|,
    \end{align*}
    which can be rearranged to 
    \begin{align*}
        &|\langle \phi | \psi \rangle| \geq 1 - \frac{\epsilon^2}{2}.
    \end{align*}
    By the identity for trace distance of pure states,
    \begin{align*}
        \TD( \ket{\psi} \bra{\psi}, \ket{\phi} \bra{\phi}) &= \sqrt{1 - |\langle \psi | \phi \rangle|^2} \\
        &= \sqrt{1 - \left(1 - \frac{\epsilon^2}{2} \right)^2} \\
        &= \sqrt{\epsilon^2 - \frac{\epsilon^4}{4}} \\
        &\leq \epsilon.
    \end{align*}
    This completes the proof.
\end{proof}

\begin{lemma}
\label{lem:ED-to-SD}
For any two distributions $\cD$ and $\cD'$,
\[
\|\ket{\cD} - \ket{\cD'}\| \leq \sqrt{2\cdot\SD(\cD,\cD')}
\]
\end{lemma}
\begin{proof}
    Recall that $\ket{\cD} = \sum_x \sqrt{\Pr_{\cD}[x]} \ket{x}$ and $\ket{\cD'} = \sum_x \sqrt{\Pr_{\cD'}[x]} \ket{x}$. Therefore, 
    \begin{align*}
        \|\ket{\cD} - \ket{\cD'}\| &= \|\sum_x (\sqrt{\Pr_{\cD}[x]} - \sqrt{\Pr_{\cD'}[x]}) \ket{x}\|\\
        &= \sqrt{\sum_x (\sqrt{\Pr_{\cD}[x]} - \sqrt{\Pr_{\cD'}[x]})^2}\\
        &\leq \sqrt{\sum_x |(\sqrt{\Pr_{\cD}[x]} - \sqrt{\Pr_{\cD'}[x]})(\sqrt{\Pr_{\cD}[x]} + \sqrt{\Pr_{\cD'}[x]})|}\\
        &=\sqrt{\sum_x |\Pr_{\cD}[x] - \Pr_{\cD'}[x]|}\\
        &=\sqrt{2\cdot \SD(\cD,\cD')}
    \end{align*}
\end{proof}
    
\begin{theorem}[Theorem 4 in \cite{jls}]\label{thm:jls}
Let $\ket{\psi} \in \mathcal{H}$ be a quantum state. Define oracle $O_\psi = \bbI - 2 \ket{\psi}\bra{\psi}$ to be the reflection about $\ket{\psi}$. Let $\ket{\xi}$ be a state not necessarily independent of $\ket{\psi}$. Let $\mathcal{A}^{O_\psi}$ be an oracle algorithm that makes $q$ queries to $O_\psi$. For any integer $n > 0$, there is a quantum algorithm $\mathcal{B}$ that makes no queries to $O_\psi$ such that
\[
\TD\bigl(\mathcal{A}^{O_\psi}(\ket{\xi}), \mathcal{B}(\ket{\psi}^{\otimes n} \otimes \ket{\xi})\bigr) \leq \frac{q\sqrt{2}}{\sqrt{n+1}}.
\]
Moreover, the running time of $\mathcal{B}$ is polynomial in that of $\mathcal{A}$ and $n$.
\end{theorem}

\subsection{Quantum Cryptographic Primitives}
\begin{definition}[One-way Puzzles]
\label{def:owp}
A one-way puzzle is a pair of sampling and verification algorithms $(\Gen,\mathsf{Ver})$
with the following syntax. 
\begin{itemize}
\item $\Gen(1^n) \rightarrow (s,k)$, is a QPT algorithm that outputs a pair of classical strings $(s,k)$.
We refer to $s$ as the puzzle and $k$ as its key. Without loss of generality we may assume that $k\in\bin^n$.
\item $\mathsf{Ver}(s,k) \rightarrow \top$ or $\bot$,
is a Boolean function that maps every pair of classical strings $(k,s)$ to either $\top$ or $\bot$.
\end{itemize}
These satisfy the following properties.
\begin{itemize}
\item {\bf Correctness.} Outputs of the sampler pass verification with overwhelming probability, i.e., 
$$\Prr_{(s,k) \leftarrow \Gen(1^n)} [\Ver(s,k) = \top] = 1 - \negl(n)$$ 
\item {\bf Security.}
Given $s$, it is (quantum) computationally infeasible to find $k$ satisfying $\Ver(s,k) = \top$, i.e., for every quantum polynomial-sized adversary $\cA$ and every quantum advice state $\ket{\tau} = \{\ket{\tau_n}\}_{n \in \mathbb{N}}$,
 $$\Prr_{(s,k) \leftarrow \Gen(1^n)}[\Ver(s,\mathcal{A}(\ket{\tau},s)) = \top] = \negl(n)$$
\end{itemize}
\end{definition}

\begin{definition}[$\varepsilon$-Distributional One-way Puzzles]
\label{def:dist-owp}
    For $\varepsilon: \bbN \rightarrow \bbR$, a $\varepsilon$-distributional one-way puzzle is defined by a quantum polynomial-time generator $\Gen(1^n)$ that outputs a pair of classical strings $(s,k)$ such that 
    for every quantum polynomial-time adversary $\cA$, every (non-uniform, quantum) advice ensemble $\ket{\tau} = \{ \ket{\tau_n}\}_{n \in \mathbb{N}}$, for large enough $n \in \mathbb{N}$, 
     \[ \mathsf{SD} \left(
    \{s,k\}\, \{s,\cA(\ket{\tau},s)\} \right)
    \geq \varepsilon(n)\] where $(s,k)\leftarrow \Gen(1^n)$.
\end{definition}
We will sometimes simply refer to distributional one-way puzzles. This is taken to mean $1/p(n)$-distributional one-way puzzles for some non-zero polynomial $p$. 

The following theorem shows that distributional one-way puzzles can be amplified to (standard) one-way puzzles.
\begin{theorem}[Theorem 33 from \cite{CGG24}, rephrased]\label{thm:owp-amplification}
If there exists a polynomial $p(\cdot)$ for which $1/p(n)$-distributional one-way puzzles exist, then one-way puzzles exist.
\end{theorem}
\subsubsection{Proofs of Quantumness}

\begin{definition}[Proof of Quantumness](Adapted from \cite{mori-other})\label{def:poq}
A $(c(\lambda), s(\lambda))$-\emph{proof of quantumness} consists of a pair of algorithms: a QPT prover $P$ and a PPT verifier $V$, as well as correctness and soundness parameters $c(\cdot)$ and $s(\cdot)$.  
$P$ and $V$ interact over multiple rounds, at the end of which $V$ either accepts or rejects.  
We write
$\langle P,V\rangle(1^\lambda)$
to denote the final outcome of running the protocol with security parameter $\lambda$. We require the protocol to satisfy the following properties:
\begin{itemize}
    \item \textbf{Non-Triviality.} There exists a polynomial $t$ such that for all large enough $\lambda$,
\[
c(\lambda) - s(\lambda) \geq \frac{1}{t(\lambda)}.
\]
\item \textbf{Correctness.} For all large enough $\lambda$,
\[
\Pr[\langle P,V\rangle(1^\lambda)=1] \geq c(\lambda).
\]
\item \textbf{Soundness}
For all PPT provers $\widetilde{P}$, for infinitely many $\lambda$,
\[
\Pr[\langle \widetilde{P},V\rangle(1^\lambda)=1] \leq s(\lambda)
\]
\end{itemize}
If the final step of the verifier is computationally unbounded we call the construction an IV-PoQ. If the verifier messages consist of uniformly random strings and the final verification step only depends on the transcript we call the construction public coin. We will often drop the completeness and soundness parameters from the name of the primitive.
\end{definition}

\begin{definition}[(Inefficiently) Falsifiable Assumption]
\kabir{can be fleshed out a bit more}
\label{def:falsifiable}
A \emph{falsifiable assumption} consists of a (potentially inefficient) challenger $C$ and constant $\thres \in [0,1]$.  
The challenger interacts with an adversary and outputs either accept or reject. We say that an adversary $A$ \emph{breaks the assumption} if there exists a polynomial $p$ such that for infinitely many $\lambda$,
\[
\Pr\big[ \langle A(1^\lambda), C(1^\lambda) \rangle = 1 \big] \geq \thres + \frac{1}{p(\lambda)}.
\]
\end{definition}

\begin{definition}[PoQ with Black-Box Reduction]
\label{def:poq-bb}
A \emph{proof of quantumness} with black-box reduction to assumption $(C,\thres)$ consists of three algorithms: a QPT prover $P$, a verifier $V$, and a QPT reduction (potentially with non-uniform quantum advice) $R$, as well as correctness and soundness parameters $c(\cdot)$ and $s(\cdot)$.  
$P$ and $V$ interact over multiple rounds, at the end of which $V$ either accepts or rejects.  
We write
$\langle P,V\rangle(1^\lambda)$
to denote the final outcome of running the protocol with security parameter $\lambda$. We require the protocol to satisfy the following properties:
\begin{itemize}
    \item \textbf{Non-Triviality.} There exists a polynomial $t$ such that for all large enough $\lambda$,
\[
c(\lambda) - s(\lambda) \geq \frac{1}{t(\lambda)}.
\]
\item \textbf{Correctness.} For all $\lambda$,
\[
\Pr[\langle P,V\rangle(1^\lambda)=1] \geq c(\lambda).
\]
\item \textbf{Soundness by Reduction.}
For all polynomials $p'$, there exists a polynomial $p$ such that for all classical provers $\widetilde{P}$ where for infinitely many $\lambda$,
\[
\Pr[\langle \widetilde{P},V\rangle(1^\lambda)=1] \geq s(\lambda) + \frac{1}{p'(\lambda)},
\]
it is the case that for infinitely many $\lambda$,
\[
\Pr[\langle R^{\widetilde{P}},C\rangle(1^\lambda)=1] \geq \thres + \frac{1}{p(\lambda)}.
\]
\end{itemize}
If the final step of the verifier is  computationally unbounded we call the construction an IV-PoQ with black-box reduction to $(C,\thres)$.
\end{definition}
\begin{definition}[Interactive Puzzle](Adapted from \cite{PassV20})
    An interactive puzzle is synonymous with a public coin proof of quantumness (Definition \ref{def:poq}) except the prover is allowed to be computationally unbounded.
\end{definition}
We will use the following theorem about interactive puzzles.
\begin{theorem}(Adapted from \cite{PassV20})\label{thm:pass-v}
The existence of an $\ell$-round interactive puzzle for constant $\ell \geq 4$ implies either the existence of an $(\ell-2)$-round interactive puzzle or uniform i.o OWFs.
\end{theorem}

%% file: cloning-hardness.tex
\section{Cloning Hardness}\label{sec:cloninghardness}
In this section we define a new set of primitives that capture the forms of cloning hardness we build from proofs of quantumness. We also show a simple amplification theorem.
\subsection{Definitions}
\begin{definition}[Weak Minischemes]
    \label{def:weak-minischeme} 
For a polynomial $n$ and functions $c:\mathbb{N}\to [0,1]$, $s:\mathbb{N}\to [0,1]$, 
we define $(n(\lambda),c(\lambda),s(\lambda))$--weak minischemes as a pair of QPT algorithms 
$(\Samp,\Ver)$ such that:
\begin{itemize}
    \item $\Samp(1^\lambda)$ outputs a string ${\srno}$ and a pure quantum state $\ket{\psi_\srno}$.  
    We require $\ket{\psi_{\srno}}$ to be unique given $\srno$ and $\lambda$.
    \item $\Ver({\srno},\rho)$ outputs a bit indicating accept or reject.
\end{itemize}

\noindent The following conditions must hold:
\begin{itemize}
    \item \textbf{Correctness:} For all large enough $\lambda \in \mathbb{N}$,
    \[
    \Pr\left[\Ver({\srno},\ket{\psi_{\srno}}\!\bra{\psi_{\srno}})=1 \middle| ({\srno},\ket{\psi_{\srno}}) \gets \Samp(1^\lambda)\right] 
    \ge c(\lambda).
    \]

    \item \textbf{Soundness:} For all QPT adversaries $\mathcal{A}$ with advice $\{\ket{\phi_\lambda}\}_\lambda$,  
    for all large enough $\lambda$,
    \[
    \Pr\left[ \forall t \in [n(\lambda)], 
    \Ver({\srno},\rho_{\reg{A}_t})=1 \middle| \begin{array}{l}
         ({\srno},\ket{\psi_{\srno}}) \gets \Samp(1^\lambda)\\
         \rho_{\reg{A}_1 \cdots \reg{A}_{n(\lambda)}} \gets \mathcal{A}({\srno},\ket{\psi_{\srno}}, \ket{\phi_\lambda})
    \end{array}
     \right] 
    \le s(\lambda).
    \]
\end{itemize}
If $c(\lambda) \ge 1 - \negl(\lambda)$ and $s(\lambda) \le \negl(\lambda)$, 
we refer to the scheme as a $n(\lambda)$-weak minischeme.  
If the security only holds for QPT adversaries and for each adversary there are infinitely many $\lambda$ such that security holds 
(as opposed to for all large enough $\lambda$), we call the scheme a $(n(\lambda), c(\lambda), s(\lambda))$-weak uniform i.o. minischeme.
\end{definition}
\begin{definition}[Weak Lightning]
    \label{def:weak-lightning} 
For a polynomial $n$ and functions $c:\mathbb{N}\to [0,1]$, $s:\mathbb{N}\to [0,1]$ 
we define $(n(\lambda),c(\lambda),s(\lambda))$--weak lightning as a pair of QPT algorithms 
$(\Samp,\Ver)$ such that:
\begin{itemize}
    \item $\Setup(1^\lambda)$ outputs classical public parameters $\pp$.  
    \item $\Samp(\pp)$ outputs a string ${\srno}$ and a pure quantum state $\ket{\psi_{\srno}}$.  
    We require $\ket{\psi_{\srno}}$ to be unique given ${\srno}, \pp,$ and $\lambda$.
    \item $\Ver(\pp,{\srno},\rho)$ outputs a bit indicating accept or reject.
\end{itemize}

\noindent The following conditions must hold:
\begin{itemize}
    \item \textbf{Correctness:} For all large enough $\lambda \in \mathbb{N}$,
    \[
    \Pr\left[\Ver(\pp, {\srno},\ket{\psi_{\srno}}\!\bra{\psi_{\srno}})=1 \middle|\begin{array}{l}
        \pp \leftarrow \Setup(1^\lambda)\\
         ({\srno},\ket{\psi_{\srno}}) \gets \Samp(\pp)
    \end{array} \right] 
    \ge c(\lambda).
    \]

    \item \textbf{Soundness:} For all QPT adversaries $\mathcal{A}$ with advice $\{\ket{\phi_\lambda}\}_\lambda$,  
    for all large enough $\lambda$,
    \[
    \Pr\left[ \forall t \in [n(\lambda)], 
    \Ver(\pp, {\srno},\rho_{\reg{A}_t})=1 \middle| \begin{array}{l}
        \pp \leftarrow \Setup(1^\lambda)\\
         {\srno}, \rho_{\reg{A}_1 \cdots \reg{A}_{n(\lambda)}} \gets \mathcal{A}(\pp, \ket{\phi_\lambda})
    \end{array}
     \right] 
    \le s(\lambda).
    \]
\end{itemize}
If $c(\lambda) \ge 1 - \negl(\lambda)$ and $s(\lambda) \le \negl(\lambda)$, 
we refer to the scheme as $n(\lambda)$-weak lightning.  
If the security only holds for QPT adversaries and for each adversary there are infinitely many $\lambda$ such that security holds 
(as opposed to for all large enough $\lambda$), we call the scheme $(n(\lambda), c(\lambda), s(\lambda))$-weak uniform i.o. lightning.
\end{definition}
\begin{definition}[Weak Signature Tokens]
    \label{def:weak-tokens}
For a polynomial $n$ and functions $c:\mathbb{N}\to [0,1]$, $s:\mathbb{N}\to [0,1]$ 
we define $(n(\lambda),c(\lambda),s(\lambda))$-weak signature tokens as a tuple of QPT algorithms 
$(\Samp,\Sign,\Ver)$ such that:
\begin{itemize}
    \item $\Samp(1^\lambda)$ outputs a string $\pp$ and a pure quantum state $\ket{\psi_\pp}$.  
    We require $\ket{\psi_\pp}$ to be unique given $\pp$ and $\lambda$.
    \item $\Sign(\pp, \ket{\psi}, r)$ outputs a classical signature $\sig$.
    \item $\Ver(\pp, r,\sig)$ outputs a bit indicating accept or reject.
\end{itemize}

\noindent The following conditions must hold:
\begin{itemize}
    \item \textbf{Correctness:} For all large enough $\lambda \in \mathbb{N}$,
    \[
    \Pr\left[\Ver(\pp,r,\sig)=1 \middle| \begin{array}{r}
        (\pp,\ket{\psi_\pp}) \gets \Samp(1^\lambda)\\
        r \leftarrow \bin^\lambda\\
        \sig\leftarrow\Sign(\pp, \ket{\psi_\pp}, r)
    \end{array}\right] 
    \ge c(\lambda).
    \]

    \item \textbf{Soundness:} For all QPT adversaries $\mathcal{A}$ with advice $\{\ket{\phi_\lambda}\}_\lambda$,  
    for all large enough $\lambda$,
    \[
    \Pr\left[ \begin{array}{l}
      \forall t \in [n(\lambda)], \\
      \Ver(\pp, r_t,\sig_t)=1
    \end{array}
     \middle| \begin{array}{r}
        (\pp,\ket{\psi_\pp}) \gets \Samp(1^\lambda)\\
        r_1, \ldots, r_{n(\lambda)} \leftarrow \bin^\lambda\\
        \sig_1, \ldots, \sig_{n(\lambda)}\leftarrow\cA(\pp, \ket{\psi_\pp}, r_1, \ldots, r_{n(\lambda)}, \ket{\phi_\lambda})
    \end{array}
     \right] 
    \le s(\lambda).
    \]
\end{itemize}
If $c(\lambda) \ge 1 - \negl(\lambda)$ and $s(\lambda) \le \negl(\lambda)$, 
we refer to the scheme as a $n(\lambda)$-weak token scheme.  
If the security only holds for QPT adversaries and for each adversary there are infinitely many $\lambda$ such that security holds 
(as opposed to for all large enough $\lambda$), we call the scheme a $(n(\lambda), c(\lambda), s(\lambda))$-weak uniform i.o. signature token scheme.
\end{definition}
\begin{definition}[Weak One-Shot Signatures]
    \label{def:weak-sig}
For a polynomial $n$ and functions $c:\mathbb{N}\to [0,1]$, $s:\mathbb{N}\to [0,1]$ 
we define $(n(\lambda),c(\lambda),s(\lambda))$-weak one-shot signatures (OSS) as a tuple of QPT algorithms 
$(\Samp,\Sign,\Ver)$ such that:
\begin{itemize}
    \item $\Setup(1^\lambda)$ outputs a string $\pp$.  
    \item $\Samp(\pp)$ outputs a string $\srno$ and a pure quantum state $\ket{\psi_\srno}$.
    We require $\ket{\psi_\srno}$ to be unique given $\pp,\srno,$ and $\lambda$.
    \item $\Sign(\pp, \srno, \ket{\psi_\srno}, r) $ outputs a classical signature $\sig$.
    \item $\Ver(\pp, \srno, r,\sig)$ outputs a bit indicating accept or reject.
\end{itemize}

\noindent The following conditions must hold:
\begin{itemize}
    \item \textbf{Correctness:} For all large enough $\lambda \in \mathbb{N}$,
    \[
    \Pr\left[\Ver(\pp,\srno,r,\sig)=1 \middle| \begin{array}{r}
        \pp \leftarrow \Setup(1^\lambda)\\
        (\srno,\ket{\psi_\srno}) \gets \Samp(\pp)\\
        r \leftarrow \bin^\lambda\\
        \sig\leftarrow\Sign(\pp,\srno, \ket{\psi_\pp}, r)
    \end{array}\right] 
    \ge c(\lambda).
    \]

    \item \textbf{Soundness:} For all QPT adversaries $\mathcal{A}$ with advice $\{\ket{\phi_\lambda}\}_\lambda$,  
    for all large enough $\lambda$,
    \[
    \Pr\left[ \forall t \in [n(\lambda)], 
    \Ver(\pp, \srno,r_t,\sig_t)=1 \middle| \begin{array}{r}
     \pp \leftarrow \Setup(1^\lambda)\\
        (\srno,\sigma) \gets \cA(\pp, \ket{\phi_\lambda})\\
        r_1, \ldots, r_{n(\lambda)} \leftarrow \bin^\lambda\\
        \sig_1, \ldots, \sig_{n(\lambda)}\leftarrow\cA(\sigma, r_1, \ldots, r_{n(\lambda)})
    \end{array}
     \right] 
    \le s(\lambda).
    \]
\end{itemize}
If $c(\lambda) \ge 1 - \negl(\lambda)$ and $s(\lambda) \le \negl(\lambda)$, 
we refer to the scheme as a $n(\lambda)$-weak OSS scheme.  
If the security only holds for QPT adversaries and for each adversary there are infinitely many $\lambda$ such that security holds 
(as opposed to for all large enough $\lambda$), we call the scheme a $(n(\lambda), c(\lambda), s(\lambda))$-weak uniform i.o.  OSS scheme.
\end{definition}

\subsection{Amplification For Weak Minischemes}
\begin{theorem}\label{thm:minischeme-amp}
Suppose there exists a $(n(\lambda), c(\lambda), s(\lambda))$-weak (uniform i.o.) minischeme where such that there exists a polynomial\ $t(\lambda)$ such that 
\[
  (1-s(\lambda)) > n(\lambda)(1-c(\lambda)) + \frac{1}{t(\lambda)}.
\]
Then there exists a $(n(\lambda))$-weak (uniform i.o.) minischeme.
\end{theorem}

\begin{proof}
We will prove for weak minischemes with standard security. Essentially the same proof applies for weak uniform i.o. minischemes.

First we show the existence of $(n(\lambda), 1-\mathsf{negl}(\lambda), 1-\tfrac{1}{p(\lambda)})$-weak minischemes for some polynomial $p(\lambda)$.  
Let $(\mathsf{Samp}, \mathsf{Ver})$ be a $(n(\lambda), c(\lambda), s(\lambda))$-weak minischeme.  
Set $\ell(\lambda) := \lambda t(\lambda)$ and define $\mathsf{Samp}', \mathsf{Ver}'$ as:
\begin{itemize}
    \item $\mathsf{Samp}'(1^\lambda):$
    \begin{itemize}
        \item For $i \in [\ell(\lambda)]$:
        \begin{itemize}
            \item $s_i, \ket{\psi_i}\leftarrow\Samp(1^\lambda)$
        \end{itemize}
        \item Output $(s_1 \cdots s_{\ell(\lambda)}, \ket{\psi_1}\otimes \cdots \otimes \ket{\psi_{\ell(\lambda)}})$
    \end{itemize}
    \item $\mathsf{Ver}'(s,\sigma):$
    \begin{itemize}
        \item Parse $s$ as $s_1, \ldots, s_{\ell(\lambda)}$
        \item Apply $\mathsf{Ver}(s_1,\cdot) \otimes \cdots \otimes \mathsf{Ver}(s_{\ell(\lambda)},\cdot)$ to $\sigma$
        \item Accept only if at least $c(\lambda) - \tfrac{1}{2t(\lambda)}$ fraction accept.
    \end{itemize}
\end{itemize}
We claim that $(\mathsf{Samp}', \mathsf{Ver}')$ is a $(n(\lambda), 1-\mathsf{negl}, 1-\tfrac{1}{2n(\lambda)t(\lambda)})$-weak minischeme.  
The correctness follows from a Chernoff bound argument. Suppose there exists an adversary $\mathcal{A}$ that breaks security.  
We define an adversary $\mathcal{A}'$ that breaks security of the underlying weak minischeme for the same $\lambda$ values on which $\cA$ breaks security.
\begin{itemize}
    \item $\mathcal{A}'(s, \ket{\psi}):$
    \begin{itemize}
        \item Sample $i^* \gets [\ell(\lambda)]$.
        \item Sets $s_{i^*}, \ket{\psi_{i^*}} \gets (s, \ket{\psi})$.
        \item For $i \in [\ell(\lambda)] \setminus \{i^*\}:$
        \begin{itemize}
            \item Sample $(s_i, \ket{\psi_i}) \gets \mathsf{Samp}(1^\lambda)$
        \end{itemize}
        \item $\rho_{\reg{A}_1,\ldots,\reg{A}_{n(\lambda)}} \leftarrow \mathcal{A}(s_1,\ldots,s_{\ell(\lambda)}, \ket{\psi_1}\otimes\cdots\otimes\ket{\psi_{\ell(\lambda)}})$
        \item Output $\rho_{\reg{A}_{1,i^*},\reg{A}_{2,i^*},\ldots,\reg{A}_{n(\lambda),i^*}}$,  
where for $j\in[n(\lambda)]$, $\reg{A}_{j,i^*}$ is the $i^*$-th subsystem of register $\reg{A}_j$.
    \end{itemize}
\end{itemize}
We know by the assumption if $\mathsf{Ver}'$ was applied to $\rho$ output by $\cA$, over the randomness of the experiment it would succeed with probability at least 
$1-\tfrac{1}{2n(\lambda)t(\lambda)}$. When this occurs, for each $\reg{A}_{j}$ at least a $c(\lambda) - \tfrac{1}{2n(\lambda)t(\lambda)}$ fraction of subsystems pass $\mathsf{Ver}$.  
By a union bound, this means that for at least an $1 - n(\lambda)\left( 1+ \tfrac{1}{2n(\lambda)t(\lambda)} -c(\lambda)\right)$ fraction of choices of $i$, $\mathsf{Ver}$ succeeds for all of $(\rho_{\reg{A}_{1,i}}, \ldots, \rho_{\reg{A}_{n(\lambda),i}})$.  
Thus $\mathcal{A}'$ successfully breaks the security of the underlying scheme whenever $i^*$ is in this fraction.
Therefore, the probability that $\mathcal{A}'$ succeeds is at least
\begin{align*}
    \Bigl(1 - \tfrac{1}{2n(\lambda)t(\lambda)}\Bigr)&\cdot\Bigl(1 - n(\lambda)\bigl( 1+ \tfrac{1}{2n(\lambda)t(\lambda)} -c(\lambda)\bigr)\Bigr)\\
    &\geq 1 - \tfrac{1}{2n(\lambda)t(\lambda)} - n(\lambda)\bigl( 1+ \tfrac{1}{2n(\lambda)t(\lambda)} -c(\lambda)\bigr)\Bigr)\\
    &\geq 1 - \tfrac{1}{t(\lambda)} - n(\lambda)(1-c(\lambda))\\
    &> 1 - (1-s(\lambda))\\
    &= s(\lambda)
\end{align*}
which contradicts the security of the underlying weak minischeme.\\

To conclude the proof we note that $(n(\lambda), 1-\mathsf{negl}(\lambda), 1-\tfrac{1}{p(\lambda)})$ weak minischemes imply $n(\lambda)$-weak minischemes by a simple parallel repetition argument . \kabir{maybe parallel repetition isn't so simple, should cite/clarify}
\end{proof}

%% file: generic.tex
\section{Generic Lower Bounds for Constant-Round PoQ}\label{sec:generic}
In this section we show generic lower bounds for proofs of quantumness. 
Existing literature on interactive protocols allows us to make some simple observations. First we note that in the absence of (a variant of) one-way functions, we can assume the PoQ is public coin "for free".
\begin{theorem}\label{thm:private-to-public}Constant-round PoQ either imply \emph{public-coin} constant-round PoQ or imply classically-secure uniform auxiliary-input infinitely-often one-way functions.
\end{theorem}
\begin{proof}
    Follows from Theorem 8 in \cite{private-to-public}. 
\end{proof}
\noindent Second, we note that for public-coin PoQ, an inverse-polynomial gap between completeness and soundness is sufficient to obtain public-coin $(1-\negl(\lambda), \negl(\lambda))$-PoQ by Chernoff-style parallel repetition.
\begin{theorem}\label{thm:poq-parallel-rep}
    Public-coin $(c(\lambda), s(\lambda))$-PoQ (which satisfy non-triviality as in Definition \ref{def:poq}) imply the existence of public-coin $(1-\negl(\lambda), \negl(\lambda))$-PoQ. Additionally, the new PoQ is constructed by parallel repetition of the original protocol where the final verifier accepts if the number of accepting repetitions is above a specified threshold.
\end{theorem}
\begin{proof}
    Follows from Theorem 1 in \cite{parallel-rep}.
\end{proof}
\noindent Finally, we show the main result of this section: a round collapse theorem for public-coin PoQs in the absence of certain types of cryptography.
\begin{theorem}
Let $\langle P,V\rangle$ be an $\ell$-round public-coin PoQ for constant $\ell \geq 4$.  
Then one of the following must hold.
\begin{itemize}
    \item $(\ell-2)$-round public-coin PoQs exist.
    \item quantum-secure uniform i.o.~OWFs exist.
    \item Uniform i.o.~$(\poly(\lambda))$-weak signature tokens (Definition \ref{def:weak-tokens}) exist.
\end{itemize} 
If $\ell = 4$ we additionally obtain that either there exists a $2$-round public-coin PoQ or uniform i.o.~OWFs exist or uniform i.o.~$(\poly(\lambda), 1-\negl, 0.51)$-weak OSS (Definition \ref{def:weak-sig}) exist.
\end{theorem}

\begin{proof}
By Theorem \ref{thm:poq-parallel-rep} we may assume $1-\negl(\lambda)$ correctness and $\negl(\lambda)$ soundness without loss of generality. By Theorem \ref{thm:pass-v}, for any $\ell$-round public-coin interactive puzzle, either there exists an $(\ell-2)$-round public-coin interactive puzzle or classically-secure uniform i.o.~OWFs exist.  
Additionally, opening up the proof of Theorem \ref{thm:pass-v} in \cite{PassV20},  the round-reduced construction is as follows.  

For the first $\ell-4$ rounds, the prover and verifier proceed identically to the original interactive puzzle, obtaining transcript $\tau$.
In the $(\ell-3)$-th round, the verifier additionally samples $p(\lambda)$ random strings $r_1,\ldots,r_{p(\lambda)}$ for some polynomial $p(\lambda)$, and sends them along with the original verifier's $(\ell-3)$-th message $r$.  
In the $(\ell-2)$-th message, the prover sends $m$, which is the original prover's message on transcript $\tau\|r$, and $(m_1',\ldots,m_{p(\lambda)}')$ where $m_i$ is the original prover's $\ell$th message on transcript $(\tau\|r\|m\|r_i)$.  
The verifier accepts if for all $i \in [p(\lambda)]$ the original verifier accepts $(\tau \| r \| m \| r_i \| m'_{\ell})$.  
This round-collapsed protocol is an $(\ell-2)$-round $(1-\mathsf{negl}(\lambda), \tfrac{1}{2})$-interactive puzzle, which can be amplified via parallel repetition.  

Since PoQ are interactive puzzles, in the absence of uniform classically-secure i.o.~OWFs the above round collapse can be applied to a $\ell$-round public-coin PoQ to obtain an $(\ell-2)$-round public coin protocol that preserves soundness. However, this alone does not give an $(\ell-2)$-round public-coin PoQ. This is because the honest prover in the round-reduced protocol must an $m'_{i}$ for each $i$. Producing this message may require a measurement of the prover's internal state, and therefore it may not be possible to produce more than one such message, i.e., it may not be possible to "rewind" to quantum prover. Luckily, we are able to harness this hardness of rewinding to build cloning hardness. Specifically we show that in the absence of uniform i.o.~weak signature tokens, the round collapsed protocol has an efficent quantum prover, which gives a $(\ell-2)$-round public coin PoQ. 

Define $\mathsf{Samp}, \mathsf{Sign}, \mathsf{Ver}$ as follows:

\begin{itemize}
    \item $\mathsf{Samp}(1^\lambda):$
    \begin{itemize}
        \item Run honest quantum prover with the honest verifier for the original protocol up to the $(\ell-2)$-th round, obtaining transcript $\tau \| r \| m$ and (purified) internal prover state $\ket{\psi}$.  
        \item Output $(\tau \| r \| m, \ket{\psi})$.
    \end{itemize}

    \item $\mathsf{Sign}(\tau \| r \| m, \ket{\psi}, r'):$  
    \begin{itemize}
        \item Using the transcript, internal state, and $r'$ interpreted as the $(\ell-1)$-th message, compute the final prover message $m'$ for the original protocol and output $m'$.
    \end{itemize}

    \item $\mathsf{Ver}(\tau \| r \| m, r', m'):$  
    \begin{itemize}
        \item Accept if the original verifier accepts $\tau \| r \| m \| r' \| m'$.
    \end{itemize}
\end{itemize}
By the correctness of the original PoQ, $(\mathsf{Samp},\mathsf{Sign},\mathsf{Ver})$ satisfies correctness of a $p(\lambda)$-weak token scheme.  
Consider the two possibilities:
\begin{itemize}
    \item For all QPT adversaries $\mathcal{A}$, there exist infinitely many $\lambda$ such that
    \[
    \Pr\Big[ \forall i \in [p(\lambda)],\ \mathsf{Ver}(\mathsf{pp}, r_i, s_i) = 1 \Big] \le 0.51,
    \]
    where $(r_1,\ldots,r_{p(\lambda)}) \gets \{0,1\}^{*}$, $(\mathsf{pp}, \ket{\psi}) \gets \mathsf{Samp}(1^\lambda)$, \\and $ (s_1,\ldots,s_{p(\lambda)}) \gets \mathcal{A}(\mathsf{pp},\ket{\psi},r_1,\ldots,r_{p(\lambda)})$.

    \item There exists a QPT adversary $\mathcal{A}$ such that for all $\lambda$,
    \[
    \Pr\Big[ \forall i \in [p(\lambda)],\ \mathsf{Ver}(\mathsf{pp}, r_i, s_i) = 1 \Big] > 0.51
    \]
    where $(r_1,\ldots,r_{p(\lambda)}) \gets \{0,1\}^{*}$, $(\mathsf{pp}, \ket{\psi}) \gets \mathsf{Samp}(1^\lambda)$, \\and $ (s_1,\ldots,s_{p(\lambda)}) \gets \mathcal{A}(\mathsf{pp},\ket{\psi},r_1,\ldots,r_{p(\lambda)})$.
\end{itemize}
In the former case, we obtain a uniform i.o.~$(p(\lambda), 1-\mathsf{negl}(\lambda), 0.51)$-weak token scheme, which can be amplified by parallel repetition \kabir{cite}.

In the latter case, we can use $\mathcal{A}$ to compute the final prover message of the round-collapsed protocol by using the honest prover to interact with the honest verifier up to the $(\ell-1)$-th round, receiving $(r, r_1,\ldots,r_{p(\lambda)})$ from the verifier, sending $r$ to the prover to obtain $m$ and internal state $\ket{\psi}$, running $\mathcal{A}(\tau \| r \| m, \ket{\psi}, r_1,\ldots,r_{p(\lambda)})$ to obtain $m_1',\ldots,m_{p(\lambda)}'$, and sending $m'$ to the verifier.  
This new prover achieves completeness $0.51$ and therefore the round-reduced protocol is a public-coin $(0.51, \tfrac{1}{2})$-PoQ in $\ell-2$ rounds, which can then be amplified as in Theorem \ref{thm:poq-parallel-rep}.\\

\noindent In the case where $\ell=4$ we obtain a uniform i.o.~$(p(n), 1-\negl, 0.51)$-weak one-shot signature or a public coin $(\ell-2)$-PoQ as follows. 
\begin{itemize}
    \item $\mathsf{Setup}(1^\lambda):$
    \begin{itemize}
        \item Output the first verifier message as $\pp$.
    \end{itemize}
    \item $\mathsf{Samp}(\pp):$
    \begin{itemize}
        \item Run the honest prover on the first verifier message to obtain $\srno$ and (purified) prover state $\ket{\psi_\srno}$.
    \end{itemize}

    \item $\mathsf{Sign}(\pp, \srno,  \ket{\psi_\srno}, r):$  
    \begin{itemize}
        \item Using the transcript, internal state, and $r$ interpreted as the second verifier message, compute the final prover message $\sig$ for the original protocol and output $\sig$.
    \end{itemize}

    \item $\mathsf{Ver}(\pp,\srno,r,\sig):$  
    \begin{itemize}
        \item Accept if the original verifier accepts $\pp\|\srno\|r\|\sig$.
    \end{itemize}
\end{itemize}
The win-win follows from an almost identical argument as above, although we do not obtain the same amplification for the weak OSS. Since the security game for weak OSS is four round, parallel repetition is not trivial. While it is possible that existing literature on parallel repetition for public coin protocols can amplify security for weak OSS we do not show this here. \kabir{cite someone?}

All that remains is to show that classically-secure uniform i.o.~OWFs imply either quantum-secure uniform i.o.~OWFs or 2-round PoQ. Let $f$ be a classically-secure uniform i.o.~OWF. If $f$ is a quantum-secure uniform i.o.~weak OWF we are done since by \cite{Yao} these can be amplified to give quantum-secure uniform i.o.~OWF. If not then there exists a QPT $\cA$ that inverts $f$ with probability $1-1/\poly(\lambda)$. We define a two round $(1-1/\poly(\lambda), \negl(\lambda))$-PoQ as follows: The verifier samples a random $x$ and sends $f(x)$. The verifier accepts if the prover's message is $x'$ such that $f(x) = f(x')$. Soundness holds by the security of the one-way function and $(1-1/\poly(\lambda))$-completeness holds since a quantum prover can run $\cA$ and succeed whenever $\cA$ successfully inverts. This PoQ can then be amplified by parallel repetition.
\end{proof}

\begin{corollary}
The existence of constant-round public-coin proofs of quantumness implies one of the following:
\begin{itemize}
    \item There exist quantum-secure uniform i.o.~OWFs
    \item There exist uniform i.o.~($\poly(\lambda)$)-weak token schemes.
    \item There exist 2-round public-coin proofs of quantumness
\end{itemize}
\end{corollary}

\begin{proof}
The statement follows by repeated application of the theorem.  
Since there are only constant rounds, the runtimes remain polynomial.
\end{proof}
\kabir{We include for completeness the following theorem, which provides a generic lower bound for public-coin constant-round PoQ without win-wins}

%% file: separation.tex
\section{Separating Constant-Round IV-PoQ from iO and OWP}\label{sec:iOimposs}
In this section we rule out the existence of fully black-box constructions of constant-round PoQ from (even sub-exponentially secure) indistinguishability obfuscation (iO) and one-way permutations (OWPs). 
We will be building upon the framework introduced by~\cite{pdqp} (specifically, the sequence of hybrids in Subsection \ref{pdqp-subsec}) \kabir{flag in sections too} for ruling out constructions of forms of quantum collision resistance from iO and OWPs, although we include the full details of the separation here for completeness. 

Our proof will consist of assuming the existence of a black-box construction of $\ell$-round PoQ from iO and OWPs for some constant $\ell$. We will then construct an oracle $(O,P)$ such that
\begin{itemize}
    \item There exists a black-box construction $C^O$ of iO and OWP that is secure in the presence of $P$.
    \item There \textit{does not} exist a black box construction $\langle \cP^O, \cV^O\rangle$ of $\ell$-round IV-PoQ that is secure in the presence of $P$.
\end{itemize}
This contradicts the existence of a fully black-box construction (i.e with black-box construction and black-box reduction) of $\ell$-round PoQ from iO and OWP. This is because we may instantiate the construction using oracle $S$ (which provides secure constructions of iO and OWPs) obtaining a candidate construction of $\ell$-round PoQ. However, no construction of $\ell$-round PoQ from oracle $O$ is secure in the presence of oracle $P$, i.e there exists an efficient classical adversary $A^{S, P}$ that breaks the security of the candidate PoQ. Combined with the black-box reduction this results in an efficient adversary that uses $S$ and $P$ to break the security of our constructed iO and OWP. Since our constructions are shown to be secure, this leads to contradiction.

\kabir{It feels like the section numbering is not the order a reader should encounter them in. I think i will reorder things to make it clearer. The order below is how the concepts should flow}
Fully black-box constructions of constant-round PoQ are defined in Section 6.5. We then define the oracles O and P using the distributions defined in section 6.3. We show in sections 6.5.1 and 6.5.2 that there exist constructions of iO and OWPs from $O$ that are secure using $P$. To do so we crucially use the discussion on simulating classical oracles using unitaries from sections 6.1 and 6.2, as well as the one-way to hiding properties of the simulating unitary for $P$ proven in section 6.4. Next, in section 6.5.3 we show that $P$ can be used to break every construction of $\ell$-round PoQ from $O$. Finally, this leads to theorem 6.7 that rules out fully black-box constructions.

For $q \in \bbN$, we define a $q$-query algorithm as an algorithm that makes at most $q$ quantum queries, each of length at most $q$. We also assume that the output of a $q$-query algorithm is of length at most $q$. If $q$ is a function of the security parameter $\lambda$ the queries and output lengths are understood to be bounded by $q(\lambda)$.
\subsection{Defining Quantum Oracle Queries}
\label{subsec:def-quantum-queries}
Let $F$ 
be a function from bitstrings to bitstrings.
Quantum queries to $F$ may be modeled as follows. The query prepares registers 
$\reg{X}\reg{Y}\reg{Z}$ where $\reg{X}$ contains the query string, $\reg{Y}$ is 
the response register, and $\reg{Z}$ may contain auxiliary information. 
A query consists of applying the unitary
\[
  \sum_x \ket{x}\!\bra{x}_{\reg{X}} \otimes X^{F(x)}_{\reg{Y}} .
\]
Equivalently, we may initialize register $\reg{D} = \bigotimes_{x \in \{0,1\}^*} \reg{D}_x$ 
to the truth table of $F$. Formally, $\reg{D}$ is initialized to 
$\bigotimes_x \ket{F(x)}_{\reg{D}_x}$. A query may now be modeled as
\[
  \sum_x \ket{x}\!\bra{x}_{\reg{X}} \otimes \mathsf{CNOT}_{\reg{D}_x \reg{Y}} .
\]
where $\mathsf{CNOT}_{\reg{D}_x \reg{Y}}$ maps $\ket{s_1}_{\reg{D}_x}\ket{s_2}_\reg{Y}$ to $\ket{s_1}_{\reg{D}_x}\ket{s_2 \oplus s_1}_\reg{Y}$.

\subsection{Non-uniform oracles}
\label{subsec:def-comp-oracle-for-products}
We introduce tools for analyzing quantum queries to a class of random oracles. 
Let $\{\cD_x\}_{x \in \{0,1\}^*}$ be a collection of distributions over strings. 
Let $\cD$ be the product distribution induced by these distributions, i.e., $\cD := \bigotimes_x \cD_x$. 
We say an oracle $O$ is sampled from $\cD$ when the $x$-th row of its truth table 
is sampled independently from $\cD_x$. Equivalently, the truth table $O$ is sampled from $\cD$.

We define the \emph{compressed oracle} $\mathsf{CStO}_\cD$ as follows. 
Initialize register $\reg{D} := \bigotimes_x \reg{D}_x$ with $\ket{\bot}_{\reg{D}} := \bigotimes_x \ket{\bot}_{\reg{D}_x}$. 
For each $x$, let $\ket{\cD_x}$ represent a pure state of the form
\[
  \ket{\cD_x} :=\sum_z \sqrt{\Pr_{\cD_x}[z]} \ket{z}.
\]
Define 
\[
  U_x := \ket{\cD_x}\!\bra{\bot}_{\reg{D}_x} + \ket{\bot}\!\bra{\cD_x}_{\reg{D}_x}
  + \big(\bbI - \ket{\bot}\!\bra{\bot}_{\reg{D}_x} - \ket{\cD_x}\!\bra{\cD_x}_{\reg{D}_x}\big).
\]
Then define
\[
  U := \sum_x \ket{x}\!\bra{x}_{\reg{X}} \otimes U_x.
\]
We refer to $U$ as the \emph{compression unitary} for $\cD$. 
The querier prepares registers $\reg{XYZ}$ and a query to $\mathsf{CStO}_\cD$ is implemented as
\[
  \sum_x \ket{x}\!\bra{x}_{\reg{X}} \otimes (U_x \circ \mathsf{CNOT}_{\reg{D}_x \reg{Y}} \circ U_x).
\]
We note that $\mathsf{CStO}_\cD$ can be implemented by initializing $\reg{D}$ to $\ket{\bot}_\reg{D}$ and making two queries to $U$ per $\mathsf{CStO}_\cD$ query.

\begin{theorem}[Compressed Oracle]
\label{thm:comp-oracle}
For any adversary $\adv$,
\[
  \Pr[\adv^{\mathsf{CStO}_\cD} = 1] = \Prr_{O \leftarrow \cD}[\adv^O = 1].
\]
\end{theorem}
\begin{proof} Follows from the proof of Lemma 4 in \cite{recording} with minimal modification.
\end{proof}
\begin{theorem}[One-way to hiding for compression unitaries]
\label{thm:ow2h-comp}
Let $\cD$ and $\cD'$ be product distributions where $\cD := \bigotimes_x \cD_x$ and $\cD' := \bigotimes_x \cD'_x$ and let $U$ and $U'$ 
be the corresponding compression unitaries.  For any $q$-query adversary $\adv$, 
let $\advB$ be the algorithm that samples $i \gets [q]$ and runs $\adv$ up to just before 
the $i$-th query, then measures the query register in the computational basis 
and returns the measurement output. Let $\Delta := \td(\ket{\psi}, \ket{\psi'})$ 
where $\ket{\psi}$ and $\ket{\psi'}$ are the final states of $\adv^U$ and $\adv^{U'}$ 
respectively. Then
\[
  \Exp_{x \leftarrow \cB^U} 
  \big[ \SD(\cD_x, \cD'_x) \big] 
  \geq\frac{\Delta^2}{16 q^2}.
\]
\end{theorem}

\begin{proof}
Consider the following sequence of hybrids $\{\cH_i\}_{i\in[0,q]}$, where $\cH_i$ consists of running $\adv$, answering the first $i$ queries using $U$, and answering the remaining queries using $U'$. 
Let the (purified) final state of $\adv$ in $H_i$ be $\ket{\psi_i}$. Note that $\ket{\psi_0} = \ket{\psi'}$ and $\ket{\psi_q} = \ket{\psi}$
By the triangle inequality,
\[
  \sum_{i=1}^q \td(\ket{\psi_i},\ket{\psi_{i-1}}) \geq\Delta.
\]
Define $\varepsilon_i := \| \ket{\psi_i} - \ket{\psi_{i-1}} \|$. Note that by Theorem \ref{thm:trace-dist-and-euclidean-dist}, $\sum_i \varepsilon_i \geq \Delta$.
The state just before the $i$-th query in $\cH_i$ and $\cH_{i-1}$ is identical since both hybrids are identical until the $i$-th query. Let this state be $\ket{\widetilde{\psi_i}}$. If we write the operation of $\adv$ between queries as the unitary $A$ and consider the initial state to be $\ket{0}$ without loss of generality then
\[
  \ket{\widetilde{\psi}_i} = (AU)^{i-1} A \ket{0}
\]
The final state in $\cH_i$ is $(AU')^{q-i} A (U \ket{\widetilde{\psi}_i})$ 
and in $H_{i-1}$ it is $(AU')^{q-i} A (U' \ket{\widetilde{\psi}_i})$ .  Thus,

\begin{align*}
  \varepsilon_i &= \|\ket{\psi_i} - \ket{\psi_{i-1}}\| \\ 
  &= \|(AU')^{q-i} A (U \ket{\widetilde{\psi}_i}) - (AU')^{q-i} A (U' \ket{\widetilde{\psi}_i})\|\\
  &= \|U \ket{\widetilde{\psi}_i} - U' \ket{\widetilde{\psi}_i}\|\\
  &= \|(U - U') \ket{\widetilde{\psi}_i}\|
\end{align*}

We can decompose $\ket{\widetilde{\psi}_i}$ into
\[
  \ket{\widetilde{\psi}_i} = \sum_x \sqrt{p_x}\ket{x}\ket{\phi_x}.
\]
where the first register is the query register and $p_x$ is the probability that measureing the query register in the computational basis returns $x$. Applying $U$ gives
\[
  U\ket{\widetilde{\psi}_i} = \sum_x \sqrt{p_x}\ket{x} \otimes U_x \ket{\phi_x},
\]
while applying $U'$ gives
\[
  U'\ket{\widetilde{\psi}_i} = \sum_x \sqrt{p_x}\ket{x} \otimes U'_x \ket{\phi_x}.
\]
Thus
\begin{align}
 \label{eq:ow2h-helper}   
  \varepsilon_i^2 = \|(U - U')\ket{\widetilde{\psi}_i}\|^2 
  = \sum_x p_x \| (U_x - U'_x)\ket{\phi_x} \|^2.
\end{align}
Let $\ket{\xi_x} := \ket{\cD_x} - \ket{\cD'_x}$ be an un-normalized state. By the definition of $U_x$ and $U'_x$
\begin{align*}
  U_x &= \ket{\cD_x}\!\bra{\bot} + \ket{\bot}\!\bra{\cD_x}
  + \big(\bbI - \ket{\bot}\!\bra{\bot} - \ket{\cD_x}\!\bra{\cD_x}\big)\\
 U'_x &= \ket{{\cD'}_x}\!\bra{\bot} + \ket{\bot}\!\bra{{\cD'}_x}
  + \big(\bbI - \ket{\bot}\!\bra{\bot} - \ket{{\cD'}_x}\!\bra{{\cD'}_x}\big)
\end{align*}
Taking their difference,
\begin{align*}
  U_x - U'_x &= \ket{\xi_x}\!\bra{\bot} + \ket{\bot}\!\bra{\xi_x} + (\ket{\cD'_x}\!\bra{\cD'_x} - \ket{\cD_x}\!\bra{\cD_x})\\
  &= \ket{\xi_x}\!\bra{\bot} + \ket{\bot}\!\bra{\xi_x} + \Big(\ket{\cD'_x}(\bra{\cD_x} - \bra{\xi_x}) - (\ket{\cD'_x} + \ket{\xi_x})\bra{\cD_x}\Big)\\
  &= \ket{\xi_x}\!\bra{\bot} + \ket{\bot}\!\bra{\xi_x} -\ket{\cD'_x}\!\bra{\xi_x}) - \ket{\xi_x}\!\bra{\cD_x}
\end{align*}
Applying the operator to $\ket{\phi_x}$ and noting that $\cD_x$ and $\cD'_x$ are not supported on $\bot$,
\begin{align*}
  \|(U_x - U'_x)\ket{\widetilde{\psi}_i}\| 
    &= \| \Big(\ket{\xi_x}\!\bra{\bot} + \ket{\bot}\!\bra{\xi_x} -\ket{\cD'_x}\!\bra{\xi_x}) - \ket{\xi_x}\!\bra{\cD_x}\Big)\ket{\phi_x}\| \\
    &=\| \Big(\sqrt{2}\cdot\ket{\xi_x}\tfrac{\bra{\bot} - \bra{\cD_x}}{\sqrt{2}}  + \sqrt{2}\cdot\tfrac{\ket{\bot} - \ket{\cD'_x}}{\sqrt{2}}\bra{\xi_x}\Big)\ket{\phi_x}\| \\
    &\leq2\sqrt{2} \|\ket{\xi_x}\| \\
    &\leq2\sqrt{2} \|\ket{\cD_x} - \ket{\cD'_x}\| \\
    &\leq4 \sqrt{\SD(\cD_x,\cD'_x)} \quad\text{(by Lemma \ref{lem:ED-to-SD})}
\end{align*}
Plugging this bound back into \eqref{eq:ow2h-helper}
\[\Rightarrow \varepsilon_i^2 \leq\sum_x p_x \cdot 16 \cdot \SD(\cD_x,\cD'_x)
   \leq16 \cdot \Exp_x[ \SD(\cD_x,\cD'_x)] .
\]
where $x$ is sampled by measuring the query register of $\ket{\widetilde{\psi}_i}$, which is exactly the distribution output by $\cB$ conditioned on sampling $i$.

\begin{align*}
\Exp_{z \leftarrow \cB} \left[ \SD(\cD_z,\cD'_z) \right]
   &\geq \Exp_{i \leftarrow [q]} \left[ \frac{\varepsilon_i^2}{16} \right] \\
   &\geq \frac{1}{16}\cdot \Exp_{i \leftarrow [q]}\left[\varepsilon_i \right]^2\\
   &=\frac{1}{16}\left( \frac{\sum_i \varepsilon_i}{q}\right)^2\\ 
   &=\frac{\Delta^2}{16q^2}.
\end{align*}

\end{proof}
\kabir{use actual BBBV}
\begin{corollary}\label{cor:ow2h-classical}
For any oracles $\cO, \cO'$,for any $q$-query adversary $\adv$, 
let $\advB$ be the algorithm that samples $i \gets [q]$ and runs $\adv$ up to just before 
the $i$-th query, then measures the query register in the computational basis 
and returns the measurement output. Let $\Delta := \td(\ket{\psi}, \ket{\psi'})$ 
where $\ket{\psi}$ and $\ket{\psi'}$ are the final states of $\adv^U$ and $\adv^{U'}$ 
respectively. Then
\[
  \Prr_{x \leftarrow \cB^U} 
  \big[ \cO(x) \neq \cO'(x)\big] 
  \geq\frac{\Delta^2}{16 q^2}.
\]
\end{corollary}
\subsection{Defining the Breaking Oracle}
Here we describe a set of oracles that can be defined alongside arbitrary oracles $O$ and which will allow an efficient classical prover to break the soundness of proofs of quantumness constructed \kabir{use S} using $O$.

\noindent For any oracle $O$ and any $v = (C_1^{(\cdot)}, s_1, \ldots, C_t^{(\cdot)}, s_t)$, define the state $\ket{\psi^O_v}$ via the following iterative process:
\begin{itemize}
\item Define $\ket{\phi_0} := \ket{0}$.  
\item For $j=1$ to $t$:  
\begin{itemize}
   \item $\Pi_j := \ket{s_j}\!\bra{s_j} \otimes \bbI$  
   \item $\ket{\widetilde{\phi_j}} := C_j^O \ket{\phi_{j-1}}$  
   \item $\ket{\phi_j} := \frac{\Pi_j \ket{\widetilde{\phi_j}}}{\|\Pi_j \ket{\widetilde{\phi_j}}\|}$  
\end{itemize}
\item $\ket{\psi^O_v} := \ket{\phi_t}$.  
\end{itemize}
Fix any constant $\ell \in \mathbb{N}$. For any set of oracles $S = (H,\chk,O)$, every input of form $(v,\sigma,C)$, for each $i \in [\ell]$, define the distribution $\cD^{S,i}_{(v,\sigma,C)}$ as:
\begin{itemize} 
\item If $\chk((v,i-1),\sigma)\ne 1$ output $\bot$.  If $i=1$ then skip this step.
\item Parse $v$ as $(C_1^{(\cdot)}, s_1, \ldots, C_{i-1}^{(\cdot)}, s_{i-1})$.
\item Prepare $\ket{\psi} := C^O \ket{\psi_v^O} = \sum_s \alpha_s \ket{s}\ket{\psi_s}$.  
\item Sample $s$ from the distribution obtained by measuring the first register of $\ket{\psi}$ in the computational basis.  
\item Return $(s, H(v,C,s,i))$.  
\end{itemize}
Define $\cD^{S,i} :=\bigotimes_{(v,\sigma,C)}\cD^{S,i}_{(v,\sigma,C)}$ and let $W^S_i$ be the corresponding compression unitary.  
We will use the shorthand $W^S$ to refer to the collection $\{W^S_i\}_{i\in [\ell]}$.  

\begin{lemma}
\label{lem:ow2h-multi}
For any $S$ and $S'$, let $A$ be a $q$-query adversary. Let 
\[
    \varepsilon := |\Pr[\cA^{W^S}=1] - \Pr[\cA^{W^{S'}}=1]|
\]
Let $\cB$ be the algorithm that samples $t \leftarrow [q]$, measures the query register of the $t$-th query of $\cA^{W^S}$, obtains result $x$, and outputs $(x,i)$ where the $t$-th query is made to oracle $W^S_i$.  
Then
\[
\Exp_{(x,i)\leftarrow \cB^{W^S}}\Big[ \SD(\cD_x^{S,i}, \cD_x^{S',i}) \Big] 
  \geq\frac{\varepsilon^2}{16q^2}.
\]
\end{lemma}

\begin{proof}
Let $U$ be the unitary that maps $\ket{i}\ket{\psi}$ to $\ket{i} \otimes W^S_i \ket{\psi}$ and 
let $U'$ be the unitary that maps $\ket{i}\ket{\psi}$ to $\ket{i} \otimes W^{S'}_i \ket{\psi}$.  
Any query to $W^S$ or $W^{S'}$ on state $\ket{\psi}$ can therefore be implemented by querying $U$ or $U'$ on $\ket{i}\ket{\psi}$.  
Additionally, $U$ and $U'$ are the corresponding compression unitaries of $\bigotimes_{i,x} \cD_x^{S,i}$ and $\bigotimes_{i,x} \cD_x^{S',i}$  
Putting these together, let $\widetilde{A}$ be the algorithm that runs $A$ and answers each $W^S$ or $W^{S'}$ query using $U$ or $U'$, and let $\widetilde{\cB}$ be the algorithm that samples $t \leftarrow [q]$, measures the query register of the $t$-th query of $\widetilde{A}$, and outputs the measurement result $(x,i)$.  
Since
\[
|\Pr[\widetilde{A}^U=1] - \Pr[\widetilde{A}^{U'}=1]| = \varepsilon,
\]
by Theorem \ref{thm:ow2h-comp}
\[
\Exp_{(x,i)\leftarrow \widetilde{\cB}^U}\Big[ \SD(\cD_x^{S,i}, \cD_x^{S',i}) \Big] 
  \geq\frac{\varepsilon^2}{16q^2}.
\]
The claim follows from the observation that $\widetilde{\cB}^U$ and $\cB^{W^S}$ have identical output distributions.
\end{proof}
\subsection{One-Way to Hiding for the Breaking Oracle}
\kabir{informal description/proof sketch}
\begin{theorem}
\label{thm:sep-core}
Let $(O,O',z)$ be joint random variables where $O$ and $O'$ represent oracles and $z$ represents some side information. Let $T$ be the set $\{x: O(x)\ne O'(x)\}$.  
Let $H$ be a random oracle from $\{0,1\}^* \to \{0,1\}^\lambda$.  
Define $\chk(x,y) := \ind\{H(x)=y \ne \bot\}$.  
Let $S := (H,\chk,O)$ and $S' := (H,\chk,O')$.  
Let $q \leq2^{o(\lambda)}$.  
Suppose for all $q$-query algorithms $\cB$,
\[
\Delta := \max_\cB \Big( \Prr_{O,O',z,H}[\cB^{W^S}(z) \in T] \Big)
\]
then for all $q$-query adversaries $\cA$ and for all measurements $\cM$ that are allowed to depend on $O,O',H,z$,
\[
\Big| \Prr_{O,O',H,z}[\cM(\cA^{W^S}(z))=1] - \Prr_{O,O',H,z}[\cM(\cA^{W^{S'}}(z)=1] \Big| 
   \leq\poly(q)\cdot \widetilde{\Delta}^{\frac{1}{2\cdot 4^\ell}},
\]
where $\widetilde{\Delta} := \max(\Delta, 1/2^{\lambda/2})$
\kabir{where W is ... from section...}
\end{theorem}

\begin{proof}
Let $\tau_0 := 0$ and let $\tau_i := \Delta^{1/4^i}$ for $i>0$.  
For each $i\in[\ell]$, define random variable $\mathsf{Bad}_i$ in terms of $O,O'$ as follows:
\[
\mathsf{Bad}_i := \{ v: v \text{ of form } (C_1,s_1,\ldots,C_i,s_i) \text{ and } \TD(\ket{\psi_v^O}\,\ket{\psi_v^{O'}})>\tau_i\}.
\]
For each $i\in[\ell]$, define $H_i$ in terms of $H,O,O'$ as
\[
H_i(x) := \begin{cases}
\bot & \text{if $x$ is of form $(v,j)$ for $j\leq i$ and $v \in \mathsf{Bad}_i$}, \\
H(x) & \text{otherwise}.
\end{cases}
\]
Also define $\chk_i(x,y) := \ind\{H_i(x)=y\ne \bot\}$.

\noindent For $t\in[0,\ell-1]$,  define the following random variables in terms of $H,O,O'$:  
\begin{itemize}
  \item $S_t := (H_t,\chk_t,O)$ 
  \item $S'_t := (H_{t+1},\chk_t,O)$  
  \item $S_\ell := (H_\ell,\chk_\ell,O)$  
\end{itemize}

\begin{claim}
\label{clm:sep-core-1}
For all $q$-query algorithms $\cA$ and for all measurements $\cM$ that may depend on $O,O',H,z$, for all $t\in[0,\ell-1]$,
\[
\Big| \Prr_{O,O',H,z}[\cM(\cA^{W^{S_t}}(z))=1] - \Prr_{O,O',H,z}[\cM(\cA^{W^{S_t'}}(z))=1] \Big| 
   \leq\frac{4\sqrt{2}q(4q\sqrt{\Delta'}+\tau_t)}{\tau_{t+1}},
\]
where $\Delta' := \max_{q\text{-query } \cB} \Prr_{O,O',H,z}[\cB^{W^{S_t}}(z)\in T]$
\end{claim}

\begin{proof}
For any fixing of $O,O',H,z$, let $\varepsilon$ be defined as
\[
\varepsilon := \Big|\Pr[\cM(\cA^{W^{S_t}}(z))=1] - \Pr[\cM(\cA^{W^{S_t'}}(z))=1]\Big|.
\]
Let $\cB(z)$ be the algorithm that samples $t \leftarrow [q]$, measures the query register of the $t$-th query of $\cA(z)$, obtains result $x$, and outputs $(x,i)$ where the $t$-th query is made to the $i$-th oracle. Therefore, by Lemma \ref{lem:ow2h-multi}
\[
\Exp_{(x,i)\leftarrow \cB^{W^{S_t}}(z)}[\SD(\cD_x^{S_t,i},\cD_x^{S_t',i})] \geq\frac{\varepsilon^2}{16q^2}.
\]
Let $\widetilde{\cB}(z)$ be the algorithm that runs $B^{W^{S_t}}(z)$ to obtain $(x,i)$, where $x$ is of form $(v,\sigma,C)$. The algorithm then samples from $\cD_{(v,\sigma,\widetilde{C})}^{S_t,i}$ where $\widetilde{C}$ is defined as the circuit that samples $j\leftarrow[q]$, runs $C^O$, measures the $j$-th query register, and outputs the value.  Note that this can be done efficiently by querying $W^{S_t}_i$.  

For any $x = (v,\sigma, C)$, consider $\SD\left(\cD_x^{S_t,i}, \cD_x^{S_t',i}\right)$. Let $E$ be the event that $i=t+1$ and $\chk_t((v,t),\sigma)=1$ and let $\bbS$ be the set of $s$ such that $(v,C,s)\in \mathsf{Bad}_{t+1}$. If $E$ does not occur then the statistical distance is zero. If $E$ occurs, the distributions only differ upon sampling $s \in \bbS$. Let $C^{O}\ket{\psi^O_v} = \sum_s \sqrt{p_s}\ket{s}\ket{\psi_s}$ for real, positive $p_s$ values. Then,
\[
    \delta := \SD(\cD_x^{S_t,i}, \cD_x^{S_t',i}) = \ind\{E\} \cdot \sum_{s\in \bbS} p_s
\]
Additionally, if $E$ occurs then $\chk_t((v,t),\sigma)=1$ and as a result $v\notin \mathsf{Bad}_t$. Therefore,
\[
\ind\{E\}\cdot \td(\ket{\psi^O_v}, \ket{\psi^{O'}_v}) \leq \tau_t.
\]
This implies that by the triangle inequality
\begin{align*}
\ind\{E\}\cdot\td(C^O\ket{\psi^O_v}, C^{O'}\ket{\psi^O_v}) &\geq \ind\{E\}\cdot\td(C^O\ket{\psi^O_v}, C^{O'}\ket{\psi^{O'}_v}) - \ind\{E\}\cdot\td(C^{O'}\ket{\psi^O_v}, C^{O'}\ket{\psi^{O'}_v})\\
&= \ind\{E\}\cdot\td(C^O\ket{\psi^O_v}, C^{O'}\ket{\psi^{O'}_v}) - \ind\{E\}\cdot\td(\ket{\psi^O_v}, \ket{\psi^{O'}_v})\\
&\geq \ind\{E\}\cdot\td(C^O\ket{\psi^O_v}, C^{O'}\ket{\psi^{O'}_v}) - \tau_t
\end{align*}
We will now bound the term on the RHS in terms of $\tau_{t+1}$ and $\delta$.  
First suppose $C^{O'}\ket{\psi^{O'}_v} = \sum_s \sqrt{p'_s} \ket{s}\ket{\psi'_s}$ for non-negative $p'_s$ values. For any real $\theta$,
\begin{align*} 
\|C^O\ket{\psi^O_v} - e^{i\theta}C^{O'}\ket{\psi^{O'}_v}\|^2 &= \left\|\sum_s \left(\sqrt{p_s}\ket{\psi_s} - e^{i\theta}\sqrt{p'_s}\ket{\psi'_s}\right)\ket{s}\right\|^2\\
&= \sum_s \left\|\sqrt{p_s}\ket{\psi_s} - e^{i\theta}\sqrt{p'_s}\ket{\psi'_s}\right\|^2,
\end{align*}
Now, for any $s$
\begin{align*} 
\left\| \sqrt{p_s}\,|\psi_s\rangle - e^{i\theta} \sqrt{p'_s}\,|\psi'_s\rangle \right\|^2 
&\geq p_s + p'_s - 2\sqrt{p_s p'_s}\, \big|\langle \psi_s | \psi'_s \rangle\big|\\
&= \left( \sqrt{p'_s} - \sqrt{p_s}\,|\langle \psi_s | \psi'_s \rangle| \right)^2 
+ p_s \left( 1 - |\langle \psi_s | \psi'_s \rangle|^2 \right)\\
&\geq p_s \left( 1 - |\langle \psi_s | \psi'_s \rangle|^2 \right)\\
& = p_s \cdot \TD\big(|\psi_s\rangle, |\psi'_s\rangle\big)^2
\end{align*}
which implies
\begin{align*}  
\|C^O\ket{\psi^O_v} - e^{i\theta}C^{O'}\ket{\psi^{O'}_v}\|^2 &\geq\sum_s p_s\cdot \td(\ket{\psi_s},\ket{\psi'_s})^2\\
&\geq \sum_{s\in\bbS} p_s\cdot \td(\ket{\psi_s},\ket{\psi'_s})^2\\
&\geq \sum_{s\in\bbS} p_s\cdot \tau_{t+1}^2
\end{align*}
Since $\theta$ was chosen to be arbitrary, by Lemma \ref{lem:thetaED-to-TD}
\[
\td(C^O\ket{\psi^O_v}, C^{O'}\ket{\psi^{O'}_v}) \geq\tau_{t+1}\sqrt{\sum_{s\in\bbS} p_s/2}
\]
which by the definition of $\delta$ gives
\[
\ind\{E\}\cdot\td(C^O\ket{\psi^O_v}, C^{O'}\ket{\psi^{O'}_v}) \geq\tau_{t+1}\sqrt{\delta/2}
\]
which implies
\[
\ind\{E\}\cdot\td(C^O\ket{\psi^O_v}, C^{O'}\ket{\psi^O_v}) \geq\tau_{t+1}\sqrt{\delta/2} - \tau_t
\]
Let $d := \max(0,\tau_{t+1}\sqrt{\delta/2} - \tau_t)$. Then by Corollary \ref{cor:ow2h-classical} 
\[
\Pr[\ind\{E\}\cdot\widetilde{C}^O\ket{\psi^O_v} \in T] \geq\frac{d^2}{16q^2}.
\]
where $\widetilde{C}$ is the circuit that samples $j\leftarrow[q]$, runs $C^O$, measures the $j$-th query register, and outputs the value.
Recall that $\widetilde{\cB}(z)$ samples $((v,\sigma, C), i) \leftarrow\cB^{W^{S_t}}(z)$ and then outputs a sample from $\cD_{(v,\sigma,\widetilde{C})}^{S_t,i}$. Now, by the definition of the distribution, the probability that a sample lies in $T$ is the probability that $\chk_t((v,i),\sigma) =1$ and $\widetilde{C}^O\ket{\psi_v^O} \in T$. Therefore
\begin{align*}
\Pr\left[ \widetilde{\cB}^{W^{S_t}}(z) \in T \right]
&= \Exp_{(v, \sigma, c), i \leftarrow \cB^{W^{S_t}}(z)}
\left[ \ind\{\chk_t((v,i),\sigma) =1\}\cdot\Pr\left[\widetilde{C}^O\ket{\psi_v^O} \in T\right] \right] \\
&\geq \Exp_{(v, \sigma, c), i \leftarrow \cB^{W^{S_t}}(z)}
\left[ \ind\{E\}\cdot\Pr\left[\widetilde{C}^O\ket{\psi_v^O} \in T\right] \right]\\
&\geq \Exp_{(v, \sigma, c), i \leftarrow \cB^{W^{S_t}}(z)}
\left[ \frac{d^2}{16q^2}\right]\\
&\geq \frac{1}{16q^2}\Exp_{(v, \sigma, c), i \leftarrow \cB^{W^{S_t}}(z)}
\left[ \max(0,\tau_{t+1}\sqrt{\delta/2} - \tau_t)^2\right]
\end{align*}
Observing that $\max(0,\tau_{t+1}\sqrt{\delta/2} - \tau_t)^2$ is convex in $\delta$ and applying Jensen's inequality,
\begin{align*}
   \Pr\left[ \widetilde{\cB}^{W^{S_t}}(z) \in T \right]
   &\geq 
 \frac{1}{16 q^2}
\left( \max\left(0,\tau_{t+1}\sqrt{\Exp_{(v, \sigma, c), i \leftarrow \cB^{W^{S_t}}(z)}[\delta/2]} - \tau_t\right)\right)^2\\
 &\geq \frac{1}{16 q^2}
\left( \max\left( 0, \frac{\varepsilon \tau_{t+1}}{4 \sqrt{2}q } - \tau_t \right) \right)^2 \\
&\geq \frac{1}{16 q^2}
\left(\frac{\varepsilon \tau_{t+1}}{4 \sqrt{2}q } - \tau_t\right)^2 
\end{align*}
Taking expectation over $O, O', H, z$ and applying Jensen's inequality:
\begin{align*} 
\Prr_{O,O',H,z}\left[ \widetilde{\cB}^{W^{S_t}}(z) \in T \right]
&\geq \frac{1}{16q^2}
\left(\frac{\Exp_{O,O',H,z}[\varepsilon] \tau_{t+1}}{4\sqrt{2} q } - \tau_t\right)^2 
\end{align*}
Noting that the LHS is at most $\Delta'$ and rearranging
\begin{align*}
\Exp_{O,O',H,z}[\varepsilon]
&\leq \frac{4\sqrt{2} q  \big( 4 q \sqrt{\Delta'} + \tau_t \big)}{\tau_{t+1}}.
\end{align*}
which concludes the proof of the claim.
\end{proof}

\begin{claim}
\label{clm:sep-core-2}
For all $q$-query algorithms $\cA$ and measurements $\cM$ that may depend on $O,O',H,z$, for all $t\in[0,\ell-1]$,
\[
\Big|\Prr_{O,O',H,z}[\cM(\cA^{W^{S'_t}}(z))=1] - \Prr_{O,O',H,z}[\cM(\cA^{W^{S_{t+1}}}(z))=1]\Big| \leq4q/\sqrt{2^\lambda}.
\]
\end{claim}

\begin{proof}
For any fixing of $O,O',H,z$, define
\[\varepsilon := \left|\Pr[\cM(\cA^{W^{S'_t}}(z))=1] - \Pr[\cM(\cA^{W^{S_{t+1}}}(z))=1]\right|\] 
Let $\cB(z)$ be the algorithm that samples $t \leftarrow [q]$, measures the query register of the $t$-th query of $\cA(z)$, obtains result $x$, and outputs $(x,i)$ where the $t$-th query is made to the $i$-th oracle. Therefore, by Lemma \ref{lem:ow2h-multi}
\[
\Exp_{(x,i)\leftarrow \cB^{W^{S_{t+1}}}(z)}[\SD(\cD_x^{S_t,i},\cD_x^{S_{t+1},i})] \geq\frac{\varepsilon^2}{16q^2}.
\]
For $x = (v,\sigma,C)$, $\cD_x^{S_t,i}$ and $\cD_x^{S_{t+1}^O,i}$ are identical unless the following conditions hold
\begin{itemize}
    \item $i = t+2$
    \item $\chk_t((v,t+1),\sigma)=1$
    \item $\chk_{t+1}((v,t+1),\sigma)\neq1$
\end{itemize}
in which case the distributions are at distance $1$ from each other. The latter two conditions are met only if $v\in\mathsf{Bad}_{t+1}$ and $H_{t+1}(v,t+1)=\sigma$.  
\[
\Prr_{(x,i)\leftarrow\cB^{W^{S_{t+1}}}(z)}[v\in\mathsf{Bad}_{t+1}\wedge H_{t+1}(v,t+1)=\sigma] \geq\frac{\varepsilon^2}{16q^2}.
\]
Taking expectation over $O,O',H, z$ and applying Jensen's inequality
\[
\Prr_{\substack{O,O',H,z\\(x,i)\leftarrow\cB^{W^{S_{t+1}}}(z)}}[v\in\mathsf{Bad}_{t+1}\wedge H_{t+1}(v,t+1)=\sigma] \geq\frac{\Exp_{O,O',H,z}[\varepsilon]^2}{16q^2}.
\]
But for all $v\in\mathsf{Bad}_{t+1}$, the value of $H(v,t+1)$ is sampled randomly independent of the view of $\cB$. Therefore
\[
\frac{1}{2^\lambda} \geq\frac{\Exp_{O,O',H,z}[\varepsilon]^2}{16q^2}.
\]
which implies
\[
\Exp(\varepsilon) \leq\frac{4q}{2^{\lambda/2}}
\]
concluding the proof of the claim.
\end{proof}
\begin{claim}
\label{clm:sep-core-combined}
For all $q$-query algorithms $\cA$ and for all measurements $\cM$ that may depend on $O,O',H,z$,
\begin{align*}
\Big|\Prr_{O,O',H,z}\big[\cM\left(\cA^{W^{S_0}}(z)\right)=1\big]-\Prr_{O,O',H,z}\big[\cM\left(\cA^{W^{S_\ell}}(z)\right)=1\big]\Big|
\leq\poly(q)\cdot \widetilde{\Delta}^{1/4^\ell}
\end{align*}
\end{claim}

\begin{proof}
For $i \in [\ell]$, let $\delta_i$ and $\Delta_i$ be defined as
\begin{align*}
\delta_i&:=\max_{\cA,\cM}\Big|\Prr_{O,O',H,z}\big[\cM(\cA^{W^{S_{i-1}}}(z))=1\big]-\Prr_{O,O',H,z}\big[M(\cA^{W^{S_{i}}}(z))=1\big]\Big|\\
\Delta_i&:=\max_{\substack{q\text{-query} \cB}}\Prr_{O,O',H,z}\big[\cB^{W^{S_{i}}}(z)\in T\big]
\end{align*}
and let $\Delta_0 := \Delta$.
First observe that by Claims \ref{clm:sep-core-1} and \ref{clm:sep-core-2},
\begin{align*}
\delta_i\le\frac{4\sqrt{2}q\big(4q\sqrt{\Delta_{i-1}}+\tau_{i-1}\big)}{\tau_i}+\frac{4q}{2^{\lambda/2}}.
\end{align*}
By setting $\cM$ to be the measurement that accepts inputs that belong to $T$, we can see that
\begin{align*}
\Delta_i\leq\Delta_{i-1}+\frac{4\sqrt{2}q\big(4q\sqrt{\Delta_{i-1}}+\tau_{i-1}\big)}{\tau_i}+\frac{4q}{2^{\lambda/2}}.
\end{align*}
Note that $\Delta_0=\Delta\leq\widetilde{\Delta}$. Additionally recall that $\tau_0=0$ and $\tau_i:=\widetilde{\Delta}^{1/4^i}$.
We prove by induction on $i$ that $\Delta_i\leq\poly(q)\cdot \widetilde{\Delta}^{1/4^i}$. For $i=0$ this holds trivially. Now suppose that $\Delta_{i-1}\leq\poly(q)\cdot \widetilde{\Delta}^{1/4^{i-1}}$. Then for large enough $\lambda$,
\begin{align*}
\Delta_i&\leq\Delta_{i-1}+\frac{\poly(q)\sqrt{\Delta_{i-1}}+\poly(q)\cdot \widetilde{\Delta}^{1/4^{i-1}}}{\widetilde{\Delta}^{1/4^i}}+\frac{\poly(q)}{2^{\lambda/2}}\\
&\leq\poly(q)\cdot \widetilde{\Delta}^{1/4^{i-1}}
+\frac{\poly(q)\cdot \widetilde{\Delta}^{2/4^{i-1}}+\poly(q)\cdot \widetilde{\Delta}^{4/4^i}}{\widetilde{\Delta}^{1/4^i}}
+\frac{\poly(q)}{2^{\lambda/2}}\\
&\leq \poly(q)\cdot \widetilde{\Delta}^{1/4^i}.
\end{align*}
which proves the induction hypothesis. As a result, we may now also bound
\begin{align*}
\delta_i
&\leq\frac{\poly(q)\cdot \widetilde{\Delta}^{2/4^i}+\poly(q)\cdot \widetilde{\Delta}^{4/4^i}}{\widetilde{\Delta}^{1/4^i}}
+\frac{\poly(q)}{2^{\lambda/2}}\\
&\leq \poly(q)\cdot \widetilde{\Delta}^{1/4^i}.
\end{align*}
Finally, we obtain the bound in the claim by adding up
\begin{align*}
\sum_{i=1}^{\ell}\delta_i&\leq\sum_{i=1}^{\ell}\poly(q)\cdot \widetilde{\Delta}^{1/4^i}\\
&\leq \poly(q)\cdot \widetilde{\Delta}^{1/4^\ell}
\end{align*}
which concludes the proof of the claim.
\end{proof}
We now have the necessary tools to complete the proof of Theorem \ref{thm:sep-core}. Suppose there exists a $q$-query adversary along with measurement $\cM$ such that
\begin{align*}
\Big|\Prr_{O,O',H,z}\big[\cM(\cA^{W^{S}}(z))=1\big]-\Prr_{O,O',H,z}\big[\cM(\cA^{W^{S'}}(z))=1\big]\Big|\geq d(\lambda)
\end{align*}
for some function $d(\cdot)$. For some fixing of $O,O',H$, let $\varepsilon$ be the random variable defined as
\begin{align*}
\varepsilon:=\Big|\Pr\big[\cM(\cA^{W^{S}}(z))=1\big]-\Pr\big[\cM(\cA^{W^{S'}}(z))=1\big]\Big|.
\end{align*}
Let $\cB(z)$ be the algorithm that samples $t \leftarrow [q]$, measures the query register of the $t$-th query of $\cA(z)$, obtains result $x$, and outputs $(x,i)$ where the $t$-th query is made to the $i$-th oracle. Therefore, by Lemma \ref{lem:ow2h-multi}
\begin{align*}
\Exp_{x,i\leftarrow \cB^{W^{S}}(z)}\left[\SD\left(\cD^{S,i}_{x},\cD^{S',i}_{x}\right)\right]\geq\frac{\varepsilon^{2}}{16q^{2}}
\end{align*}
and by Jensen’s inequality,
\begin{align*}
\Exp_{O,O',H,z}\Exp_{x,i\leftarrow \cB^{W^{S}}(z)}\left[\SD\left(\cD^{S,i}_{x},\cD^{S',i}_{x}\right)\right]
\geq \Exp_{O,O',H,z}\left[\frac{\varepsilon^{2}}{16q^{2}}\right]\geq\frac{d^{2}(\lambda)}{16q^{2}}.
\end{align*}
For $x$ of form $(v,\sigma,C)$, the process of sampling from $\cD^{S,i}_{x}$ first involves checking that $H(v,i-1)=\sigma$ and outputs $\bot$ if not. Then it samples from the measurement distribution of $C^{O}\ket{\psi_v^{O}}$. Sampling from $\cD^{S',i}_{x}$ is identical except it samples from $C^{O'}\ket{\psi_v^{O'}}$. Therefore, for any $x$ of form $(v,\sigma,C)$,
\begin{align*}
\SD\left(\cD^{S,i}_{x},\cD^{S',i}_{x}\right)
&\leq\ind\big\{H(v,i-1)=\sigma\big\}\cdot \TD\left(C^{O}\ket{\psi_v^{O}},C^{O'}\ket{\psi_v^{O'}}\right)\\
&\leq\ind\big\{H(v,i-1)=\sigma\big\}\cdot\Big(\TD\left(C^{O}\ket{\psi_v^{O}},C^{O'}\ket{\psi_v^{O}}\right)+\TD\left(C^{O'}\ket{\psi_v^{O}},C^{O'}\ket{\psi_v^{O'}}\right)\Big)\\
&= \ind\big\{H(v,i-1)=\sigma\big\}\cdot\Big(\TD\left(C^{O}\ket{\psi^O_v},C^{O'}\ket{\psi^O_v}\right)+\TD\left(\ket{\psi_v^{O}},\ket{\psi_v^{O'}}\right)\Big),
\end{align*}
where $\ind\big\{H(v,i-1)=\sigma\big\}$ is $0$ if $x$ is not of form $(v,\sigma,C)$. Define $\alpha_x$ and $\beta_x$ as 
\begin{align*}
    \alpha_x &:= \ind\big\{H(v,i-1)=\sigma\big\}\cdot \TD\left(C^{O}\ket{\psi_v^{O}},C^{O'}\ket{\psi_v^{O'}}\right)\\
    \beta_x &:=  \ind\big\{H(v,i-1)=\sigma\big\}\cdot \TD\left(\ket{\psi_v^{O}},\ket{\psi_v^{O'}}\right)
\end{align*} Hence
\begin{align*}
\Exp_{O,O',H,z}\Exp_{x,i\leftarrow \cB^{W^{S}}(z)}\left[\alpha_x+\beta_x\right]
\geq
\Exp_{O,O',H,z}\Exp_{x,i\leftarrow \cB^{W^{S}}(z)}\left[\SD\left(\cD^{S,i}_{x},\cD^{S',i}_{x}\right)\right]
\geq\frac{d^{2}}{16q^{2}}.
\end{align*}
Therefore it must be the case that either
\begin{align*}
\Exp_{O,O',H,z}\Exp_{x,i\leftarrow \cB^{W^{S}}(z)}\left[\alpha_x\right]\ge\frac{d^{2}}{32q^{2}}
\quad\text{or}\quad
\Exp_{O,O',H,z}\Exp_{x,i\leftarrow \cB^{W^{S}}(z)}\left[\beta_x\right]\ge\frac{d^{2}}{32q^{2}}.
\end{align*}
First consider the case when the former holds, i.e., $\Exp_{O,O',H,z}\Exp_{x,i\leftarrow \cB^{W^{S}}(z)}\left[\alpha_x\right]\ge\frac{d^{2}}{32q^{2}}$.
Let $\widetilde{\cB}(z)$ be the algorithm that runs $\cB^{W^{S}}(z)$ to obtain $(x,i)$, where $x$ is of form $(v,\sigma,C)$, and then samples from $\cD^{S,i}_{(v,\sigma,\widetilde{C})}$, where $\widetilde{C}$ is the circuit that samples $j\leftarrow[q]$, runs $C$, measures the $j$-th query register, and outputs the value. Note that this can be done efficiently by querying $W^{S}_i$. Then
\begin{align*}
\frac{d^{2}}{32q^{2}}
&\leq\Exp_{O,O',H,z}\Exp_{x,i\leftarrow \cB^{W^{S}}(z)}\left[\ind\big\{H(v,i-1)=\sigma\big\}\cdot \TD\left(C^{O}\ket{\psi_v^{O}},C^{O'}\ket{\psi_v^{O'}}\right)\right]\\
&\leq\Exp_{O,O',H,z}\Exp_{x,i\leftarrow \cB^{W^{S}}(z)}\left[\ind\big\{H(v,i-1)=\sigma\big\}\cdot 4q\sqrt{\Pr\big[\widetilde{C}^O\ket{\psi_v^{O}}\in T\big]}\right]\\
&=4q\Exp_{O,O',H,z}\left(\Exp_{x,i\leftarrow \cB^{W^{S}}(z)}\left[\sqrt{\ind\big\{H(v,i-1)=\sigma\big\}\cdot \Pr\big[\widetilde{C}^O\ket{\psi_v^{O}}\in T\big]}\right]\right)\\
&\leq4q\Exp_{O,O',H,z}\left(\sqrt{\Exp_{x,i\leftarrow \cB^{W^{S}}(z)}\left[\ind\big\{H(v,i-1)=\sigma\big\}\cdot \Pr\big[\widetilde{C}^O\ket{\psi_v^{O}}\in T\big]\right]}\right)\\
&\leq4q,\sqrt{\Prr_{O,O',H,z}\left[\widetilde{\cB}^{W^{S}}\in T\right]},
\end{align*}
where the second step uses Corollary \ref{cor:ow2h-classical}  and the fourth step uses Jensen's inequality. By the definition of $\Delta$ this implies $d\leq 8q\sqrt{2}\Delta^{1/4} \leq 8q\sqrt{2}\widetilde{\Delta}^{1/4}$.

If instead it is the case that $\Exp_{O,O',H,z}\Exp_{x,i\leftarrow \cB^{W^{S}}(z)}\left[\beta_x\right]\ge\frac{d^{2}}{32q^{2}}$, i.e.,
\begin{align*}
\Exp_{O,O',H,z}\Exp_{x,i\leftarrow \cB^{W^{S}}(z)}\left[\ind\big\{H(v,i-1)=\sigma\big\}\cdot \TD\left(\ket{\psi_v^{O}},\ket{\psi_v^{O'}}\right)\right]\ge\frac{d^{2}}{32q^{2}},
\end{align*}
then for any fixing of $O,O',H,z$, let the random variable $\varepsilon$ be defined as
\begin{align*}
\varepsilon:=\Exp_{x,i\leftarrow \cB^{W^{S}}(z)}\left[\ind\big\{H(v,i-1)=\sigma\big\}\cdot \TD\left(\ket{\psi_v^{O}},\ket{\psi_v^{O'}}\right)\right].
\end{align*}
Let $\bbX_{H,O,O'}$ be the set defined as
\begin{align*}
\bbX_{H,O,O'}:=\left\{(x,i): H(v,i-1)=\sigma\ \text{and}\ \TD\left(\ket{\psi_v^{O}},\ket{\psi_v^{O'}}\right)\geq\tau_\ell\right\}
\end{align*}
and let $\cM:=\sum_{(x,i)\in \bbX_{H,O,O'}}\ket{x,i}\bra{x,i}$. Then by a Markov argument,
\begin{align*}
\Pr\big[\cM(\cB^{W^{S}}(z))=1\big]\ge\varepsilon-\tau_\ell
\end{align*}
Taking expectation,
\begin{align*}
\Prr_{O,O',H,z}\big[\cM(\cB^{W^{S}}(z))=1\big]\geq\frac{d^{2}}{32q^{2}}-\tau_\ell,
\end{align*}
Now by Claim \ref{clm:sep-core-combined} and noting that $S_0 = S$
\begin{align*}
\Prr_{O,O',H,z}\big[\cM(\cB^{W^{S_\ell}}(z))=1\big]\geq\frac{d^{2}}{32q^{2}}-\poly(q)\cdot \widetilde{\Delta}^{1/4^\ell}-\tau_\ell
\ge\frac{d^{2}}{32q^{2}}-\poly(q)\cdot \widetilde{\Delta}^{1/4^\ell}.
\end{align*}
However, note $H(v,i-1)$ for $(x,i)\in \bbX_{H,O,O'}$ is sampled independently at random. Therefore,
\begin{align*}
\frac{1}{2^{\lambda}}\ge\frac{d^{2}}{32q^{2}}-\poly(q)\cdot \widetilde{\Delta}^{1/4^\ell}
\end{align*}
which implies that
\begin{align*}
d\leq \poly(q)\cdot \widetilde{\Delta}^{1/(2\cdot 4^{\ell})}.
\end{align*}
which concludes the proof of the theorem.
\end{proof}
\subsection{Ruling Out Fully Black-Box Constructions}
\label{pdqp-subsec}
\begin{definition}[Fully Black-Box Construction of $\ell$-Round IV-PoQ from iO and OWPs]
    For any constant $\ell \in\bbN$, a fully black-box construction of an $\ell$-round ($c(\lambda), s(\lambda)$)-IV-PoQ from iO and OWPs consists of an oracle-aided QPT prover $\cP$ and an oracle aided PPT (with unbounded final step) verifier $\cV$ along with an oracle-aided $q(\lambda)$-query reduction $R$ for some polynomial $q$ and some functions $c(\lambda) > s(\lambda)$. The algorithms must satisfy the following:
    \begin{itemize}
        \item \textbf{Correctness:} For all large enough $\lambda\in\bbN$, for any permutation $f:\bin^\lambda \rightarrow\bin^\lambda$, and for any pair of functions $(\obf,\eval)$ such that $\eval(\obf(C),x) = C{^f}(x)$ for all classical circuits $C$, 
        \[
            \Pr[\langle \cP^{f,\obf,\eval},\cV^{f,\obf,\eval}\rangle(1^\lambda) = 1] \geq c(\lambda)
        \]
        \item \textbf{Black-Box Security Proof:} For any family of permutations $f:\bin^\lambda \rightarrow\bin^\lambda$, and for any pair of functions $(\obf,\eval)$ such that $\eval(\obf(C),x) = C{^f}(x)$ for all classical circuits $C$, 
        for any polynomial $p$, and for any classical algorithm $A$, if for infinitely many $\lambda$,
        \[
            \Pr[\langle A,\cV^{f,\obf,\eval}\rangle(1^\lambda) = 1] \geq c(\lambda)
        \]
       then there exist polynomials $p_1$ and $p_2$ such that for infinitely many $\lambda$,
        \begin{align*}
            \Prr_{x\leftarrow\bin^\lambda} \left[ R^{A,f,\obf,\eval} (f(x))=x \right] \ge \frac{1}{p_1(\lambda)}
        \end{align*}
        or 
        \begin{align*}
            \left| \Pr\left[ \cG^{\mathsf{iO}}_{0,\lambda}(R^{A,f,\obf,\eval})=1\right]-\Pr\left[\cG^{\mathsf{iO}}_{1,\lambda}(R^{A,f,\obf,\eval})=1 \right] \right| \ge \frac{1}{p_2(\lambda)}
        \end{align*}
        where $\cG^{\mathsf{iO}}_{b,\lambda}$ is the iO security game defined as follows for any adversary $\cA$:
        \begin{itemize}
            \item $\cG^{\mathsf{iO}}_{b,\lambda}(\cA)$:
        \begin{itemize}
            \item $(C_0,C_1)\leftarrow \cA^{f,\obf,\eval}(1^\lambda)$
            \item If $C_0$ is not functionally equivalent to $C_1$, or $C_1\notin\{0,1\}^{\lambda}$, or $C_0\notin\{0,1\}^{\lambda}$, return $\bot$
            \item $r\leftarrow\{0,1\}^\lambda$
            \item Return $\cA^{f,\obf,\eval}\big(\obf(C_b,r)\big)$
            \end{itemize}
        \end{itemize}
    \end{itemize}
\end{definition}

\noindent To show the impossibility of black-box constructions, we define a set of oracles $P$ such that 
\begin{itemize}
    \item we can instantiate $f, \eval, \obf$ that satisfy the correctness requirements using $P$,
    \item no query-bounded adversary can break the security of $f$ or $(\eval, \obf)$,
    \item and there exists a query-bounded classical algorithm $A$ such that
        \[
            \Pr[\langle A,V^{f,\obf,\eval}\rangle(1^\lambda) = 1] \geq c(\lambda)
        \] for all large enough $\lambda$.
\end{itemize}

\noindent Formally, we define a set of  oracles\footnote{Note that although we sample the oracles randomly, we sample the whole truth table at once and fix the oracle. i.e., the oracle is not a \textit{randomized} oracle. \kabir{this may be confusing} It is in fact possible to show the existence of some fixed oracle $P$ for which the result holds, but it is standard to omit this as it is unnecessary for fully black-box separations.}:
\begin{itemize}
\item Define $f$ such that for all $\lambda$, for all $x\in\bin^{\lambda}$, $f(x)=f_\lambda(x)$ where $f_\lambda$ is a random permutation on $\lambda$-bit strings.
\item Define $\obf$ such that for all $\lambda$, for all $C\in\bin^{\lambda}$, for all $r\in\bin^{\lambda}$, $\obf(C,r)=\obf_\lambda(C,r)$ where $\obf_\lambda : \bin^\lambda \times \bin^\lambda \rightarrow \bin^{3\lambda}$
is a random injective function.
\item For any oracle $O$, define $\eval^{O}$ that on input $(\obfC, z)$ performs the following:
\begin{itemize}
\item If $\obfC \notin \mathsf{Image}(\obf)$, return $\bot$.
\item Let $(C,r):=\obf^{-1}(\obfC)$. Parse $C$ as a (classical) circuit.
\item Return $C^{O}(z)$.
\end{itemize}
\item For each $\lambda$, let $H_\lambda$ be a random function from $\bin^*$ to $\bin^{\lambda}$ and define $\chk_\lambda(x,y):=\ind\{H_\lambda(x)=y\neq\bot\}$.
\item For each $\lambda$, for each $i\in[\ell]$, let $P_{i,\lambda}$ be drawn from $\cD^{S_{\lambda},i}$ where $S_\lambda$ is defined as $\big((f,\obf,\eval^{\obf,f}),H_\lambda,\chk_\lambda\big)$. We use the shorthand $P_\lambda$ to refer to $\{P_{i,\lambda}\}_{i\in[\ell]}$.
\end{itemize}


\subsubsection{One-Way Permutation Security
}\begin{theorem}[One-Way Permutation Security]
\label{thm:owp}
For all $q\leq2^{o(\lambda)}$, for all $q$-query adversaries $\cA$, and for large enough $\lambda$,
\begin{align*}
\Prr_{P_\lambda,x\leftarrow\bin^{\lambda}}\Big[\cA^{P_\lambda}\big(f(x)\big)=x\Big]\le\frac{1}{2^{\lambda/4^{2(\ell+2)}}}
\end{align*}
\end{theorem}

\begin{proof}
We drop $\lambda$ from subscripts when clear by context. We first note that we may replace the oracle $P_\lambda$ with its corresponding compressed-oracle simulation. Additionally, the compressed oracle may be implemented by two queries to the corresponding compression unitaries $W^{S}=\{W^{S}_i\}_{i\in[\ell]}$. Let $\widetilde{\cA}$ be the $2q$-query algorithm that runs $\cA$ by answering its queries as above using $W^{S}$. Then
\begin{align*}
\Prr_{S, x\leftarrow\bin^{\lambda}}\Big[\widetilde{\cA}^{W^{S}}\big(f(x)\big)=x\Big]
=\Prr_{S, x\leftarrow\bin^{\lambda}}\Big[\cA^{P_\lambda}\big(f(x)\big)=x\Big].
\end{align*}
We will now show that for any fixing of $\obf$, for any $q \leq 2^{o(\lambda)}$, for any $q$-query adversary $\cB$,
\begin{align*}
\Prr_{f,H,x}\Big[\cB^{W^{S}}\big(f(x)\big)=x\Big]\le\frac{1}{2^{\lambda/4^{2(\ell+2)}}}
\end{align*}
This suffices to show that
\begin{align*}
\Prr_{S,x}\Big[\widetilde{\cA}^{W^{S}}\big(f(x)\big)=x\Big]\le\frac{1}{2^{\lambda/4^{2(\ell+2)}}}
\end{align*}
For the remainder of the proof, fix any $\obf$.

\begin{claim}
\label{clm:owp-find}
For any $f',f''$ and $y$ such that $f'(y)\ne f''(y)$ or $\eval^{\obf,f'}(y)\ne \eval^{\obf,f''}(y)$,
\begin{align*}
\Pr\Big[f'(y')\ne f''(y')\Big|y'\leftarrow \mathsf{Find}^{f'}(y)\Big]\ge\frac{1}{2|y|},
\end{align*}
where $\mathsf{Find}^{f'}(y)$ does the following:
\begin{itemize}
\item With probability $1/2$ output $y$.
\item With probability $1/2$, interpret $y$ as $(\obfC, z)$ and do the following:
\begin{itemize}
\item Compute $(C,\_)=\obf^{-1}(\obfC)$.
\item Run $C^{\obf,f'}(z)$ and keep track of the $f'$-queries made by $C$.
\item Return a random query.
\end{itemize}
\end{itemize}
\end{claim}

\begin{proof}
If $\eval^{\obf,f'}(y)= \eval^{\obf,f''}(y)$ then it must be the case that $f'(y) \neq f''(y)$ already.
If $\eval^{\obf,f'}(y)\ne \eval^{\obf,f''}(y)$, then $C^{\obf,f'}$ must query some $y'$ such that $f'(y')\ne f''(y')$. The claim follows from observing that $C$ is a circuit of size less than $|y|$ and so makes less than $|y|$ queries.
\end{proof}
\noindent For any set $\bbX$ and all $x^*\in\{0,1\}^*$, we define
\[
f_{\bbX}(x^*)=
\begin{cases}
f(x^*) & \text{if } x^*\notin\bbX,\\
\bot & \text{if } x^*\in\bbX.
\end{cases}
\]
Let $x'$ be sampled uniformly from $\{0,1\}^\lambda$. Define the random variables
\[
\cO_1:=(f_{\{x'\}},\obf,\eval^{\obf,f_{\{x'\}}}),\quad S_1:=(\cO_1,H,\chk),
\]
\[
\cO_2:=(f_{\{x,x'\}},\obf,\eval^{\obf,f_{\{x,x'\}}}),\quad S_2:=(\cO_2,H,\chk),
\]
\[
\cO_3:=(f_{\{x\}},\obf,\eval^{\obf,f_{\{x\}}}),\quad S_3:=(\cO_3,H,\chk),
\]
\[
\cO_4:=(f,\obf,\eval^{\obf,f}),\quad S_4:=(\cO_4,H,\chk).
\]
Note that the random variables are correlated with the choice of $x$ and $x'$ which are both sampled uniformly at random as stated above.

\begin{claim}
\label{clm:owp-1}
For all adversaries $\cB$,
\[
\Prr_{S,x,x'}\big[\cB^{W^{S_1}}(f(x'))=x\big]\le \frac{1}{2^\lambda}.
\]
\end{claim}
\begin{proof}
Follows from the observation that $x$ is sampled independent of the adversary’s view.
\end{proof}

\begin{claim}
\label{clm:owp-2}
For all $q\le 2^{o(\lambda)}$, all $q$-query adversaries $\cB$, and large enough $\lambda$,
\[
\Prr_{S,x,x'}\big[\cB^{W^{S_2}}(f(x))=x\big]\le \frac{\poly(q)}{2^{\lambda/4^{\ell+1}}}.
\]
\end{claim}
\begin{proof}
Let $T:=\{y:\cO_1(y)\ne \cO_2(y)\}$. Then for all $q$-query adversaries $\cB$,
\[
\max_{\cB}\Prr_{S,x,x'}\big[\cB^{W^{S_1}}(f(x'))\in T\big]\le \frac{2q}{2^\lambda}
\]
If $\cB$ could exceed this bound, then by Claim \ref{clm:owp-find} running $\mathsf{Find}^{f_{\{x'\}}}$ on the output would result in $x$ with probability $>1/2^\lambda$. Since a call to $\mathsf{Find}^{f_{\{x'\}}}$ on an input of size at most $q$ can be implemented with at most $q$ calls to $W^{S_1}$, this gives a $2q$ query algorithm that finds x with probability greater than $1/2^\lambda$, contradicting Claim \ref{clm:owp-1}. Let $\cM$ be the measurement that measures in computational basis and accepts if the output is  $x$. By Theorem \ref{thm:sep-core} applied to $\cM,\cO_1,\cO_2$, for all $q$-query adversaries $\cB$,
\[
\Big|\Prr_{S,x,x'}\big[\cB^{W^{S_1}}(f(x'))=x\big]-\Prr_{S,x,x'}\big[\cB^{W^{S_2}}(f(x'))=x\big]\Big|\leq \frac{\poly(q)}{2^{\lambda/4^{\ell+1}}},
\]
which along with Claim \ref{clm:owp-1} and the fact that $q\le 2^{o(\lambda)}$ implies the claim.
\end{proof}

\begin{claim}
\label{clm:owp-3}
For all $q\le 2^{o(\lambda)}$, all $q$-query adversaries $\cB$, and large enough $\lambda$,
\[
\Prr_{S,x,x'}\big[\cB^{W^{S_2}}(f(x'))=x'\big]=\Prr_{S,x',x}\big[\cB^{W^{S_2}}(f(x))=x\big]\leq \frac{\poly(q)}{2^{\lambda/4^{\ell+1}}}.
\]
\end{claim}
\begin{proof}
Follows from Claim \ref{clm:owp-2} and the fact that $x$ and $x'$ are symmetric in the view of $\cB$.
\end{proof}

\begin{claim}
\label{clm:owp-4}
For all adversaries $\cB$,
\[
\Prr_{S,x,x'}\big[\cB^{W^{S_3}}(f(x))=x'\big]\le \frac{1}{2^\lambda}
\]
\end{claim}
\begin{proof}
Follows from the observation that $x'$ is sampled independent of the adversary’s view.
\end{proof}

\begin{claim}
\label{clm:owp-5}
For all $q\le 2^{o(\lambda)}$, all $q$-query adversaries $\cB$, and large enough $\lambda$,
\[
\Prr_{S,x,x'}\big[\cB^{W^{S_3}}(f(x))=x\big]\leq \frac{\poly(q)}{2^{\lambda/4^{\ell+1}}}.
\]
\end{claim}
\begin{proof}
Let $T:=\{y:\cO_2(y)\ne\cO_3(y)\}$. Then for all $q$-query adversaries $\cB$,
\[
\max_{\cB}\Prr_{S,x,x'}\big[\cB^{W^{S_3}}(f(x))\in T\big]\le \frac{2q}{2^\lambda}
\]
If $\cB$ could exceed this bound, then by Claim \ref{clm:owp-find} running $\mathsf{Find}^{f_{\{x\}}}$ on the output would result in $x'$ with probability $>1/2^\lambda$. Since a call to $\mathsf{Find}^{f_{\{x\}}}$ on an input of size at most $q$ can be implemented with at most $q$ calls to $W^{S_3}$, this gives a $2q$ query algorithm that finds x' with probability greater than $1/2^\lambda$, contradicting Claim \ref{clm:owp-4}. Let $\cM$ be the measurement that measures in computational basis and accepts if the output is  $x$. By Theorem \ref{thm:sep-core} applied to $\cM,\cO_2,\cO_3$, for all $q$-query adversaries $\cB$,
\[
\Big|\Prr_{S,x,x'}\big[\cB^{W^{S_2}}(f(x))=x\big]-\Prr_{S,x,x'}\big[\cB^{W^{S_3}}(f(x))=x\big]\Big|\leq \frac{\poly(q)}{2^{\lambda/4^{\ell+1}}},
\]
which along with Claim \ref{clm:owp-3} and the fact that $q\le 2^{o(\lambda)}$ implies the claim.
\end{proof}

\begin{claim}
\label{clm:owp-6}
For all $q\le 2^{o(\lambda)}$, all $q$-query adversaries $\cB$, and large enough $\lambda$,
\[
\Prr_{S,x,x'}\big[\cB^{W^{S_4}}(f(x))=x\big]\le \frac{\poly(q)}{2^{\lambda/ 4^{2\ell+3}}}.
\]
\end{claim}
\begin{proof}
Let $T:=\{y:\cO_3(y)\ne\cO_4(y)\}$. Then for all $q$-query adversaries $\cB$,
\[
\max_{\cB}\Prr_{S,x,x'}\big[\cB^{W^{S_3}}(f(x))\in T\big]\le \frac{1}{2^{\lambda /4^{\ell+2}}}
\]
If $\cB$ could exceed the bound, then by Claim \ref{clm:owp-find} running $\mathsf{Find}^{f_{\{x\}}}$ on the output would result in $x$ with probability $>\frac{1}{2q\cdot2^{\lambda /4^{\ell+2}}}$. Since a call to $\mathsf{Find}^{f_{\{x\}}}$ on an input of size at most $q$ can be implemented with at most $q$ calls to $W^{S_3}$, this gives a $2q$ query algorithm that finds $x$ with probability greater than $\frac{1}{2q\cdot2^{\lambda /4^{\ell+2}}}$, which contradicts Claim \ref{clm:owp-5} since $q \leq 2^{o(\lambda)}$. Let $\cM$ be the measurement that measures in computational basis and accepts if the output is  $x$. By Theorem \ref{thm:sep-core} applied to $\cM,\cO_3,\cO_4$, for all $q$-query adversaries $\cB$,
\[
\Big|\Prr_{S,x,x'}\big[\cB^{W^{S_3}}(f(x))=x\big]-\Prr_{S,x,x'}\big[\cB^{W^{S_4}}(f(x))=x\big]\Big|\leq \frac{\poly(q)}{2^{\lambda/4^{2\ell+3}}},
\]
which along with Claim \ref{clm:owp-5} and the fact that $q\le 2^{o(\lambda)}$ implies the claim.
\end{proof}
\noindent Since Since $S_4=S$ and $q\le 2^{o(\lambda)}$, the theorem follows from Claim \ref{clm:owp-6}.
\end{proof}

\subsubsection{Indistinguishability Obfuscation Security}
Define $\cG_{b,\lambda}(\cA)$ as follows:
\begin{itemize}
\item $(C_0,C_1)\leftarrow \cA^{P_\lambda}(1^\lambda)$
\item If $C_0$ is not functionally equivalent to $C_1$, or $C_1\notin\{0,1\}^{\lambda}$, or $C_0\notin\{0,1\}^{\lambda}$, return $\bot$
\item $r\leftarrow\{0,1\}^\lambda$
\item Return $\cA^{P_\lambda}\big(\obf(C_b,r)\big)$
\end{itemize}

\begin{theorem}\label{thm:io-security}
For all $q\le 2^{o(\lambda)}$, all $q$-query adversaries $\cA$, and large enough $\lambda$,
\[
\Big|\Prr_{P_\lambda}\big[\cG_{0,\lambda}(\cA)=1\big]-\Prr_{P_\lambda}\big[\cG_{1,\lambda}(\cA)=1\big]\Big|\leq \frac{1}{2^{\lambda/4^{2(\ell+2)}}}.
\]
\end{theorem}

\begin{proof}
We drop $\lambda$ from subscripts when clear. We first note that we can replace $P_\lambda$ with its corresponding compressed-oracle simulation. The compressed oracle may be implemented by two queries to compression unitaries $W^S=\{W_i^S\}_{i\in[\ell]}$. Let $\widetilde{\cA}$ be the $2q$-query algorithm that runs $\cA$ by answering queries using $W^S$. By Theorem \ref{thm:comp-oracle}, replacing $\cA^{P_\lambda}$ with $\widetilde{\cA}^{W^S}$ does not affect the adversary’s advantage in distinguishing the games.

For $b\in\{0,1\}$, define game $\cG'_{b,\lambda}(\widetilde{\cA})$ as:
\begin{itemize}
\item $(C_0,C_1)\leftarrow \widetilde{\cA}^{W^S}(1^\lambda)$
\item If $C_0$ is not functionally equivalent to $C_1$, or $C_1\notin\{0,1\}^{\lambda}$, or $C_0\notin\{0,1\}^{\lambda}$, return $\bot$
\item $r\leftarrow\{0,1\}^\lambda$
\item Return $\widetilde{\cA}^{W^S}\big(\obf(C_b,r)\big)$
\end{itemize}
Then for $b\in\{0,1\}$,
\[
\Prr_{r}\big[\cG'_{b,\lambda}(\widetilde{\cA})=1\big]=\Prr_{r}\big[\cG_{b,\lambda}(\cA)=1\big].
\]
\noindent For any set $\bbX$ and all $x^*\in\{0,1\}^*$, we define
\[
    \obf_{\bbX}(x^*)=
    \begin{cases}
    \obf(x^*) & \text{if } x^*\notin\bbX,\\
    \bot & \text{if } x^*\in\bbX.
    \end{cases}
\]
Also for any $r^* \in \bin^\lambda$ define the set 
\[
    (*,r^*) : = \{(C,r^*): C\in\bin^\lambda\}
\]
Next, for all $r^*\in\{0,1\}^\lambda$ define
\[
\cO_{r^*}:=\big(f,\obf_{(*,r^*)},\eval^{f,\obf_{(*,r^*)}}\big),\qquad S_{r^*}:=(\cO_{r*},H,\chk),
\]
Note that $\eval^{f,\obf_{(*,r^*)}}$ continues to use $\obf^{-1}$ to compute $C$, and the change in oracles only affects oracle queries made by $C$.
For all $b\in\{0,1\}$, define $\cG''_{b,\lambda}(\cA)$ as:
\begin{itemize}
\item $r\leftarrow\{0,1\}^\lambda$
\item $(C_0,C_1)\leftarrow \widetilde{\cA}^{W^{S_r}}(1^\lambda)$,
\item If $C_0$ is not functionally equivalent to $C_1$, or $C_1\notin\{0,1\}^{\lambda}$, or $C_0\notin\{0,1\}^{\lambda}$, return $\bot$
\item Return $\widetilde{\cA}^{W^{S_{r}}}\big(\obf(C_b,r)\big)$.
\end{itemize}

\begin{claim}
\label{clm:obf-1}
For all $\cA$, $\Prr\big[\cG''_{0,\lambda}(\cA)=1\big]=\Prr\big[\cG''_{1,\lambda}(\cA)=1\big]$.
\end{claim}
\begin{proof}
In $\cA$’s view, $\obf(C_0,r)$ and $\obf(C_1,r)$ are sampled under the constraint that they do not collide with $\obf(C,r')$ for any $r'\ne r$ or any $C\notin\{C_0,C_1\}$, and under the constraint that
\[
\eval^{f,\obf_{(*,r)}}\big(\obf(C_0,r),x\big)=\eval^{f,\obf_{(*,r)}}\big(\obf(C_1,r),x\big)=C_0^{f,\obf_{(*,r)}}(x).
\]
By symmetry both games are identical.
\end{proof}

By Claim \ref{clm:obf-1} it suffices to show that for all $b\in\{0,1\}$, all $q\leq 2^{o(\lambda)}$, all $q$-query adversaries $\cA$, and large enough $\lambda$,
\[
\Big|\Prr\big[\cG'_{b,\lambda}(\cA)=1\big]-\Prr\big[\cG''_{b,\lambda}(\cA)=1\big]\Big|\le \frac{1}{2}\cdot\frac{1}{2^{\lambda/4^{2(\ell+2)}}}.
\]
To prove this, define a sequence of random variables as follows. Let $r'\leftarrow\{0,1\}^\lambda$ and $r\leftarrow\{0,1\}^\lambda$. Define
\[
\cO_1:=(f,\obf_{(*,r')},\eval^{f,\obf_{(*,r')}}),\quad S_1:=(\cO_1,H,\chk),
\]
\[
\cO_2:=(f,\obf_{(*,r') \cup (*,r)},\eval^{f,\obf_{(*,r') \cup (*,r)}}) ,\quad S_2:=(\cO_2,H,\chk),
\]
\[
\cO_3:=(f,\obf_{(*,r)},\eval^{f,\obf_{(*,r)}}),\quad S_3:=(\cO_3,H,\chk),
\]
\[
\cO_4:=(f,\obf,\eval^{f,\obf}),\quad S_4:=(\cO_4,H,\chk).
\]
Let $\mathsf{Rev}$ be the truth table for the algorithm that on input $y = (\obfC, z)$ computes $(C,\_)=\obf^{-1}(\obfC)$ and returns $C$. For any $r^*\in\bin^\lambda$ let $\obf(*,r^*):= \{\obf(C,r^*): C \in \bin^\lambda\}$. Note that the random variables are correlated with $r$ and $r'$ which are both sampled uniformly at random as specified above.

\begin{claim}
\label{clm:obf-2}
For all adversaries $\cB$,
\[
\Prr_{r,r',s}\big[\cB^{W^{S_1}}(\obf(*,r'),\mathsf{Rev})=r\big]\le \frac{1}{2^\lambda}
\]
\end{claim}
\begin{proof}
Follows from the observation that $r$ is sampled independent of the adversary’s view.
\end{proof}

\begin{claim}
\label{clm:obf-3}
For all $q\le 2^{o(\lambda)}$, all $q$-query adversaries $\cB$, and large enough $\lambda$,
\[
\Prr_{r,r',S}\big[\cB^{W^{S_2}}(\obf(*,r'),\mathsf{Rev})=r\big]\le \frac{\poly(q)}{2^{\lambda/4^
{\ell+1}}}
\]
\end{claim}
\begin{proof}
Let $T:=\{y:\cO_1(y)\ne\cO_2(y)\}$. Then for all $q$-query adversaries $\cB$,
\[
\max_{\cB}\Prr_{r,r',S}\big[\cB^{W^{S_2}}(\obf(*,r'),\mathsf{Rev})\in T\big]\le \frac{2q}{2^\lambda}.
\]
This is because if $y\in T$ then either $y\in (*,r)$ or $\eval^{\obf_{(*,r')},f}(y)\neq \eval^{\obf_{(*,r') \cup (*,r)} ,f}(y)$. In the latter case, computing $C = \mathsf{Rev}(y)$, running $C^{\obf_{(*,r')}, f}(z')$ (where $y=(\_,z')$) and outputting a random query results in an element of $(*,r)$ with probability at least $1/q$. Therefore, given $y\in T$, we can output $r$ with probability at least $1/2q$. Combined with Claim \ref{clm:obf-2} this gives the above bound. Let $\cM$ be the measurement that measures in computational basis and accepts if the output is  $r$. By Theorem \ref{thm:sep-core} applied to $\cM,\cO_1,\cO_2$, and $z=(\obf(*,r'),\mathsf{Rev})$), for all $q$-query adversaries $\cB$,
\[
\Big|\Prr_{r,r',S}\big[\cB^{W^{S_1}}(\obf(*,r'),\mathsf{Rev})=r\big]-\Prr_{r,r',S}\big[\cB^{W^{S_2}}(\obf(*,r'),\mathsf{Rev})=r\big]\Big|\le \frac{\poly(q)}{2^{\lambda/ 4^{\ell+1}}},
\]
which, along with Claim \ref{clm:obf-2} and the fact that $q\le 2^{o(\lambda)}$, implies the claim.
\end{proof}

\begin{claim}
\label{clm:obf-4}
For all $q\le 2^{o(\lambda)}$, for all $q$-query adversaries $\cB$, and for large enough $\lambda$,
\[
\Pr_{r,r',S}\Big[\cB^{W^{S_2}}(\obf(*,r),\mathsf{Rev})=r\Big]
= \Pr_{r,r',S}\Big[\cB^{W^{S_2}}(\obf(*,r'),\mathsf{Rev})=r\Big] \leq \frac{\poly(q)}{2^{\lambda/4^
{\ell+1}}}
\]
\end{claim}
\begin{proof}
Follows from Claim \ref{clm:obf-3} and the fact that $r$ and $r'$ are symmetric in the view of $\cB$.
\end{proof}

\begin{claim}
\label{clm:obf-5}
For all adversaries $\cB$,
\[
\Pr_{r,r',S}\Big[\cB^{W^{S_3}}(\obf(*,r),\mathsf{Rev})=r'\Big] \leq \frac{1}{2^\lambda}.
\]
\end{claim}
\begin{proof}
Follows from the observation that $r'$ is sampled independent of the adversary’s view.
\end{proof}

\begin{claim}
\label{clm:obf-6}
For all $q\le 2^{o(\lambda)}$, all $q$-query adversaries $\cB$, and large enough $\lambda$,
\[
\Pr_{r,r',S}\Big[\cB^{W^{S_3}}(\obf(*,r),\mathsf{Rev})=r\Big]\le \frac{\poly(q)}{2^{\lambda/4^{\ell+1}}}.
\]
\end{claim}
\begin{proof}
Let $T:=\{y:\cO_2(y)\ne\cO_3(y)\}$. Then for all $q$-query adversaries $\cB$,
\[
\max_{\cB}\Pr_{r,r',S}\big[\cB^{W^{S_3}}(\obf(*,r),\mathsf{Rev})\in T\big]\le \frac{2q}{2^\lambda}.
\]
This is because if $y\in T$ then either $y\in (*,r')$ or $\eval^{\obf_{(*,r)},f}(y)\neq \eval^{\obf_{(*,r') \cup (*,r)} ,f}(y)$. In the latter case, computing $C = \mathsf{Rev}(y)$, running $C^{\obf_{(*,r)}, f}(z')$ (where $y=(\_,z')$) and outputting a random query results in an element of $(*,r')$ with probability at least $1/q$. Therefore, given $y\in T$, we can output $r'$ with probability at least $1/2q$. Combined with Claim \ref{clm:obf-5} this gives the above bound. Let $\cM$ be the measurement that measures in computational basis and accepts if the output is  $r$. By Theorem \ref{thm:sep-core} applied to $\cM,\cO_2,\cO_3$, and $z=(\obf(*,r'),\mathsf{Rev})$), for all $q$-query adversaries $\cB$,
\[
\Big|\Prr_{r,r',S}\big[\cB^{W^{S_2}}(\obf(*,r),\mathsf{Rev})=r\big]-\Prr_{r,r',S}\big[\cB^{W^{S_3}}(\obf(*,r),\mathsf{Rev})=r\big]\Big|\le \frac{\poly(q)}{2^{\lambda/ 4^{\ell+1}}}
\]
which, along with Claim \ref{clm:obf-4} and the fact that $q\le 2^{o(\lambda)}$, implies the claim.
\end{proof}

\begin{claim}
\label{clm:obf-7}
For all $q\le 2^{o(\lambda)}$, all $q$-query adversaries $\cB$, and large enough $\lambda$,
\[
\Big|\Prr_{r,r',S}\big[\cB^{W^{S_3}}(\obf(*,r),\mathsf{Rev})=1\big]-\Prr_{r,r',S}\big[\cB^{W^{S_4}}(\obf(*,r),\mathsf{Rev})=1\big]\Big|\le \frac{1}{2^{\lambda/ 4^{2\ell+3}}}
\]
\end{claim}
\begin{proof}
Let $T:=\{y:\cO_3(y)\ne\cO_4(y)\}$. Then for all $q$-query adversaries $\cB$,
\[
\max_{\cB}\Pr_{r,r',S}\big[\cB^{W^{S_3}}(\obf(*,r),\mathsf{Rev})\in T\big]\le \frac{1}{2^{\lambda/4^{\ell+2}}}.
\]
This is because if $y\in T$ then either $y\in (*,r)$ or $\eval^{\obf,f}(y)\neq \eval^{\obf_{(*,r)} ,f}(y)$. In the latter case, computing $C = \mathsf{Rev}(y)$, running $C^{\obf_{(*,r)}, f}(z')$ (where $y=(\_,z')$) and outputting a random query results in an element of $(*,r)$ with probability at least $1/q$. Therefore, given $y\in T$, we can output $r$ with probability at least $1/2q$. Combined with Claim \ref{clm:obf-6} and the fact that $q\le 2^{o(\lambda)}$, this gives the above bound. Let $\cM$ be the measurement that measures in computational basis and accepts if the output is  $1$. By Theorem \ref{thm:sep-core} applied to $\cM,\cO_3,\cO_4$, and $z=(\obf(*,r),\mathsf{Rev})$), for all $q$-query adversaries $\cB$,
\[
\Big|\Prr_{r,r',S}\big[\cB^{W^{S_3}}(\obf(*,r),\mathsf{Rev})=1\big]-\Prr_{r,r',S}\big[\cB^{W^{S_4}}(\obf(*,r),\mathsf{Rev})=1\big]\Big|\le \frac{\poly(q)}{2^{\lambda/ (2\cdot4^{2\ell+2})}}
\]
which, along with the fact that $q\le 2^{o(\lambda)}$, implies the claim.
\end{proof}
\noindent Now suppose there exist $q \leq 2^{o(\lambda)}$, $q$-query adversary $\cA$, and bit $b$ such that for all large enough $\lambda$
\[
\Big|\Prr\big[\cG'_{b,\lambda}(\cA)=1\big]-\Prr\big[\cG''_{b,\lambda}(\cA)=1\big]\Big|\ge \frac{1}{2\cdot 2^{\lambda/4^{2(\ell+2)}}}.
\]
We use this to construct $\cB$ that contradicts Claim \ref{clm:obf-7}. For $k\in\{3,4\}$, let $\cB^{W^{S_k}}(\obf(*,r),\mathsf{Rev})$ perform the following:
\begin{itemize}
\item $(C_0,C_1)\leftarrow \cA^{W^{S_k}}(1^\lambda)$
\item If $C_0$ is not functionally equivalent to $C_1$, or $C_1\notin\{0,1\}^{\lambda}$, or $C_0\notin\{0,1\}^{\lambda}$, return $\bot$
\item Return $\cA^{W^{S_k}}\big(\obf(C_b,r)\big)$. Note that $\cB$ receives the entire set $\obf(*,r)$ and can therefore compute $\obf(C_0,r)$.
\end{itemize}
Since $S_3\equiv S_r$ and $S_4\equiv S$, this implies that 
\[
\Prr_{r,r',S}\big[\cB^{W^{S_3}}(\obf(*,r),\mathsf{Rev})=1\big] = \Prr\big[\cG''_{b,\lambda}(\cA)=1\big]
\]
and
\[
\Prr_{r,r',S}\big[\cB^{W^{S_4}}(\obf(*,r),\mathsf{Rev})=1\big] = \Prr\big[\cG'_{b,\lambda}(\cA)=1\big]
\]
Therefore,
\[
\Big|\Prr_{r,r',S}\big[\cB^{W^{S_3}}(\obf(*,r),\mathsf{Rev})=1\big]-\Prr_{r,r',S}\big[\cB^{W^{S_4}}(\obf(*,r),\mathsf{Rev})=1\big]\Big|\ge \frac{1}{2\cdot 2^{\lambda/4^{2(\ell+2)}}}
\]
contradicting Claim \ref{clm:obf-7}.
\end{proof}
\subsubsection{Achieving Contradiction}
Finally, we show how a classical adversary can use $P$ to break every $\ell$-round IV-PoQ.
\begin{theorem}
    For every $QPT$ prover $\cP$ and verifier $\cV$ such that 
    for all large enough $\lambda\in\bbN$, 
        \[
            \Pr_{P_\lambda}[\langle \cP^{f,\obf,\eval^{\obf, f}},\cV^{f,\obf,\eval^{\obf,f}}\rangle(1^\lambda) = 1] \geq c(\lambda)
        \]
    there exists an oracle aided PPT adversary $\cA$ such that for all large enough $\lambda$
    \[
        \Pr_{P_\lambda}[\langle \cA^{P_\lambda},\cV^{f,\obf,\eval^{\obf,f}}\rangle(1^\lambda) = 1] \geq c(\lambda)
    \]
\end{theorem}
\begin{proof}
We may purify the QPT prover's circuit to obtain a circuit $\cC$ that takes as input the transcript so far and the (purified) prover state after the previous message and returns an output register and a state register. The next message is sampled by measuring the output register, obtaining a next message string and residual prover state. The classical adversary $\cA$ produces the first prover message by querying $P_{1,\lambda}$ on $(\bot, \bot, \cC_1)$ and obtaining the message $s_1$ along with a hash string $\sigma_1$, where $\tau_1$ is the transcript so far (which only consists of the verifier first message, if any) and $\cC_1 = \cC(\tau_1,\cdot))$ is the circuit with the first input hardcoded. For $i>1$, $\cA$ obtains the $i$-th message by querying $P_{i,\lambda}$ on $((\cC_1, s_1, \cC_2, s_2, \ldots, \cC_{i-1}, s_{i-1}), \sigma_{i-1}, \cC_i)$ and obtaining the message $s_i$ along with a hash string $\sigma_i$, where $\tau_i$ is the transcript so far and $\cC_i = \cC(\tau_i,\cdot))$ is the circuit with the first input hardcoded. This allows $\cA$ to perfectly imitate the behavior of the QPT prover, and hence the verifier accepts with exactly the same probability.
\end{proof}
\noindent 
Putting these results together we can obtain our final theorem for this section.
\begin{theorem}
    Fully black-box constructions of constant-round IV-PoQ from indistinguishability obfuscation and one-way permutations do not exist.
\end{theorem}
\begin{proof}
    For every QPT prover $\cP$ and verifier $\cV$ such that 
    for all large enough $\lambda\in\bbN$, 
        \[
            \Pr_{P_\lambda}[\langle \cP^{f,\obf,\eval^{\obf, f}},\cV^{f,\obf,\eval^{\obf,f}}\rangle(1^\lambda) = 1] \geq c(\lambda)
        \]
    there exists an oracle aided PPT adversary $\cA$ such that for all large enough $\lambda$
    \[
        \Pr_{P_\lambda}[\langle \cA^{P_\lambda},\cV^{f,\obf,\eval^{\obf,f}}\rangle(1^\lambda) = 1] \geq c(\lambda)
    \]
    Therefore if $\langle P,V \rangle$ is an $\ell$-round $(c(\lambda), s(\lambda))$-IV-PoQ with a fully black-box reduction $R$, it must be the case that there exist polynomials $p_1$ and $p_2$ such that for infinitely many $\lambda$,
        \begin{align*}
            \Prr_{P_\lambda, x\leftarrow\bin^\lambda} \left[ R^{\cA,f,\obf,\eval^{\obf,f}} (f(x))=x \right] \ge \frac{1}{p_1(\lambda)}
        \end{align*}
        or 
        \begin{align*}
            \left| \Prr_{P_\lambda}\left[ \cG_{0,\lambda}(R^{\cA,f,\obf,\eval^{\obf,f}})=1\right]-\Prr_{P_\lambda}\left[\cG_{1,\lambda}(R^{A,f,\obf,\eval})=1 \right] \right| \ge \frac{1}{p_2(\lambda)}
        \end{align*}
        which would contradict Theorem \ref{thm:owp} or Theorem \ref{thm:io-security}.
\end{proof}

%% file: comp-oracle.tex
\section{IV-PoQ with Quantum Black-Box Reductions}\label{sec:bbreductions}
\subsection{Compression Framework}
In this subsection we show how to efficiently simulate quantum query access to a random classical oracle where the distribution on outputs for each query is non-uniform and which can only be efficiently (quantumly) sampled using a copy of advice state $\ket{\psi}$. Somewhat surprisingly, despite the fact that a quantum querier can query the oracle in superposition on super-polynomially many points, we only require $\poly(q)$ copies of the state where $q$ is an upper bound on the number of queries made to the oracle.\\

\noindent Let $\cD := \bigotimes_{x \in \{0,1\}^n} \cD_x$.  
Suppose there exists a state $\ket{\psi} = \sum_z \alpha_z \ket{z}$ and efficient unitary $U = \sum_x \ket{x}\!\bra{x} \otimes U_x$ such that $\cD_x$ is the distribution that results from measuring $U_x\ket{\psi}$ in the computational basis. For $a,b \in \mathbb{N}$, define
\[
\ket{S_{a,b}}:= \ket{\bot}^{\otimes a} \otimes \ket{\psi}^{\otimes b}.
\]
Define a database $D$ as a collection of $(x,y)$ pairs, stored in sorted order by the first parameter, where for any $x$ there exists at most one tuple of form $(x,y)$ in $D$.
We say that $D(x) = \bot$ if $D$ does not contain a tuple of form $(x,y)$.  
Otherwise if $(x,y) \in D$ then $D(x)=y$.  For $D'$ such that $D'(x) = \bot$,  $D' \cup (x,y)$ refers to the database that results from inserting $(x,y)$ in the appropriate spot in $D'$.

We define the operation $\Comp_x$ on states of form $\ket{D}\otimes \ket{S_{a,b}}$. The operation will never be applied to states not of this form and may therefore be defined arbitrarily on such states while maintaining unitarity.
\begin{enumerate}
\item If $D(x)=\bot$:  
\[
\Comp_x \ket{D}\otimes \ket{S_{a,b}}= \sum_z \alpha_z \ket{D\cup(x,z)}\otimes |S_{a+1,b-1}\rangle
\]
\item For any $\ket{\varphi} = \sum_z \alpha_z \ket{D'\cup(x,z)}$ for some $D'$ s.t. $D'(x)=\bot$:  
\[
\Comp_x \ket{\varphi} \otimes \ket{S_{a,b}}= \ket{D}\otimes |S_{a-1,b+1}\rangle
\]
\item For any $\ket{\varphi'} = \sum_z \alpha'_z \ket{D'\cup(x,z)}$ for some $D'$ s.t. $D'(x)=\bot$ and some $\{\alpha'_z\}_z$ such that $\sum_z \alpha'_z \ket{z}$ is orthogonal to $\ket{\psi}$:  
\[
\Comp_x \ket{\varphi} \otimes \ket{S_{a,b}}= \ket{\varphi} \otimes \ket{S_{a,b}}
\]
\end{enumerate}
Note that the defined states span the set of states of the form $\ket{D}\otimes \ket{S_{a,b}}$. Since this mapping preserves orthogonality, there exists a unitary implementing $\Comp_x$.

\begin{lemma}
\label{lem:comp-from-reflection}
Given oracle access to $P = \bbI - 2\ket{\psi}\bra{\psi}$, there exists an efficient algorithm to implement $\Comp_x$, where the runtime is polynomial in the size of the database and $a+b$.
\end{lemma}

\begin{proof}
We describe the algorithm on states of the form $\ket{D}\otimes \ket{S_{a,b}}$.
\begin{itemize}
\item If $D(x)=\bot$:  
\begin{itemize}
    \item Swap the first non-$\bot$ register from $\ket{S_{a,b}}$ with $|\bot\rangle$ to obtain $\ket{D}\otimes \ket{\psi} \otimes |S_{a+1,b-1}\rangle$.  If $b=0$, abort.  
    \item By acting coherently on computational basis states, map the obtained state to
    \[
    \sum_z \alpha_z \ket{D\cup(x,z)}\otimes |S_{a+1,b-1}\rangle,
    \]
    where ancillas are restored to $\ket{0}$ and hence omitted.
\end{itemize}
\item If $D = D' \cup \{(x,y)\}$ where $D'(x)=\bot$: 
\begin{itemize}
    \item Map the state to 
    \[
      \ket{D'}\otimes |y\rangle \otimes \ket{0} \otimes \ket{S_{a,b}}.
    \]
    \item Use $P$ to implement the projector onto $\ket{\psi}$, and apply 
    \[
      \ket{\psi}\!\bra{\psi} \otimes X + \big(\bbI - \ket{\psi}\!\bra{\psi}\big)\otimes \bbI
    \]
    to $|y\rangle \otimes \ket{0}$. This results in the state
    \[
      \ket{D'}\otimes \big(\beta\ket{\psi}\ket{1} + (1-\beta)\ket{\phi}\ket{0}\big)\otimes \ket{S_{a,b}},
    \]
    where $\ket{\phi}$ is some state orthogonal to $\ket{\psi}$ and $\beta = \langle \psi|y\rangle$.
    \item Conditioned on the bit in the third register, swap the second register with the last $\ket{\bot}$ in $\ket{S_{a,b}}$, to obtain
    \[
      \ket{D'}\otimes \big(\beta\ket{\bot}\ket{1}\ket{S_{a-1,b+1}}
        + (1-\beta)\ket{\phi}\ket{0}\ket{S_{a,b}}\big).
    \]
    \item Apply 
    \[
      \ket{\bot}\!\bra{\bot}\otimes X + \big(\bbI - \ket{\bot}\!\bra{\bot}\big)\otimes \bbI
    \]
    again and remove ancillas, to get
    \[
      \ket{D'}\otimes \big(\beta\ket{\bot}\ket{S_{a-1,b+1}}
        + (1-\beta)\ket{\phi}\ket{S_{a,b}}\big).
    \]
\end{itemize}

\end{itemize}

It is straightforward to see that the algorithm satisfies the requirements for $\Comp_x$ and makes at most 2 queries to $P$.
\end{proof}
\noindent
Define 
\[
\Comp = \sum_x \ket{x}\!\bra{x} \otimes \Comp_x,
\]
and let $\StdO$ be the unitary that maps 
\[
\ket{x,y}\ket{D} \mapsto \ket{x, y \oplus D(x)}\ket{D},
\]
where $\ket{x}$ and $\ket{y}$ are query and response registers respectively.  
For every $x$, let $\widetilde{U}_x$ be the unitary that acts on databases $\ket{D}$ as:
\begin{itemize}
    \item If $D(x)=\bot$, then do nothing.
    \item If $D = D' \cup \{(x,y)\}$ for $D'(x)=\bot$, suppose
    \[
      U_x \ket{y} = \sum_z \beta_z \ket{z},
    \]
    then
    \[
      \widetilde{U}_x \ket{D} = \sum_z \beta_z \ket{D' \cup (x,z)}.
    \]
\end{itemize}
Let $\widetilde{U} := \sum_x \ket{x}\!\bra{x} \otimes \widetilde{U}_x$.  
We now define the compressed oracle with advice.  
Let $q$ be an upper bound on the number of queries made to the oracle.
Initialize database register $\reg{D}$ to the empty database $\ket{\{\}}$ and register $\reg{S}$ to $|S_{0,q}\rangle$. 
The querier prepares query register $\reg{X}$ and response register $\reg{Y}$, and a query to $\AdvO$ is implemented as
\[
\sum_x \ket{x}\!\bra{x}_{\reg{X}} \otimes 
\big( \mathsf{Comp}_x^\dagger \circ \widetilde{U}_x^\dagger \circ \mathsf{StdO} \circ \widetilde{U}_x \circ \mathsf{Comp}_x \big)
\]
applied to query register $\reg{X}$, response register $\reg{Y}$, the database stored in $\reg{D}$ and the stored states in $\reg{S}$.
\begin{theorem}
\label{thm:perfect-comp-oracle-with-advice}
For any $q$-query adversary $\adv$,
\[
\Prr_{O \leftarrow \cD}[\adv^{O}=1] = \Prr[\adv^{\AdvO}=1],
\]
\end{theorem}

\begin{proof}
We prove via a sequence of hybrids.

\paragraph{Hybrid 1.} Sample $O$ from $\cD$ and answer queries by mapping  
$$\ket{x,y} \ket{O} \mapsto \ket{x,y\oplus O(x)}\ket{O}.$$ 
where $\ket{O}$ contains the truth table representation of $O$. This is the standard definition of queries to $O$.

\paragraph{Hybrid 2.} Instead of sampling $O$ directly, we initialize the truth table register to a purification of $O$. Specifically we initialize the truth table to $\bigotimes_x U_x\ket{\psi}$.  \\

\noindent Since this is a valid purification of the truth table register, measuring it in the computational basis recovers Hybrid 1. Since the operations performed on the truth table register are diagonal in the computational basis, they commute with this measurement and therefore both hybrids are indistinguishable.

\paragraph{Hybrid 3.} View the purification register $\bigotimes_x U_x\ket{\psi}$ as a superposition of multiple truth tables.  
We encode each such truth table $O$ in database form, i.e., the database corresponding to $O$ is defined as $D = \{(0,O(0)), (1,O(1)),\dots,(2^n-1,O(2^n-1))\}$. We answer queries by applying $\StdO$ which maps $\ket{x,y}\ket{D} \mapsto \ket{x, y \oplus D(x)}\ket{D}$. Specifically, if for each $x$,
\[
U_x \ket{\psi} = \sum_y \alpha_{x,y} \ket{y}.
\]
for some $\{\alpha_{x,y}\}_y$, then the the database register is instantiated to
\[
\bigotimes_x \sum_y \alpha_{x,y}\ket{(x,y)}.
\]
\noindent Since we only change the encoding, Hybrid 2 and Hybrid 3 are identical in the adversary’s view.

\paragraph{Hybrid 4.} Define the global compression procedure
\[
\mathsf{GComp} := \mathsf{Comp}_0 \circ \mathsf{Comp}_1 \circ \dots \circ \mathsf{Comp}_{2^n-1},
\]
and similarly $G U := \widetilde{U}_0 \circ \dots \circ \widetilde{U}_{2^n-1}$.
Initialize the database register $\reg{D}$ to the empty database $\ket{\{\}}$ and initialize register $\reg{S}$ to $\ket{S_{0,2^n}}$. Implement oracle queries as
\[
\mathsf{GComp}^\dagger \circ G U^\dagger \circ \mathsf{StdO} \circ G U \circ \mathsf{GComp}.
\]
Note that $\GComp$ and $GU$ only act on registers private to the oracle and therefore commute with computation on the adversary’s registers.  
As a result, between two calls to $\StdO$, the applications of $\GComp^\dagger \circ GU^\dagger$ and $GU \circ \GComp$ cancel out.  In the beginning of the hybrid, $GU \circ \GComp$ is applied to $\ket{\{\}} \otimes \ket{S_{0,2^n}}$, 
resulting in
\[
\left( \bigotimes_x \widetilde{U}_x \sum_y \alpha_z \ket{(x,z)} \right) \otimes \ket{S_{2^n,0}}
= \bigotimes_x \sum_y \alpha_{x,y}\ket{(x,y)} \otimes \ket{S_{2^n,0}}.
\]
This is exactly the starting state in Hybrid~3 tensored with $\ket{S_{2^n,0}}$.  
Therefore, both hybrids are identical in the adversary’s view.

\paragraph{Hybrid 5.}  
We switch to implementing queries as 
\[
\GComp^\dagger \circ \widetilde{U}^\dagger \circ \StdO \circ \widetilde{U} \circ \GComp.
\]
For any database $D$, and any $x,y$, and any $x' \neq x$, since $\widetilde{U}_{x'}$ does not affect tuples of form $(x,y)$,
\begin{align*} 
\StdO \circ \widetilde{U}_{x'} \ket{x,y}\ket{D} 
&= \StdO \ket{x, y} \otimes \widetilde{U}_{x'}\ket{D}\\
&= \ket{x, y \oplus D(x)} \otimes \widetilde{U}_{x'}\ket{D}\\
&= \widetilde{U}_{x'} \ket{x, y \oplus D(x)} \otimes \ket{D}\\
&= \widetilde{U}_{x'} \circ \StdO \ket{x,y}\ket{D}.
\end{align*}
Additionally, for $x \neq x'$, $\widetilde{U}_x$ and $\widetilde{U}_{x'}$ commute, since they act on different tuples in the database and do not add, remove, or change the ordering of the $x$’s in the database.  
Therefore, for all $D,x,y$,
\begin{align*}
GU^\dagger \circ \StdO \circ GU \ket{x,y}\ket{D} &= \widetilde{U}_x^\dagger  \circ \StdO \circ \widetilde{U}_x\ket{x,y}\ket{D}\\
&= \widetilde{U}^\dagger  \circ \StdO \circ \widetilde{U} \ket{x,y}\ket{D},
\end{align*}
which shows the hybrids are identical in the adversary’s view.  

\paragraph{Hybrid 6.}  
We switch to implementing the queries as 
\[
\Comp^\dagger \circ \widetilde{U}^\dagger  \circ \StdO \circ \widetilde{U} \circ \Comp.
\]
Let $\widetilde{\StdO} := \widetilde{U}^\dagger \circ \StdO \circ \widetilde{U}$. First we note that for any $x \neq x'$, for any state of the form 
$\ket{D}\otimes \ket{S_{|D|,2^n-|D|}}$, where $|D|$ refers to the number of tuples in $D$,  
\begin{align*}
  \Comp_x \circ \Comp_{x'} \ket{D}\otimes \ket{S_{|D|,2^n-|D|}}
  &= \Comp_{x'} \circ \Comp_x \ket{D}\otimes \ket{S_{|D|,2^n-|D|}}.
\end{align*}
and
\begin{align*}
  \widetilde{\StdO} \circ \Comp_{x'} 
    \ket{x,y} \ket{D}\ket{S_{|D|,2^n-|D|}}
  &= \Comp_{x'} \circ\widetilde{\StdO}
    \ket{x,y} \ket{D} \ket{S_{|D|,2^n-|D|}}.
\end{align*}
Next, note that applying $\Comp_x$ maintains the invariant of being supported on states of the form 
$\ket{D}\otimes \ket{S_{|D|,2^n-|D|}}$. Additionally since $\widetilde{U}$ and $\StdO$ do not add, remove, or reorder tuples, and do not act on $\reg{S}$, applying $\widetilde{\StdO}$ also does not disturb the invariant. As a result,
\begin{align*}
  \GComp^\dagger& \circ \widetilde{\StdO} \circ \GComp 
    \ket{x,y}\ \ket{D}\ \ket{S_{|D|,2^n-|D|}}\\
  &= \Comp_x^\dagger \circ \widetilde{\StdO} \circ \Comp_x 
    \ket{x,y}\ \ket{D}\ \ket{S_{|D|,2^n-|D|}} \\
  &= \Comp^\dagger \circ \widetilde{\StdO} \circ \Comp 
    \ket{x,y}\ \ket{D}\ \ket{S_{|D|,2^n-|D|}}.
\end{align*}
which shows that both hybrids are indistinguishable in the adversary’s view.
\paragraph{Hybrid 7.}  
We initialize $\reg{S}$ to $\ket{S_{0,q}}$.  \\

\noindent Note that the only operations that interact with $\reg{S}$ are $\Comp_x$ and $\Comp_x^\dagger$, which only touch the $\ket{\bot}$ registers and the first non-$\bot$ register of $S$.  
Additionally, the number of $\ket{\bot}$ registers only increases by at most $1$ each query.  
Therefore, at every point before the last query, $\reg{S}$ is supported on states with at most $q-1$ $\ket{\bot}$ registers.  
Queries in Hybrid~6 therefore commute with any operation on all but the first $q$ registers of $\reg{S}$.  
Tracing out the last $2^n - q$ registers gives Hybrid~7, and both hybrids are therefore indistinguishable.  \\

\noindent Hybrid 7 is identical to querying $\AdvO$, which completes the proof the the theorem. 
\end{proof}
By Lemma \ref{lem:comp-from-reflection}, $\Comp_x$ and hence the query operation of $\AdvO$ is efficient to implement using query access to reflection around $\ket{\psi}$. The initial registers of $\AdvO$ can be prepared using $q$ copies of $\ket{\psi}$. All that remains to show efficient implementation of $\AdvO$ is to show that the reflection around around $\ket{\psi}$ can be efficiently implemented.
\begin{theorem}
\label{thm:efficient-comp-oracle-advice}
There exists a QPT simulator $\Sim$ such that for all $q$-query algorithms $\adv$, for all $n > 0$,  
\[
\Big| \Prr_{O\leftarrow\cD}[\adv^{O}=1] - \Pr[\Sim^{\adv}(\ket{\psi}^{\otimes (q+n)})=1] \Big|
\;\leq\; \frac{4\sqrt{2}q}{\sqrt{n+1}}.
\]
\end{theorem}

\begin{proof}
By Theorem \ref{thm:perfect-comp-oracle-with-advice} and Lemma \ref{lem:comp-from-reflection}, we can exactly simulate $\adv^O$ for random $O$ efficiently using $q$ copies of $\ket{\psi}$ given oracle access to $P = \bbI - 2\ket{\psi}\bra{\psi}$. Additionally, the simulation makes at most four queries (two for running $\Comp$ and two for running $\Comp^\dagger$) for every query $\adv$ makes to $O$.  
By Theorem~\ref{thm:jls}, there exists a QPT algorithm $\mathcal{B}$ that simulates oracle access to $P$ using $\ell$ copies of $\ket{\psi}$ within error $\frac{q'\sqrt{2}}{\sqrt{n+1}}$ in trace distance, where $q'$ is the number of queries made to $P$.  Using $\mathcal{B}$ to simulate $P$ and noting that $q' = 4q$ gives the bound in the theorem statement.
\end{proof}
\subsection{One-way Puzzles from IV-PoQs with Black-Box Reductions}
\begin{theorem}
An IV-PoQ with black-box reduction to a falsifiable assumption secure against efficient non-uniform quantum adversaries implies the existence of one-way puzzles.
\end{theorem}

\begin{proof}
Let $\langle P,V\rangle(1^\lambda)$ be an IV-PoQ with reduction $R$ to assumption $(C,\thres)$.  
Let the correctness and soundness parameters be $c(\cdot)$ and $s(\cdot)$, where
\[
c(\lambda) - s(\lambda) \geq \frac{1}{t(\lambda)}
\]
for some polynomial $t$.  
Let $\ell(\lambda)$ be the number of rounds in the protocol.

\paragraph{Defining the Puzzle.}  
Define a distributional one-way puzzle $\Samp(1^\lambda)$ as follows.
\begin{enumerate}
  \item Sample $j \gets [\ell(\lambda)]$.
  \item Run $\langle P,V\rangle(1^\lambda)$ up to the $j$-th round to obtain $(m_1,\dots,m_j)$.
  \item Set $\puzz := (m_1,\dots,m_{j-1})$ and $\key := m_j$.
  \item Output $(\puzz, \key)$.
\end{enumerate}

We show that this defines a $1/2\ell(\lambda)t(\lambda)$-distributional one-way puzzle.
Suppose not: then there exists a QPT adversary $A$ and advice states $\{\ket{\psi_\lambda}\}_\lambda$ such that for infinitely many $\lambda$,
\[
\SD\Big(\{\puzz,\key\}, \{\puzz, A(\puzz,\ket{\psi_\lambda})\}\Big) \leq \frac{1}{2\ell(\lambda)t(\lambda)}.
\]
where $(\puzz, \key) \gets \mathsf{Samp}(1^\lambda)$. Let $B$ be a classical machine such that $\forall \lambda, \forall x$, 
the output distribution of $B(x,1^\lambda)$ is identical to the output distribution of 
$A(x,\ket{\psi_\lambda})$. Next, define a classical prover $P'$ for the IV-PoQ that samples 
the $i$-th message by calling $B(\tau,1^\lambda)$, where $\tau$ is the transcript of 
the first $i-1$ messages.
\begin{claim}
\[
\Pr[\langle P',V\rangle(1^\lambda)=1] \geq c(\lambda) - \frac{1}{2t(\lambda)} \geq s(\lambda) + \frac{1}{2t(\lambda)}. 
\]
\end{claim}

\begin{proof}
Define a sequence of hybrid provers $\{P_i\}_{i \in [0,\ell]}$, where $P_i$ runs the honest (quantum) prover up to the first $i$ messages, then samples every subsequent message using $B$ as in $P'$.  
Note that $P_0 \equiv P$ and $P_\ell \equiv P'$.  
For $i \in [\ell]$, define
\[
\varepsilon_i := \Pr[\langle P_{i-1},V\rangle(1^\lambda)=1] - \Pr[\langle P_i,V\rangle(1^\lambda)=1].
\]
$P_{i-1}$ and $P_i$ differ only in how the $i$-th message is sampled.  
If the $i$-th message is a verifier message, both are identical.  
If it is a prover message, then $P_i$ samples the $i$-th message using $B$, 
while $\tilde{P}_i$ samples using the honest prover.  

Let $\tau$ and $\tau'$ be random variables representing the transcript up to the $i$-th 
message in $\langle P_{i-1},V\rangle(1^\lambda)$ and $\langle P_i,V\rangle(1^\lambda)$ 
respectively. Therefore, by construction,
\[
\varepsilon_i \leq \SD(\tau, \tau'),
\]
Let $(\puzz_i, \key_i)$ be random variables representing the output of 
$\mathsf{Samp}(1^\lambda)$ conditioned on sampling $j=i$. Then, $\tau$ is identically 
distributed to $(\puzz_i, \key_i)$ and $\tau'$ is identically distributed to 
$(\puzz_i, A(\puzz_i,\ket{\psi_\lambda}))$. Additionally,
\begin{align*}
\SD&\big(\{\puzz,\key\}, \{\puzz, A(\puzz,\ket{\psi_\lambda})\}\big) \\
&= \sum_j \Pr[j] \cdot \SD\big(\{\puzz_j,\key_j\}, \{\puzz_j, A(\puzz_j,\ket{\psi_\lambda})\}\big) \\
&\geq \sum_j \Pr[j] \cdot \varepsilon_j \\
&= \frac{\sum_j \varepsilon_j}{\ell(\lambda)}
\end{align*}
Since $A$ breaks the security of the puzzle, for infinitely many $\lambda$, the LHS is upper bounded by $\frac{1}{2t(\lambda)\ell(\lambda)}$. Therefore for infinitely many $\lambda$,
\begin{equation}
\label{eq:owpuzz}
\sum_{i\in[\ell]} \varepsilon_i \leq \frac{1}{2t(\lambda)}.
\end{equation}
By telescoping the difference we aim to bound, 
\begin{align*}
\Pr&[\langle P,V\rangle(1^\lambda)=1] - \Pr[\langle P',V\rangle(1^\lambda)=1]\\
&= \Pr[\langle P_0,V\rangle(1^\lambda)=1] - \Pr[\langle P_\ell,V\rangle(1^\lambda)=1]\\
&= \sum_{i\in[\ell]}\Pr[\langle P_{i-1},V\rangle(1^\lambda)=1] - \Pr[\langle P_i,V\rangle(1^\lambda)=1]\\
&= \sum_{i\in[\ell]}\varepsilon_i
\end{align*}
Finally, plugging the above into \eqref{eq:owpuzz} along with the observation that
\[
\Pr[\langle P,V\rangle(1^\lambda)=1] \geq c(\lambda)
\]
implies the statement of the claim.
\end{proof}
As a consequence of the above claim and the properties of the reduction, there exists a polynomial $p$ such that for infinitely many $\lambda$,
\[
\Pr[\langle R^{P'},C\rangle(1^\lambda)=1] \geq \thres + \frac{1}{p(\lambda)}.
\]
$R^{P'}$ can be efficiently implemented given quantum query access to $B$, 
and since $B(\cdot,1^\lambda)$ samples from the output distribution of 
$A(\cdot,\ket{\psi_\lambda})$, by Theorem \ref{thm:efficient-comp-oracle-advice}, quantum queries to $B$ can be 
efficiently simulated to upto $1/2p(n)$ error using $128p^2(\lambda)q^2(\lambda) + q(\lambda)$ copies of $\ket{\psi_\lambda}$, where $q(\lambda)$ is the 
number of queries made by $R^{(\cdot)}$. This gives an efficient quantum simulator $\Sim$ such that for infinitely many $\lambda$,
\[
\Pr[\langle \Sim(\ket{\psi_\lambda}^{\otimes \poly(\lambda)}), C\rangle(1^\lambda) = 1] 
   \geq \thres + \frac{1}{2p(\lambda)},
\]
which breaks the falsifiable assumption.
\end{proof}
\subsection{Impossibility}
\begin{theorem}
A 3-round IV-PoQ with black-box reduction to a falsifiable assumption secure against efficient non-uniform quantum adversaries does not exist.
\end{theorem}

\begin{proof}
Let $\langle P,V\rangle$ be a 3-round IV-PoQ with reduction $R$ to assumption $(C,\thres)$.  
Let the correctness and soundness parameters be $c(\cdot)$ and $s(\cdot)$ where
$c(\lambda) - s(\lambda) \geq \frac{1}{t(\lambda)}$
for some polynomial $t$.

For each $\lambda$, suppose we purify the honest prover and only measure at the end of each round to obtain the next message.  
For any first message $m$, let $|\psi_m\rangle$ be the residual prover state after measuring $m$, and let $m^*$ be the first message that maximizes the probability of acceptance in $\langle P,V\rangle(1^\lambda)$ when randomness is taken over the remaining messages and verifier output.  

Let $P'$ be a classical prover that always outputs $m^*$ as its first message, and samples the third message according to the honest prover’s third-message distribution conditioned on the first message being $m^*$.  
Note that the third message can be (quantumly) efficiently sampled given $m^*$, $|\psi_{m^*}\rangle$, and the second message. For all $\lambda$,
\[
\Pr[\langle P',V\rangle(1^\lambda)=1] \geq c(\lambda),
\]
since the third message distribution is the same and we fixed the best first message.  
By the properties of $R$, there exists a polynomial $p$ such that for infinitely many $\lambda$,
\[
\Pr[\langle R^{P'}, C\rangle(1^\lambda)=1] \geq \thres + \frac{1}{p(\lambda)}.
\]
Finally, since the first message is fixed and the third message can be efficiently sampled using advice $(m^*,|\psi_{m^*}\rangle)$, by Theorem \ref{thm:efficient-comp-oracle-advice}, there exists an efficient quantum simulator $\Sim$ such that for infinitely many $\lambda$,
\[
\Pr[\langle \Sim(m^*, |\psi_{m^*}\rangle^{\otimes O(q^2(\lambda)p^2(\lambda))}), C\rangle(1^\lambda)=1] \geq \thres + \frac{1}{2p(\lambda)},
\]
where $q(\lambda)$ is the number of queries made by $R$, which breaks the falsifiable assumption.
\end{proof}

%% file: classical-black-box.tex
\section{PoQ with Classical Black-Box Reductions}\label{sec:publiccoinreductions}
\subsection{LOCC de Finetti}
We state a de Finetti theorem for one-way LOCC measurements, which are measurements where each subsystem is measured oncee and the measurement results are sent to every other (unmeasured) subsystem. The measured system is not touched again.
\begin{theorem}\cite{1locc}\label{thm:1locc}
Let $\rho_{\reg{A}_1 \dots \reg{A}_n}$ be a permutation-invariant state on 
$\cH_\reg{A}^{\otimes n}$. Then for any integer $0 \leq k \leq n$ there exists a measure 
$\mu$ on density matrices on $\cH_\reg{A}$ such that
\[
  \Bigl\| \rho_{\reg{A}_1 \dots \reg{A}_k} - \int \sigma^{\otimes k}\, d\mu(\sigma) 
  \Bigr\|_{\mathrm{LOCC}_1} 
  \leq \sqrt{\frac{2(k-1)^2 \ln |\reg{A}|}{n-k}}.
\]
\end{theorem}
\noindent We also conjecture that a similar result holds for arbitrary LOCC measurements (without the restriction of measuring once).
\begin{conjecture}
\label{conj:locc}
For every tuple of polynomials $(q(\cdot), \nu(\cdot), k(\cdot))$ there exists a polynomial 
$n(\cdot)$ such that for all large enough $\lambda \in \mathbb{N}$, for all permutation-invariant 
states $\rho_{\reg{A}_1 \dots \reg{A}_{n(\lambda)}}$ on $\cH_\reg{A}^{\otimes n(\lambda)}$ (where 
$\reg{A} = \{\bbC^2\}^{\otimes k(\lambda)}$), there exists a measure $\mu$ on density matrices on $\cH_\reg{A}$ such that
\[
  \Bigl\| \rho_{\reg{A}_1 \dots \reg{A}_{q(\lambda)}} - \int \sigma^{\otimes q(\lambda)}\, d\mu(\sigma) 
  \Bigr\|_{\mathrm{LOCC}} \leq \frac{1}{\nu(\lambda)}.
\]
\end{conjecture}
\subsection{Constructions}
\begin{theorem}
Assuming Conjecture \ref{conj:locc}, the existence of a constant-round PoQ with $1-\negl(\lambda)$ correctness and a classical black-box reduction to a falsifiable assumption secure against QPT adversaries implies the existence of uniform i.o. $\poly(\lambda)$-weak minischemes. 
\end{theorem}

\begin{proof}
Let the PoQ be $\langle P,V \rangle$ with reduction $R$ to assumption $(C,\thres)$, with correctness $1-\negl(\lambda)$ and soundness $1-\frac{1}{p(\lambda)}$ for some polynomial $p$.  Note that by Theorem \ref{thm:poq-parallel-rep} this is without loss of generality.
Let $2\ell$ be the number of rounds, where we add a dummy round in the beginning if the number of rounds was not initially even, and suppose $R$ makes at most $q(\lambda)$ queries.  
Let $k(\lambda)$ be an upper bound on the length of the (purified) state of the honest prover during execution.

By definition of PoQ, there exists a polynomial $p'$ such that for any classical prover $\widetilde{P}$ where for infinitely many $\lambda$, $\Pr[\langle \widetilde{P},V \rangle(1^\lambda) = 1] \geq 1 - \frac{1}{2p(\lambda)}$, it must be the case that for infinitely many $\lambda$,
\[
\Pr\big[\langle R^{\widetilde{P}}, C(\lambda)=1\rangle \big] \geq \thres + \frac{1}{p'(\lambda)}.
\]
Set $\nu(\lambda) := 4\ell^2 \, p(\lambda) \, q(\lambda) \, p'(\lambda)$.  
By Conjecture \ref{conj:locc}, there exists a polynomial $n(\cdot)$ such that for all large enough $\lambda \in \mathbb{N}$, for all permutation invariant states $\rho_{\reg{A}_1 \ldots \reg{A}_{n(\lambda)}}$ on $\cH_\reg{A}^{\otimes n(\lambda)}$ (where $\reg{A}=\{\bbC^2\}^{k(\lambda)}$), there exists a measure $\mu$ on density matrices on $\cH_\reg{A}$ such that
\[
\Bigl\| \rho_{\reg{A}_1 \ldots \reg{A}_{q(\lambda)}} - \int \sigma^{\otimes q(\lambda)} \, d\mu(\sigma) \Bigr\|_{\text{LOCC}} \leq \frac{1}{\nu(\lambda)}.
\]
Let $c_{\ell-1}(\lambda)$ be the correctness of the PoQ. For $i \in [\ell-1]$, set
\begin{itemize}
    \item $s_i(\lambda) := \frac{1}{4n(\lambda)^{2i} p(\lambda)}$
    \item $c_{i-1}(\lambda) := s_i(\lambda)$
\end{itemize}
We will drop $\lambda$ when clear from context.
For any set of QPT algorithms $(A_1, \ldots, A_{\ell-1})$, we define a set of candidate weak minischemes $(\Samp_i, \Ver_i^{A_{i+1},\ldots,A_{\ell-1}})_{i \in [\ell-1]}$ as follows. 
\begin{itemize}
    \item $\Samp_i(1^\lambda):$ Run $\langle P,V \rangle(1^\lambda)$ until just after the $2i$-th message to obtain transcript $\tau_i$ and the (purified) internal state of the prover $\ket{\psi_i}$. Note that this state is uniquely determined by $\lambda$ and the transcript.
    \item $\Ver_i^{A_{i+1},\ldots,A_{\ell-1}}(\tau, \widetilde{\rho}):$
    \begin{itemize}
        \item Sample $r\gets \{0,1\}^*$ of same length as the $2i+1$'th message.
        \item $(m, \sigma) \gets P_{i+1}(\tau, r, \widetilde{\rho})$.
        \item  $\tau' \gets \tau \| r \| m$.
        \item If $i+1=\ell$ : 
        \begin{itemize}
            \item Return $V(\tau')$
        \end{itemize}
        \item $\rho_{\reg{A}^{(1)}, \ldots, \reg{A}^{(n)}} \leftarrow A_{i+1}(\tau', \sigma)$
        \item For all $t \in [n]$, run $\Ver_{i+1}^{A_{i+2},\ldots,A_{\ell-1}}(\tau', \rho_{\reg{A}^{(t)}})$.
        Accept if all accept, else reject.
    \end{itemize}
\end{itemize}
\begin{claim}
\label{clm:classical-bb-branch}
    Either:
\begin{enumerate}
    \item There exists $i \in [\ell-1]$ and a set of QPT algorithms $(A_{i+1},\ldots,A_{\ell-1})$ such that 
    \[
      (\Samp_i, \Ver_i^{A_{i+1},\ldots,A_{\ell-1}})
    \]
    is a $(n(\lambda),c_i(\lambda),s_i(\lambda))$-weak uniform io-minischeme.

    \item Or there exists a set of QPT algorithms $(A_1,\ldots,A_{\ell-1})$ such that for all $i \in [\ell-1]$:
    \begin{enumerate}
        \item $(\Samp_i, \Ver_i^{A_{i+1},\ldots,A_{\ell-1}})$ has correctness $c_i(\lambda)$; and  
        \item $A_i$ breaks the $s_i(\lambda)$-security of $(\Samp_i, \Ver_i^{A_{i+1},\ldots,A_{\ell-1}})$, i.e. for all large enough $\lambda$
        \[
    \Pr\left[ \forall t \in [n(\lambda)], 
    \Ver_i^{A_{i+1},\ldots,A_{\ell-1}}(s,\rho_{\reg{A}^{(t)}})=1 \middle| \begin{array}{l}
         (s,\ket{\psi_s}) \gets \Samp_i(1^\lambda)\\
         \rho_{\reg{A}^{(1)} \cdots \reg{A}^{({n(\lambda)})}} \gets A_i(s,\ket{\psi_s})
    \end{array}
     \right] 
    \le s_i(\lambda).
    \]
    \end{enumerate}
\end{enumerate}
\end{claim}
\begin{proof}
Suppose the first condition is false, i.e. for all $i$, for all QPT $(A_{i+1},\ldots,A_{\ell-1})$, it is the case that 
  $(\Samp_i, \Ver_i^{A_{i+1},\ldots,A_{\ell-1}})$
is not a $(n(\lambda),c_i(\lambda),s_i(\lambda))$-weak uniform io-minischeme.  
We will show the second condition by induction on $j = \ell-i$.  

First consider $j=1$, i.e.\ $i=\ell-1$.  
$(\Samp_{\ell-1}, \Ver_{\ell-1})$ has correctness $c_{\ell-1}$, since the correctness is simply the probability that an honest execution of $\langle P,V \rangle$ accepts.  
Now, since the first condition is false and the correctness condition is satisfied, there must exist a QPT machine $A_{\ell-1}$ such that for all large enough $\lambda$
\[
\Pr\left[ \forall t \in [n],  \Ver_{\ell-1}(\tau, \rho_{\reg{A}^{(t)}})=1 \middle| 
\begin{array}{l}
  \tau,\sigma \gets \Samp_{\ell-1}(1^\lambda),\\
  \rho_{\reg{A}^{(1)},\ldots,\reg{A}^{(n)}} \gets A_{\ell-1}(\tau,\sigma)
\end{array}
\right] \geq s_{\ell-1}(\lambda).
\]
Next, we show the induction step.  
For some $j \in [\ell-1]$ and $i = \ell-j$, suppose there exists a set of QPT algorithms $(A_i,\ldots,A_{\ell-1})$ such that for all $t \in [n]$,
\begin{itemize}
  \item $(\Samp_i,\Ver_i^{A_{i+1},\ldots,A_{\ell-1}})$ has correctness $c_i(\lambda)$; and
  \item $A_i$ breaks the $s_i(\lambda)$--security of $(\Samp_i,\Ver_i^{A_{i+1},\ldots,A_{\ell-1}})$.
\end{itemize}
The correctness experiment for $(\Samp_{i-1},\Ver_{i-1}^{A_i,\ldots,A_{\ell-1}})$ is identical to the security experiment for $(\Samp_i,\Ver_i^{A_{i+1},\ldots,A_{\ell-1}})$ against $A_i$.  
Therefore $(\Samp_{i-1},\Ver_{i-1}^{A_i,\ldots,A_{\ell-1}})$ has correctness at least $s_i(\lambda)=c_{i-1}(\lambda)$.  

Now, since the first condition is false and the correctness condition is satisfied, there must exist a QPT machine $A_{i-1}$ such that it breaks the $s_{i-1}(\lambda)$--security of $(\Samp_{i-1},\Ver_{i-1}^{A_i,\ldots,A_{\ell-1}})$, which is enough to satisfy the induction hypothesis for $j+1$.
\end{proof}
If the first condition from Claim \ref{clm:classical-bb-branch} is satisfied, then Theorem \ref{thm:minischeme-amp} shows the existence of uniform i.o.~$n(\lambda)$-weak minischemes. In the remainder of the proof, we will show that the second condition contradicts the security of $(C,\thres)$.\\

\noindent Suppose there exists a set of QPT algorithms $(A_1,\ldots,A_{\ell-1})$ such that for all $i \in [\ell-1]$:
    \begin{enumerate}
        \item $(\Samp_i, \Ver_i^{A_{i+1},\ldots,A_{\ell-1}})$ has correctness $c_i(\lambda)$; and  
        \item $A_i$ breaks the $s_i(\lambda)$-security of $(\Samp_i, \Ver_i^{A_{i+1},\ldots,A_{\ell-1}})$, i.e. for all large enough $\lambda$
        \begin{equation}\label{eq:classical-bb-assumption}
            \Pr\left[ \forall t \in [n(\lambda)], 
            \Ver(s,\rho_{\reg{A}^{(t)}})=1 \middle| \begin{array}{l}
                 (s,\ket{\psi_s}) \gets \Samp_i(1^\lambda)\\
                 \rho_{\reg{A}^{(1)} \cdots \reg{A}^{({n(\lambda)})}} \gets A(s,\ket{\psi_s})
            \end{array}
             \right] 
            \le s_i(\lambda).
    \end{equation}
    \end{enumerate}
Define a malicious prover $P'$ as follows. First define $P'_1,\dots,P'_\ell$.
\begin{itemize}
  \item $P'_1(r_1):$
  \begin{itemize}
    \item $(m_1,\sigma_1) \gets P_1(r_1)$
    \item $(\rho_1)_{\reg{A}^{(1)}_1,\ldots,\reg{A}_1^{(n)}} \gets A_1(r_1 \| m_1,\sigma_1)$
    \item $t \gets [n]$
    \item $\widetilde{\rho}_1 \gets (\rho_1)_{\reg{A}_1^{(t)}}$
    \item Return $(m_1,\widetilde{\rho}_1)$
  \end{itemize}

  \item For $i\in[2,\ell]$, define $P'_i(\tau_{i-1}\|r_i,\widetilde{\rho}_{i-1}):$
  \begin{itemize}
    \item $(m_i,\sigma_i) \gets P_i(\tau_{i-1}\|r_i,\widetilde{\rho}_{i-1})$
    \item $\tau_i := \tau_{i-1} \| r_i \| m_i$
    \item $(\rho_i)_{\reg{A}_i^{(1)},\ldots,\reg{A}_i^{(n)}} \gets A_i(\tau_i,\sigma_i)$
    \item $t \gets [n]$
    \item $\widetilde{\rho}_i \gets (\rho_i)_{\reg{A}_i^{(t)}}$
    \item $\text{Return } (m_i,\widetilde{\rho}_i)$
  \end{itemize}
\end{itemize}

\noindent For the first prover message, the prover $P'$ samples $(m_1,\widetilde{\rho}_1) \gets P_1(r_1)$ and sends $m_1$.  
For the $i$-th prover message, $P'$ samples $(m_i,\widetilde{\rho}_i)$ by running $P'_i(\tau_{i-1}\|r_i,\widetilde{\rho}_{i-1})$ where $\tau_{i-1}\|r_i$ is the transcript so far.

\begin{claim} 
\label{clm:P-prime-soundness}
For all large enough $\lambda$
\[
\Pr[\langle P',V\rangle=1] \ge s_1(\lambda).
\]
\end{claim}

\begin{proof}
Define $\widetilde{\Ver}_i^{A_{i+1},\dots,A_{\ell-1}}(\tau_i,\widetilde{\rho})$ as:
\begin{itemize}
    \item Sample $r\gets \{0,1\}^*$ of same length as the $2i+1$'th message.
    \item $(m, \sigma) \gets P_{i+1}(\tau, r, \widetilde{\rho})$.
    \item  $\tau' \gets \tau \| r \| m$.
    \item If $i+1=\ell$ : 
    \begin{itemize}
        \item Return $V(\tau')$
    \end{itemize}
    \item $\rho_{\reg{A}^{(1)}, \ldots, \reg{A}^{(n)}} \leftarrow A_{i+1}(\tau', \sigma)$
  \item $t \gets [n]$
  \item Return $\widetilde{\Ver}_{i+1}^{A_{i+2},\ldots,A_{\ell-1}}(\tau', \rho_{\reg{A}^{(t)}})$.
\end{itemize}
If $i=\ell-1$, then $\widetilde{\Ver}_i$ and $\Ver_i$ are identical. For all other $i$, for all $\tau,\rho$,
\[
\Pr[\widetilde{\Ver}_i^{A_{i+1},\dots,A_{\ell-1}}(\tau,\rho)=1] \ge 
\Pr[\Ver_i^{A_{i+1},\dots,A_{\ell-1}}(\tau,\rho)=1],
\]
since both verification procedures are identical except $\widetilde{\Ver}_i$ only checks a single random register instead of all registers.
Therefore,
\begin{align*}
\Pr\!&\left[\widetilde{\Ver}_1^{A_2,\dots,A_{\ell-1}}(\tau,\rho_{\reg{A}^{(t)}})=1 \middle| 
\begin{array}{l}
(\tau,\sigma)\gets \Samp_1\\
\rho_{\reg{A}^{(1)}\cdots \reg{A}^{(n)}} \gets A_1(\tau,\sigma)\\
t\gets [n]
\end{array}\right] \\
&\ge 
\Pr\!\left[\Ver_1^{A_2,\dots,A_{\ell-1}}(\tau,\rho_{\reg{A}^{(t)}})=1 \middle| 
\begin{array}{l}
(\tau,\sigma)\gets \Samp_1\\
\rho_{\reg{A}^{(1)}\cdots \reg{A}^{(n)}} \gets A_1(\tau,\sigma)\\
t\gets [n]
\end{array}\right] \\
&\ge 
\Pr\!\left[\forall t, \Ver_1^{A_2,\dots,A_{\ell-1}}(\tau,\rho_{\reg{A}^{(t)}})=1 \middle| 
\begin{array}{l}
(\tau,\sigma)\gets \Samp_1\\
\rho_{\reg{A}^{(1)}\cdots \reg{A}^{(n)}} \gets A_1(\tau,\sigma)
\end{array}\right].
\end{align*}
The final term on the RHS is exactly the security game for weak minischemes, so by \eqref{eq:classical-bb-assumption},
\[
\Pr\!\left[\widetilde{\Ver}_1^{A_2,\dots,A_{\ell-1}}(\tau,\rho_{\reg{A}^{(t)}})=1 \middle| 
\begin{array}{l}
(\tau,\sigma)\gets \Samp_1\\
\rho_{\reg{A}^{(1)}\cdots \reg{A}^{(n)}} \gets A_1(\tau,\sigma)\\
t\gets [n]
\end{array}\right]\ge s_1(\lambda).
\]
Now, the experiment performed above is exactly equivalent to running $\langle P,V \rangle$, which concludes the proof of the claim.
\end{proof}
\noindent Next, we describe an unbounded classical prover $\widetilde{P}$.  
To do so, we define algorithms $(\widetilde{P}_1,\ldots,\widetilde{P}_\ell)$.  
For any quantum state $\rho$, let $\langle \rho \rangle$ refer to the description of the density matrix of $\rho$.  
For a quantum algorithm $B$, let $B^c$ refer to a classical algorithm that simulates $B$, i.e.\ $B^c$ takes a description of a state $\langle \rho \rangle$ as input, simulates sampling $\sigma \gets B(\rho)$, and outputs $\langle \sigma \rangle$. Let $S_n$ be the set of permutations on $n$ objects and for any $\pi \in S_n$ let $P(\pi)$ be the unitary that applies permutation $\pi$ to its input registers.
\begin{itemize}
    \item $\widetilde{P}_1(r_1):$
    \begin{itemize}
        \item $m_1, \langle \sigma_1 \rangle \gets P_1^c(r_1)$
          \item $\langle(\rho_1)_{\reg{A}_1^{(1)} \cdots \reg{A}_1^{(n)}}\rangle \gets A_1(r_1 \| m_1, \langle \sigma_1 \rangle)$
          \item Define $\rho_1^{\mathsf{sym}} := \Exp_{\pi \gets S_n} [P(\pi)\,\rho_1\,P(\pi)^\dagger]$
          \item Since $\rho_1^{\mathsf{sym}}$ is permutation-invariant, by Conjecture \ref{conj:locc}. there exists measure $\mu_1$ s.t. 
          \[
            \left\| (\rho_1^{\mathsf{sym}})_{\reg{A}_1^{(1)} \cdots \reg{A}_1^{(q)}} - \int \varphi^{\otimes q} \, d\mu_1(\varphi) \right\|_{\mathrm{LOCC}} \le \tfrac{1}{\nu(\lambda)}.
         \]
          Sample $\varphi_1 \gets \mu_1$
          \item Return $m_1, \langle \varphi_1 \rangle$
    \end{itemize}
    \item For $i \in [2,\ell]$, $\widetilde{P}_i(\tau_{i-1} \| r_i, \langle \varphi_{i-1} \rangle):$
    \begin{itemize}
        \item $m_i, \langle \sigma_i \rangle \gets P_i^c(\tau_{i-1} \| r_i, \langle \varphi_{i-1} \rangle)$.  
        If $i = \ell$, return $m_i$.
        \item $\tau_i := \tau_{i-1} \| r_i \| m_i$
        \item $\langle (\rho_i)_{\reg{A}_i^{(1)} \cdots \reg{A}_i^{(n)}} \rangle \gets A_i(\tau_i, \langle \sigma_i \rangle)$
        \item Define $\rho_i^{\mathsf{sym}} := \Exp_{\pi \gets S_n} [P(\pi)\,\rho_i\,P(\pi)^\dagger]$
        \item Since $\rho_i^{\mathsf{sym}}$ is permutation-invariant, by Conjecture~\ref{conj:locc} there exists measure $\mu_i$ s.t.
        \[
            \left\| (\rho_i^{\mathsf{sym}})_{\reg{A}_i^{(1)} \cdots \reg{A}_i^{(q)}} - \int \varphi^{\otimes q} \, d\mu_i(\varphi) \right\|_{\mathrm{LOCC}} \le \tfrac{1}{\nu(\lambda)}.
        \]
        Sample $\varphi_i \gets \mu_i$
        \item Return $m_i, \langle \varphi_i \rangle$
    \end{itemize}
\end{itemize}
For the first prover message, the prover $\widetilde{P}$ samples $m_1, \langle \varphi_1 \rangle \gets \widetilde{P}_1(r_1)$ and sends $m_1$.  
For the $i$th prover message, $\widetilde{P}$ samples $m_i, \langle \varphi_i \rangle$ by running  
$\widetilde{P}_i(\tau_{i-1} \| r_i, \langle \varphi_{i-1} \rangle)$, where $\tau_{i-1} \| r_i$ is the transcript so far.  
While our description of the prover is randomized, we will think of the $\widetilde{P}_i$’s as sampling and fixing the randomness for all possible inputs in the beginning of execution.

Next we describe an efficient simulator $\Sim$ that, given access to a $q$-query classical reduction $\widetilde{R}$, and for all challengers $\widetilde{C}$, will mimic the interaction of $\widetilde{R}^{\widetilde{P}}$ with $\widetilde{C}$.  
Without loss of generality, we may assume that $\widetilde{R}$ only runs $\widetilde{P}$ on $(\tau \| r)$ for partial transcripts $\tau$ that were generated by $\widetilde{P}$, and since $\widetilde{P}$ is deterministic after fixing randomness, we may also assume $\widetilde{R}$ never runs $\widetilde{P}$ twice on the same partial transcript.  
  
$\Sim$ maintains an internal database keyed by classical strings, where each entry stores a multipartite quantum state, together with bookkeeping marking subsystems as used or unused.  
When a new state is added to the database, all its subsystems are marked as unused.
We think of an $\widetilde{R}$ query requesting an $i$th prover message as a query to $\widetilde{P}_i$.  $\Sim$ is then defined as the algorithm that interacts with $\widetilde{C}$ by running $\widetilde{R}$ and answering $\widetilde{R}$ queries as follows.

\begin{itemize}
    \item On query $r$ to $\widetilde{P}_1$:
    \begin{itemize}
        \item $m, \sigma \gets P_1(r)$
        \item $\rho \gets A_1(r \| m, \sigma)$
        \item Apply a random permutation to $\rho$ to get $\rho^{\mathsf{sym}}$
        \item Store $\rho^{\mathsf{sym}}$ in the database keyed by $r \| m$
        \item Return $m$
    \end{itemize}

    \item On query $\tau \| r$ to $\widetilde{P}_i$:
    \begin{itemize}
        \item Look up the state stored under key $\tau$ in the database.  
        If no such state exists or if all subsystems of the state are marked used, then abort.  
        Else, retrieve the marginal state $\widetilde{\rho}$ on the next unused subsystem of the stored database state.  
        Mark this subsystem as used.
        \item $m, \sigma \gets P_i(\tau \| r, \widetilde{\rho})$. If $i = \ell$, return $m$.
        \item $\rho \gets A_i(\tau \| m, \sigma)$
        \item Apply a random permutation to $\rho$ to get $\rho^{\mathsf{sym}}$
        \item Store $\rho^{\mathsf{sym}}$ in the database keyed by $\tau \| r \| m$
        \item Return $m$
    \end{itemize}
\end{itemize}
\begin{claim}
    \label{clm:classical-bb-sim}
    For any $q'(\lambda) \leq q(\lambda)$, for all $q'(\lambda)$-query reductions $\widetilde{R}$ and all challengers $\widetilde{C}$,
\[
    \Big| \Pr\big[ \langle \mathsf{Sim}^{\widetilde{R}}, \widetilde{C} \rangle = 1 \big]
      - \Pr\big[ \langle \widetilde{R}^{\widetilde{P}}, \widetilde{C} \rangle = 1 \big] \Big|
    \le \frac{(\ell-1)q'(\lambda)}{\nu(\lambda)}.
\]
\end{claim}

\begin{proof}
We prove by constructing a sequence of hybrid simulators.  

For $i \in [0,q']$, define $\mathsf{Sim}_{1,i}$ as identical to $\mathsf{Sim}$ except for the first $i$ queries to $\widetilde{P}_1$, where it does the following instead:
\begin{itemize}
    \item On query $r$, sample $m, \langle \varphi \rangle \gets \widetilde{P}_1(r)$
    \item Prepare the state $\varphi^{\otimes n}$ and store it in the database under key $r \| m$
    \item Return $m$
\end{itemize}
\begin{subclaim}
\label{subclm:classical-sim-hybrid-1}
    For $i \in [q']$,
\[
    \Big| \Pr\big[ \langle \mathsf{Sim}_{1,i}^{\widetilde{R}}, \widetilde{C} \rangle(1^\lambda) = 1 \big]
      - \Pr\big[ \langle \mathsf{Sim}_{1,i-1}^{\widetilde{R}}, \widetilde{C} \rangle(1^\lambda) = 1 \big] \Big|
    \le \tfrac{1}{\nu(\lambda)}.
\]
\end{subclaim}

\begin{proof}
The only difference between the simulators is the state stored in the database for the $i$th query to $\widetilde{P}_1$ by $\widetilde{R}$.  
Each subsystem of this stored state may then later be used to generate further messages and states to be stored in the database.  
However, no entangled quantum operations are applied to subsystems; i.e., all measurements are local, and only classical information about the subsystems is sent to $\widetilde{R}$.  
Therefore $\widetilde{R}$ and $\widetilde{C}$ together with the action of the simulator constitute an LOCC measurement on the stored state.  
However, by Conjecture \ref{conj:locc}, the states stored in the case of either simulator cannot be distinguished with probability better than $1/\nu(\lambda)$ by LOCC measurements on $q$ subsystems.  
Since only the first $q' \leq q$ subsystems can be measured in $q'$ queries, this gives the required bound.
\end{proof}

\noindent For $t \in [2,\ell-1]$, define $\mathsf{Sim}_{t,0}$ as identical to $\mathsf{Sim}$ except:
\begin{itemize}
    \item On queries of form $r$ to $\widetilde{P}_1$:
    \begin{itemize}
        \item Sample $m, \langle \varphi \rangle \gets \widetilde{P}_1(r)$
        \item Store description $\langle \varphi \rangle$ in the database under key $r \| m$
        \item Return $m$
    \end{itemize}

    \item On queries of form $\tau \| r$ to $\widetilde{P}_j$ where $2\leq j \le t-1$:
    \begin{itemize}
        \item Look up description $\langle \varphi' \rangle$ in the database keyed by $\tau$
        \item Sample $m, \langle \varphi \rangle \gets \widetilde{P}_j(\tau \| r, \langle \varphi' \rangle)$
        \item Store description $\langle \varphi \rangle$ in the database under key $\tau \| r \| m$
        \item Return $m$
    \end{itemize}

    \item On queries of form $\tau \| r$ to $\widetilde{P}_t$:
    \begin{itemize}
        \item Answer query as in $\mathsf{Sim}$ except when looking up the state in the database,  
        use the stored description $\langle \varphi \rangle$ to generate $\varphi$ and use $\varphi$ as the marginal state on the next unused subsystem.
    \end{itemize}
\end{itemize}
For $t \in [2,\ell-1]$ and $i \in [0,q']$, define $\mathsf{Sim}_{t,i}$ as identical to $\mathsf{Sim}_{t,0}$  
except that for the first $i$ queries to $\widetilde{P}_t$, it does the following instead:
\begin{itemize}
    \item On query $\tau \| r$, look up description $\langle \varphi' \rangle$ in the database keyed by $\tau$
    \item Sample $m, \langle \varphi \rangle \gets \widetilde{P}_t(\tau \| r, \langle \varphi' \rangle)$
    \item Prepare the state $\varphi^{\otimes n}$ and store in the database under key $\tau \| r \| m$
    \item Return $m$
\end{itemize}

\begin{subclaim}
\label{subclm:classical-sim-hybrid-2}
For $t \in [2,\ell-1]$ and $i \in [q']$,
\[
    \Big| \Pr\big[ \langle \mathsf{Sim}_{t,i}^{\widetilde{R}}, \widetilde{C} \rangle(1^\lambda) = 1 \big]
      - \Pr\big[ \langle \mathsf{Sim}_{t,i-1}^{\widetilde{R}}, \widetilde{C} \rangle(1^\lambda) = 1 \big] \Big|
    \le \tfrac{1}{\nu(\lambda)}.
\]
\end{subclaim}

\begin{proof}
Almost identical to the proof of SubClaim \ref{subclm:classical-sim-hybrid-1}.
\end{proof}

\begin{subclaim}
\label{subclm:classical-sim-hybrid-3}
For $t \in [2,\ell-1]$,
\[
    \Pr\big[ \langle \mathsf{Sim}_{t,0}^{\widetilde{R}}, \widetilde{C} \rangle(1^\lambda) = 1 \big]
    = \Pr\big[ \langle \mathsf{Sim}_{t-1,q'}^{\widetilde{R}}, \widetilde{C} \rangle(1^\lambda) = 1 \big].
\]
\end{subclaim}

\begin{proof}
The only difference between the simulators is whether the database entry created during a $\widetilde{P}_t$ query is stored in description form and the corresponding state is generated during lookup, or whether the database entry contains $n$ copies of the state. Both cases are identical in the view of $\widetilde{R}$ and $\widetilde{C}$. 
\end{proof}
Finally, we observe that $\widetilde{R}$’s view in $\mathsf{Sim}_{\ell-1,q'}$ is identical to its view in $\widetilde{R}^{\widetilde{P}}$.  
Putting all SubClaims \ref{subclm:classical-sim-hybrid-1}, \ref{subclm:classical-sim-hybrid-2}, and \ref{subclm:classical-sim-hybrid-3} together with this fact yields the statement in the claim.
\end{proof}

\begin{claim}
\label{clm:c-and-q-prover-soundness}
\[
    \Big| \Pr\big[ \langle \widetilde{P}, V \rangle(1^\lambda) = 1 \big]
      - \Pr\big[ \langle P', V \rangle(1^\lambda) = 1 \big] \Big|
    \le \tfrac{\ell(\ell-1)}{\nu(\lambda)}
\]
\end{claim}

\begin{proof}
This follows from the observation that the Claim \ref{clm:classical-bb-sim} applies to all reductions and challengers.  
We may view the interaction between the prover and verifier as a game where the challenger is the verifier and the reduction merely forwards prover and verifier messages to each other.  
Running the simulator with this reduction is identical to running $\langle P', V \rangle$, while running the reduction with $\widetilde{P}$ is identical to running $\langle \widetilde{P}, V \rangle$.  
Since the reduction makes at most $\ell$ queries, the statement of the claim follows.
\end{proof}
\noindent Now, recall that by Claim \ref{clm:P-prime-soundness},
\[
    \Pr\big[ \langle P', V \rangle(1^\lambda) = 1 \big] \ge s_1(\lambda) = 1 - \frac{1}{4n(\lambda)^2 p(\lambda)}.
\]
which implies by Claim \ref{clm:c-and-q-prover-soundness} that
\begin{align*}
    \Pr\big[ \langle \widetilde{P}, V \rangle(\lambda) = 1 \big] 
    &\ge 1 - \frac{1}{4n(\lambda)^2 p(\lambda)} - \frac{\ell(\ell-1)}{\nu(\lambda)}\\
    &\geq 1 - \frac{1}{4n(\lambda)^2 p(\lambda)} - \frac{1}{4p(\lambda)q(\lambda)p'(\lambda)}\\
    &\ge 1 - \frac{1}{2p(\lambda)}.
\end{align*}
By the definition of $p'$ at the beginning of the proof, implies that
\[
    \Pr\big[ \langle R^{\widetilde{P}}, C \rangle = 1 \big] \ge \mathsf{thres} + \frac{1}{p'(\lambda)}.
\]
By Claim \ref{clm:classical-bb-sim}, this means
\begin{align*}
    \Pr\big[ \langle \mathsf{Sim}^R, C \rangle = 1 \big] 
    &\ge \mathsf{thres} + \frac{1}{p'(\lambda)} - \frac{(\ell-1)q(\lambda)}{\nu(\lambda)}\\
    &\geq \mathsf{thres} + \frac{1}{p'(\lambda)} - \frac{1}{4 \ell \, p(\lambda) \cdot p'(\lambda)}\\
    &\ge \mathsf{thres} + \frac{1}{2p'(\lambda)}.
\end{align*}
which contradicts the security of $C$, completing the proof of the theorem.
\end{proof}
\begin{theorem}
The existence of a 4-round PoQ with $1-\negl(\lambda)$ correctness and a classical black-box reduction to a falsifiable assumption secure against non-uniform quantum adversaries implies the existence of $\poly(\lambda)$-weak lightning.
\begin{remark}
    Note that we obtain a primitive with almost-everywhere (as opposed to infinitely-often) security against non-uniform quantum adversaries.
\end{remark}
\end{theorem}
\begin{proof}

Let the PoQ be $\langle P,V \rangle$ with reduction $R$ to assumption $(C,\thres)$, with correctness $1-\negl(\lambda)$ and soundness $1-\frac{1}{p(\lambda)}$ for some polynomial $p$. 
Let $k(\lambda)$ be an upper bound on the length of the (purified) state of the honest prover during execution.

By definition of PoQ, there exists a polynomial $p'$ such that for any classical prover $\widetilde{P}$ where for infinitely many $\lambda$, $\Pr[\langle \widetilde{P},V \rangle(1^\lambda) = 1] \geq 1 - \frac{1}{2p(\lambda)}$, it must be the case that for infinitely many $\lambda$,
\[
\Pr\big[\langle R^{\widetilde{P}}, C(\lambda)=1\rangle \big] \geq \thres + \frac{1}{p'(\lambda)}.
\]
Set $\nu(\lambda) := 16 \, p(\lambda) \, q(\lambda) \, p'(\lambda)$.  
By Theorem \ref{thm:1locc} \ref{conj:locc}, there exists a polynomial $n(\cdot)$ such that for all large enough $\lambda \in \mathbb{N}$, for all permutation invariant states $\rho_{\reg{A}_1 \ldots \reg{A}_{n(\lambda)}}$ on $\cH_\reg{A}^{\otimes n(\lambda)}$ (where $\reg{A}=\{\bbC^2\}^{k(\lambda)}$), there exists a measure $\mu$ on density matrices on $\cH_\reg{A}$ such that
\[
\Bigl\| \rho_{\reg{A}_1 \ldots \reg{A}_{q(\lambda)}} - \int \sigma^{\otimes q(\lambda)} \, d\mu(\sigma) \Bigr\|_{\text{LOCC}_1} \leq \frac{1}{\nu(\lambda)}.
\]
Let $c(\lambda)$ be the correctness of the PoQ. Set $s(\lambda) := \frac{1}{4n(\lambda)^{2} p(\lambda)}$
We will drop $\lambda$ when clear from context.
Define
\begin{itemize}
    \item $\mathsf{Setup}(1^\lambda):$
    \begin{itemize}
        \item Output the first verifier message as $\pp$.
    \end{itemize}
    \item $\mathsf{Samp}(\pp):$
    \begin{itemize}
        \item Run the honest prover on the first verifier message to obtain $\srno$ and (purified) prover state $\ket{\psi_\srno}$.
    \end{itemize}

    \item $\mathsf{Ver}(\pp,\srno, \rho):$  
    \begin{itemize}
        \item Run the last two rounds of the honest protocol using $\rho$ as the honest prover state and $\pp\|\srno$ as the transcript so far. Accept if the verifier accepts.
    \end{itemize}
\end{itemize}
Suppose $(\Setup, \Samp, \Ver)$ is not $(n(\lambda), 1-\negl(\lambda), s(\lambda))$-weak lightning. The scheme  satisfies correctness by definition. Therefore there exists a  QPT adversaries $\mathcal{A}$ with advice $\{\ket{\phi_\lambda}\}_\lambda$ that breaks security for infinitely many $\lambda$. We will drop the advice state from the notation and assume it is always passed to the adversary whenever it is called.

Define a malicious prover $P'$ as follows. First we define $P'_1$ and $P'_2$
\begin{itemize}
  \item $P'_1(r_1):$
  \begin{itemize}
    \item $m_1, (\rho_1)_{\reg{A}^{(1)}_1,\ldots,\reg{A}_1^{(n)}} \gets \cA(r_1)$
    \item $t \gets [n]$
    \item $\widetilde{\rho}_1 \gets (\rho_1)_{\reg{A}_1^{(t)}}$
    \item Return $(m_1,\widetilde{\rho}_1)$
  \end{itemize}

  \item  $P'_2(\tau, \widetilde{\rho}_{1}):$
  \begin{itemize}
  \item Run the honest prover to compute the second prover message with transcript $\tau$ and prover state $\widetilde{\rho}_{1}$. Return the computed message.
  \end{itemize}
\end{itemize}

\noindent For the first prover message, the prover $P'$ samples $(m_1,\widetilde{\rho}_1) \gets P_1(r_1)$ and sends $m_1$.  
For the second prover message, $P'$ sends the output of $P'_2(\tau, \widetilde{\rho}_1)$ where $\tau$ is the transcript so far.

\begin{claim} 
\label{clm:P-prime-soundness-4-round}
For infinitely many $\lambda$
\[
\Pr[\langle P',V\rangle=1] \ge s(\lambda).
\]
\end{claim}

\begin{proof}
Similar to Claim \ref{clm:P-prime-soundness}, this follows from the fact that $\cA$ breaks the security of the defined lightning scheme.
\end{proof}
\noindent Next, we describe an unbounded classical prover $\widetilde{P}$.  
To do so, we define algorithms $\widetilde{P}_1$ and $\widetilde{P}_2)$.  
For any quantum state $\rho$, let $\langle \rho \rangle$ refer to the description of the density matrix of $\rho$.  
For a quantum algorithm $B$, let $B^c$ refer to a classical algorithm that simulates $B$, i.e.\ $B^c$ takes a description of a state $\langle \rho \rangle$ as input, simulates sampling $\sigma \gets B(\rho)$, and outputs $\langle \sigma \rangle$. Let $S_n$ be the set of permutations on $n$ objects and for any $\pi \in S_n$ let $P(\pi)$ be the unitary that applies permutation $\pi$ to its input registers.
\begin{itemize}
    \item $\widetilde{P}_1(r_1):$
    \begin{itemize}
          \item $m_1, \langle(\rho_1)_{\reg{A}_1^{(1)} \cdots \reg{A}_1^{(n)}}\rangle \gets \cA(r_1)$
          \item Define $\rho_1^{\mathsf{sym}} := \Exp_{\pi \gets S_n} [P(\pi)\,\rho_1\,P(\pi)^\dagger]$
          \item Since $\rho_1^{\mathsf{sym}}$ is permutation-invariant, by Theorem \ref{thm:1locc}. there exists measure $\mu_1$ s.t. 
          \[
            \left\| (\rho_1^{\mathsf{sym}})_{\reg{A}_1^{(1)} \cdots \reg{A}_1^{(q)}} - \int \varphi^{\otimes q} \, d\mu_1(\varphi) \right\|_{\mathrm{LOCC}_1} \le \tfrac{1}{\nu(\lambda)}.
         \]
          Sample $\varphi_1 \gets \mu_1$
          \item Return $m_1, \langle \varphi_1 \rangle$
    \end{itemize}
    \item  $\widetilde{P}_2(\tau, \langle \varphi_1 \rangle):$
    \begin{itemize}
        \item Run honest prover using transcript $\tau$ and internal state $\varphi_1$ recovered from the description. Return the computed second prover message.
    \end{itemize}
\end{itemize}
For the first prover message, the prover $\widetilde{P}$ samples $m_1, \langle \varphi_1 \rangle \gets \widetilde{P}_1(r_1)$ and sends $m_1$.  
For the second prover message, $\widetilde{P}$ sends the output of 
$\widetilde{P}_2(\tau, \langle \varphi_{1} \rangle)$, where $\tau$ is the transcript so far.  
While our description of the prover is randomized, we will think of the $\widetilde{P}_i$’s as sampling and fixing the randomness for all possible inputs in the beginning of execution.

Next we describe an efficient simulator $\Sim$ (with $q$ copies of non-uniform advice in order to run $\cA)$ that, given access to a $q$-query classical reduction $\widetilde{R}$, and for all challengers $\widetilde{C}$, will mimic the interaction of $\widetilde{R}^{\widetilde{P}}$ with $\widetilde{C}$.  
Without loss of generality, we may assume that $\widetilde{R}$ only runs $\widetilde{P}$ on $(\tau \| r)$ for partial transcripts $\tau$ that were generated by $\widetilde{P}$, and since $\widetilde{P}$ is deterministic after fixing randomness, we may also assume $\widetilde{R}$ never runs $\widetilde{P}$ twice on the same partial transcript.  
  
$\Sim$ maintains an internal database keyed by classical strings, where each entry stores a multipartite quantum state, together with bookkeeping marking subsystems as used or unused.  
When a new state is added to the database, all its subsystems are marked as unused.
We think of an $\widetilde{R}$ query requesting an $i$th prover message as a query to $\widetilde{P}_i$.  $\Sim$ is then defined as the algorithm that interacts with $\widetilde{C}$ by running $\widetilde{R}$ and answering $\widetilde{R}$ queries as follows.

\begin{itemize}
    \item On query $r$ to $\widetilde{P}_1$:
    \begin{itemize}
        \item $m, \rho \gets \cA(r)$
        \item Apply a random permutation to $\rho$ to get $\rho^{\mathsf{sym}}$
        \item Store $\rho^{\mathsf{sym}}$ in the database keyed by $r \| m$
        \item Return $m$
    \end{itemize}

    \item On query $\tau \| r$ to $\widetilde{P}_2$:
    \begin{itemize}
        \item Look up the state stored under key $\tau$ in the database.  
        If no such state exists or if all subsystems of the state are marked used, then abort.  
        Else, retrieve the marginal state $\widetilde{\rho}$ on the next unused subsystem of the stored database state.  
        Mark this subsystem as used.
        \item Run the honest prover with transcript $\tau$ and internal state $\widetilde{\rho}$ and return the computed second prover message.
    \end{itemize}
\end{itemize}
\begin{claim}
    \label{clm:classical-bb-sim-4-round}
    For any $q'(\lambda) \leq q(\lambda)$, for all $q'(\lambda)$-query reductions $\widetilde{R}$ and all challengers $\widetilde{C}$,
\[
    \Big| \Pr\big[ \langle \mathsf{Sim}^{\widetilde{R}}, \widetilde{C} \rangle = 1 \big]
      - \Pr\big[ \langle \widetilde{R}^{\widetilde{P}}, \widetilde{C} \rangle = 1 \big] \Big|
    \le \frac{q'(\lambda)}{\nu(\lambda)}.
\]
\end{claim}

\begin{proof}
By almost identical arguments to Claim \ref{clm:classical-bb-sim}. The only major difference is we only have bounds for one-way LOCC measurements, however this is not an issue since the measurement applied in our case is in fact one-way LOCC since there are no further measurements applied once $\widetilde{P}_2$ is run on a subsystem.
\end{proof}

\begin{claim}
\label{clm:c-and-q-prover-soundness-4-round} For all large enough $\lambda$
\[
    \Big| \Pr\big[ \langle \widetilde{P}, V \rangle(1^\lambda) = 1 \big]
      - \Pr\big[ \langle P', V \rangle(1^\lambda) = 1 \big] \Big|
    \le \tfrac{2}{\nu(\lambda)}
\]
\end{claim}

\begin{proof}
This follows from the observation that the Claim \ref{clm:classical-bb-sim-4-round} applies to all reductions and challengers.  
We may view the interaction between the prover and verifier as a game where the challenger is the verifier and the reduction merely forwards prover and verifier messages to each other.  
Running the simulator with this reduction is identical to running $\langle P', V \rangle$, while running the reduction with $\widetilde{P}$ is identical to running $\langle \widetilde{P}, V \rangle$.  
Since the reduction makes at most $2$ queries, the statement of the claim follows.
\end{proof}
\noindent Now, recall that by Claim \ref{clm:P-prime-soundness-4-round}, for infinitely many $\lambda$
\[
    \Pr\big[ \langle P', V \rangle(1^\lambda) = 1 \big] \ge s(\lambda) = 1 - \frac{1}{4n(\lambda)^2 p(\lambda)}.
\]
which implies by Claim \ref{clm:c-and-q-prover-soundness-4-round} that  for infinitely many $\lambda$
\begin{align*}
    \Pr\big[ \langle \widetilde{P}, V \rangle(\lambda) = 1 \big] 
    &\ge 1 - \frac{1}{4n(\lambda)^2 p(\lambda)} - \frac{2}{\nu(\lambda)}\\
    &\geq 1 - \frac{1}{4n(\lambda)^2 p(\lambda)} - \frac{1}{4p(\lambda)q(\lambda)p'(\lambda)}\\
    &\ge 1 - \frac{1}{2p(\lambda)}.
\end{align*}
By the definition of $p'$ at the beginning of the proof, implies that  for infinitely many $\lambda$
\[
    \Pr\big[ \langle R^{\widetilde{P}}, C \rangle = 1 \big] \ge \mathsf{thres} + \frac{1}{p'(\lambda)}.
\]
By Claim \ref{clm:classical-bb-sim-4-round}, this means  for infinitely many $\lambda$
\begin{align*}
    \Pr\big[ \langle \mathsf{Sim}^R, C \rangle = 1 \big] 
    &\ge \mathsf{thres} + \frac{1}{p'(\lambda)} - \frac{q(\lambda)}{\nu(\lambda)}\\
    &\geq \mathsf{thres} + \frac{1}{p'(\lambda)} - \frac{1}{4 p(\lambda) \cdot p'(\lambda)}\\
    &\ge \mathsf{thres} + \frac{1}{2p'(\lambda)}.
\end{align*}
which contradicts the security of $C$, completing the proof of the theorem.
\end{proof}